%% file: main.tex
\documentclass[11pt, onecolumn]{IEEEtran}

\usepackage{amsmath}	
\usepackage{amsfonts}
\usepackage{amssymb}
\usepackage{tabularx}{\tiny }
\usepackage{upref}
\usepackage{theorem}

\usepackage{graphics}
\usepackage[pdf_tex]{graphicx} 

\usepackage{psfrag}
\usepackage{enumerate}
\usepackage{caption}
\usepackage{subcaption}
\usepackage{color}
\usepackage{booktabs}
\usepackage{multirow}
\usepackage{multicol}
\usepackage{diagbox}
\usepackage{pict2e}
\usepackage{longtable}
\usepackage{bigdelim}
\usepackage[font=small,labelfont=bf,tableposition=top]{caption}

\usepackage{cite}

\newtheorem{theorem}{Theorem}
\newtheorem{lemma}{Lemma}
\newtheorem{proposition}{Proposition}

\newtheorem{definition}{Definition}
\newtheorem{remark}{Remark}
\newtheorem{example}{Example}

\definecolor{CommentBlue}{RGB}{0,80,239}
\definecolor{CommentPurple}{RGB}{145, 32, 186}
\definecolor{CommentRed}{RGB}{168,24,24}
\definecolor{CommentGreen}{RGB}{10,160,10}
\definecolor{CommentOrange}{RGB}{204,102,0}

\newcommand{\comment}[1]{}

\def\cA{\mbox{$\cal{A}$}}
\def\cB{\mbox{$\cal{B}$}}

\def\cX{\mbox{$\cal{X}$}}

\def\cV{\mbox{$\cal{V}$}}
\def\cR{\mbox{$\cal{R}$}}

\def\cU{\mbox{$\cal{U}$}}
\def\cV{\mbox{$\cal{V}$}}
\def\cG{\mbox{$\cal{G}$}}

\def\cL{\mbox{$\cal{L}$}}

\def\cS{\mbox{$\cal{S}$}}

\def\cD{\mbox{$\cal{D}$}}
\def\cJ{\mbox{$\cal{J}$}}

\def\cK{\mbox{$\cal{K}$}}

\def\cP{\mbox{$\cal{P}$}}

\def\cT{\mbox{$\cal{T}$}}

\def\cC{\mbox{$\cal{C}$}}
\def\DRS{\mbox{$\cal{D}_{RS}$}}
\def\bd{\mbox{$\bar{d}$}}

\newcommand{\TRS}{X{}} 

\newcommand{\ND}[1]{{\bar{D}_{\tilde{#1}}}}

\newcommand{\Rto}{R^{*p}_{N,K}(M)} 
\newcommand{\Rm}{R^{*}_{N,K}(M)} 
\newcommand{\Rtc}{R^{A}_{N,K}(M)} 

\newcommand{\RMaNx}{R^{\text{MAN}}_{N, NK}(M)}

\newcommand{\RMaN}{R^{\text{MAN}}_{N, K}(M)}

\newif\ifncc
\newif\ifarxiv
\arxivtrue 

\title{Fundamental Limits of\\ Demand-Private Coded Caching }

	\author{
	\IEEEauthorblockN{Chinmay~Gurjarpadhye, Jithin~Ravi, Sneha~Kamath,  Bikash~Kumar~Dey, and Nikhil~Karamchandani}
	\thanks{J.~Ravi has received funding from the European Research Council (ERC) under the European Union's Horizon 2020 research and innovation programme (Grant No.~714161). The work of B. K. Dey was supported in part by the Bharti Centre for Communication in IIT Bombay. The work of N.~Karamchandani is supported in part by a Science and Engineering Research Board (SERB) grant on ``Content Caching and Delivery over Wireless Networks". The material in this paper was presented in part at the IEEE National Conference on Communications, Kharagpur, India, February 2020, and will be presented in part at the IEEE Information Theory Workshop, Riva del Garda, Italy, April 2021.
		
		C.~Gurjarpadhye, B.~K.~Dey and N.~Karamchandani are with Department of Electrical Engineering, IIT Bombay, Mumbai, India.
		J. Ravi is with the Signal Theory and Communications Department, Universidad Carlos III de Madrid, Spain, and with the Gregorio Mara\~n\'on  Health Research Institute, Madrid, Spain. S.~Kamath is with Qualcomm, India
		(emails: cgurjarpadhye@gmail.com, rjithin@tsc.uc3m.es, snehkama@qti.qualcomm.com, bikash@ee.iitb.ac.in, nikhilk@ee.iitb.ac.in). Part of the work of J.~Ravi and S.~Kamath was done when they were at IIT Bombay.}
}
		
\begin{document}
  	\ifarxiv
  \IEEEoverridecommandlockouts
 \maketitle
 \fi
 
 \ifncc
 \maketitle
 \fi

 \begin{abstract} 
We consider the coded caching problem with an additional privacy constraint that a user should not get any information about the demands of the other users. We first show that a demand-private scheme for $N$ files and $K$ users
can be obtained from a non-private scheme that serves only a subset of the demands for the $N$ files and $NK$ users problem. We further use this fact to construct a demand-private scheme for $N$ files and $K$ users  from a particular known non-private scheme for $N$ files and $NK-K+1$ users. It is then demonstrated that, the memory-rate pair $(M,\min \{N,K\}(1-M/N))$, which is achievable for non-private schemes with uncoded transmissions, is also achievable
under demand privacy. We further propose a scheme that improves on these ideas by removing some redundant transmissions. The memory-rate trade-off achieved using our schemes is shown to be within a multiplicative factor of 3 from the optimal when $K < N$ and of 8 when $N \leq K$.  Finally, we give the  exact memory-rate trade-off for  demand-private coded caching problems with $N\geq K=2$.
 \end{abstract}

 \section{Introduction}
 
\input{intro}

 \section{Problem formulation and definitions}
 \label{sec_problem}
 \input{model.tex}

 \section{Results}
\label{sec_results}
\input{results1.tex}

 \section{Proofs}
\label{sec_proofs}
\input{proofs.tex}

 \appendices
 \input{append.tex}
 \bibliographystyle{IEEEtran}
 \bibliography{Bibliography.bib}

\end{document}

%% file: intro.tex
In their seminal work~\cite{Maddah14,maddah2014decentralized}, Maddah-Ali and Niesen analyzed the
fundamental limits of caching in a noiseless broadcast network from an
information-theoretic perspective. A server has $N$ files of equal size. There are $K$ users,
each equipped with a cache that can store $M$ files. In the \emph{placement
phase}, the cache of each user is populated with some functions  of the files.
In the \emph{delivery phase}, each user requests one of the $N$ files, and the
server broadcasts a message to serve the demands of the users. The goal of the
\emph{coded caching} problem is to identify the minimum required \emph{rate} of  transmission from the
server for any given cache  size $M$. For this setup, \cite{Maddah14} proposed an achievability scheme and by comparing the achievable rate with an information-theoretic lower bound on the optimal rate, demonstrated the
scheme to be \emph{order optimal}, i.e.,
the achievable rate is within a constant multiplicative factor from the
optimal for all system parameters $N, K, M$. The works~\cite{Amiri17, Vilardebo18, Yu18}
mainly focused on obtaining improved achievable rates while the
works~\cite{Ghasemi17,Wang18} focused on  improving the lower bounds. Different
aspects of the coded caching problem such as
\emph{subpacketization}~\cite{Yan17, Tang18, Suthan19}, \emph{non-uniform
demands}~\cite{Niesen17, JiTLC17, Zhang18} and \emph{asychnronous demands}~\cite{Ghasemi20, Yang19}  have
been investigated in the past. Fundamental limits of caching has also been studied for some
other network models, see for example \cite{Shanmugam13,Karamchandani16, JiCM16}. We refer the reader to~\cite{maddah2016coding} for a detailed survey.

The schemes proposed in~\cite{Maddah14, maddah2014decentralized} for the coded caching problem exploited the broadcast property of the network to reduce the rate of transmission. The fact that this can lead to a coding gain has also been explored in the related \emph{index coding} framework~\cite{YossefBJK11}, where the users may have a subset of files as side information and request one file from the server that they do not have access to. While the broadcast property helps in achieving a coding gain under such settings, it affects the security and privacy of users. Two types of security/privacy issues have been studied in index coding. The works~\cite{NarayananRMDKP18, DauSC12} addressed the problem of  \emph{file privacy} where the constraint is that each user should not get
any information about any file other than the requested one.
The work~\cite{Karmoose20}  studied index coding with \emph{demand privacy} where each user should not get
any information about  the identity of the file requested by other users.  Demand privacy is also studied in a different context called \emph{private information retrieval} where a user downloads her file of interest from one or more  servers and does not want to reveal the identity of the requested file to any server, see~\cite{SunJ17}  for example.

File privacy for the coded caching problem was investigated in~\cite{Sengupta15, Ravindrakumar18}. In particular, \cite{Sengupta15} considered the privacy of files
against an eavesdropper who has access to the broadcast link,
 while~\cite{Ravindrakumar18} studied the caching problem with the constraint that each user should not get
any information about any file other than the requested one. In \cite{Ravindrakumar18}, a 
private scheme was proposed using  techniques from secret
sharing, and the achievable rate was shown to be order optimal.

In this paper, we consider the coded caching problem with
an extra constraint that each user should not learn any information
about the demands of other users.
 Coded caching under demand privacy was studied from an information-theoretic framework in some recent works~\cite{Wan19, Kamath19,AravindNCC20,Yan20,WanD2D19,Aravind20}. 
The works~\cite{Wan19} and~\cite{Kamath19} (a preliminary version of this work) demonstrated that a demand-private scheme for $N$ files and $K$ users can be obtained from a non-private scheme for $N$ files and $NK$ users. The rate achievable using such schemes was shown to be order optimal for all regimes except for the case when $K<N$ and $M<N/K$. 
A demand-private scheme using MDS codes was also proposed for $M\geq N/2$ in~\cite{Wan19}. 
In~\cite{AravindNCC20}, the authors focused on obtaining demand-private schemes that achieve a weaker privacy condition such that one user should not get any information about the demand of another user, but may gain some information about the demand vector. They mainly addressed the subpacketization requirement for $N=K=2$ in~\cite{AravindNCC20} and extended their study to more general cases in~\cite{Aravind20}. Demand privacy against colluding users was studied for device-to-device network in~\cite{WanD2D19} where a trusted server helps to co-ordinate among the users to achieve a demand-private scheme. The case of colluding users for the coded caching problem was considered in~\cite{Yan20} where the privacy condition was such that one user should not learn any information about the demands of other users even if she is given all the files.

Now we briefly summarize the main contributions of this paper.
We first show that a demand-private scheme for $N$ files and $K$
users  can be obtained from a non-private
scheme that serves only a subset of demands for $N$ files and $NK$ users (Theorem~\ref{Thm_genach}). Our first achievability scheme, Scheme A, is built on this fact. We then propose Scheme B which is based on the idea that permuting broadcast symbols and not fully revealing the permutation function helps to achieve demand privacy. Our third achievability scheme, Scheme C, combines the ideas of Schemes A and B.
Using these 
achievability schemes, we show the order optimality for the case when $K < N$ and $M < N/K$, thus completing the order optimality result for all regimes\footnote{This result was first shown in a preliminary version~\cite{KamathRD20} of this work.}.  Finally, we  characterize the exact memory-rate trade-off under  demand privacy for the case $N\geq K=2$. We detail the contributions and describe the organization of the paper in the next subsection.

\subsection{Contributions and organization of the paper}

The main contributions of this paper are the following.

\begin{enumerate}
	
\item Using the fact that a demand-private scheme for $N$ files and $K$ users  can be obtained from a non-private
scheme that serves only a structured subset of demands for $N$ files and $NK$ users, we propose Scheme A for demand-private caching that uses the non-private scheme for $N$ files and
$NK-K+1$ users from~\cite{Yu18} (which we refer to as
the YMA scheme). This then implies that the memory-rate pairs achievable by the YMA scheme for $N$ files and $NK -K+1$ users are also achievable under demand privacy for $N$ files and $K$ users (Theorem~\ref{Corl_reduced_usrs} in Section~\ref{Sec_caching_random}).
	
\item  In~\cite[Example~1]{Maddah14}, it was shown that memory-rate pair $(M,\min \{N,K\}(1-M/N))$ can be achieved for non-private schemes without any coding in the placement phase or in the  delivery phase. In Theorem~\ref{th:basic} (Section~\ref{Sec_delivery_random}), we show that this memory-rate pair $(M,\min \{N,K\}(1-M/N))$
is also achievable under demand privacy. For $N\leq K$, the scheme (Scheme B) that achieves this pair is trivial, while for $K < N$, the scheme is non-trivial.

\item We then propose a demand-private scheme (Scheme C) that builds on the ideas of Schemes A and B. The memory-rate pairs achievable using Scheme C are given in Theorem~\ref{Thm_PR_SR} (Section~\ref{Sec_random_cach_delivry}).
Using numerical computations, we demonstrate that, for $K < N$, a combination of Schemes B and C outperforms Scheme A. In contrast, Scheme A outperforms  Schemes B and C for $N \leq K$.

\item The characterization of the exact memory-rate trade-off is known to be difficult for non-private schemes. So, the order optimality of the achievable rates is investigated. We show that the rates  achievable using our schemes are within a constant multiplicative gap of the optimal non-private rates (Theorem~\ref{Thm_order} in Section~\ref{Sec_order_optimal}) in all parameter regimes. In particular, we prove this for $K<N$ and $M<N/K$,
the regime that was left open in previous works. This gives the order optimality result since the optimal rates under privacy is lower bounded by the optimal non-private rates.
This also implies that the optimal private and non-private rates are always within a constant factor.

\item  One class of instances for which we have the  exact trade-off \cite{Maddah14, Tian2018} for non-private schemes is when $K=2$ and $N \geq 2$. For this class, we characterize the exact trade-off under demand privacy in Theorem~\ref{Thm_exact_region} (Section~\ref{sec_exact}). Our characterization shows that the exact trade-off region under demand privacy for this class is strictly smaller than the one without privacy. To characterize the exact trade-off, we give a converse bound that accounts for the privacy constraints. To
the best of our knowledge, this converse bound is the first of its kind, and also that this is the first instance where it is demonstrated that the optimal rates with privacy can be strictly larger than the optimal rates without privacy.
\end{enumerate}
 The rest of the paper is organized as follows. In Section~\ref{sec_problem}, we give our problem formulation. We present our results and briefly describe our proposed schemes in  Section~\ref{sec_results}.  All the proofs of our results can be found in Section~\ref{sec_proofs} and the appendices.

\subsection{Notations}
We denote the set $\{0,1, \ldots, N-1\}$ by
$[0:N-1]$, the cardinality of a set $\cA$  by $|\cA|$, and the closed interval between two real numbers $a$ and $b$ by $[a,b]$. For a positive integer $\ell$, if $\pi$ denotes a  permutation of $[0:\ell-1]$, and $Y=(Y_0, Y_1, \ldots , Y_{\ell-1})$, with abuse of notation, we define $\pi(Y) = \left(Y_{\pi^{-1}(i)}\right)_{i \in [0:\ell-1]}$. We denote random variables by upper case letters (e.g. $X$) and their alphabets by calligraphic letters (e.g. $\cX$).
For a random variable/vector $B$, $len(B)$ denotes $\log_{2} |\cB|$.

%% file: model.tex
Consider one server connected to $K$ users through a noiseless broadcast link. The server has access to $N$ independent files of $F$ bits each. These files are  denoted as $(W_0, W_1, \ldots, W_{N-1})$ and each file is uniformly distributed in $\{0,1\}^F$.
Each user has a cache of size $MF$ bits. The coded caching problem has two phases: prefetching and delivery.  In the prefetching phase, the server places at most $MF$ bits in the cache of each user. The cache content of user $k\in [0:K-1]$ is denoted by $Z_k$. In the delivery phase, each user demands one of the $N$ files from the server and this demand is conveyed secretly to the server. 
Let the demand of user $k$ be denoted by $D_k \in [0:N-1]$. We define $\bar{D} = (D_0,D_1,\ldots,D_{K-1})$. $\bar{D}$ is independent of the files $W_i, i \in [0:N-1]$ and caches $Z_k, k \in [0:K-1]$, and is uniformly distributed in $[0:N-1]^K$.

In the delivery phase, the server broadcasts a message $X$ to all the $K$ users such that user $k \in [0:K-1]$ can decode file $W_{D_k}$ using $X$ and $Z_k$  (see Fig.~\ref{Fig_cach_setup}). If message $X$ consists of $RF$ bits, then $R$ is said to be the rate of transmission. In addition to the decodability of the demanded file, demand-privacy imposes another constraint that the demands of all other users  should remain perfectly secret to each of the $K$ users. To ensure demand-privacy, the server can share some randomness denoted by $S_k$ with user $k\in [0:K-1]$ in the prefetching phase. This shared randomness is of negligible size and hence, it is not included in the memory size.
We define $S = (S_0, S_1, \ldots, S_{K-1})$. The server also has access to some private randomness which we denote by $P$. The  random variables $S , P , \{W_i | i \in [0:N-1]\} , \{D_k|k \in [0:K-1]\}$ are independent of each other.

\begin{figure}[htb]
  \centering
   \includegraphics[scale=0.5]{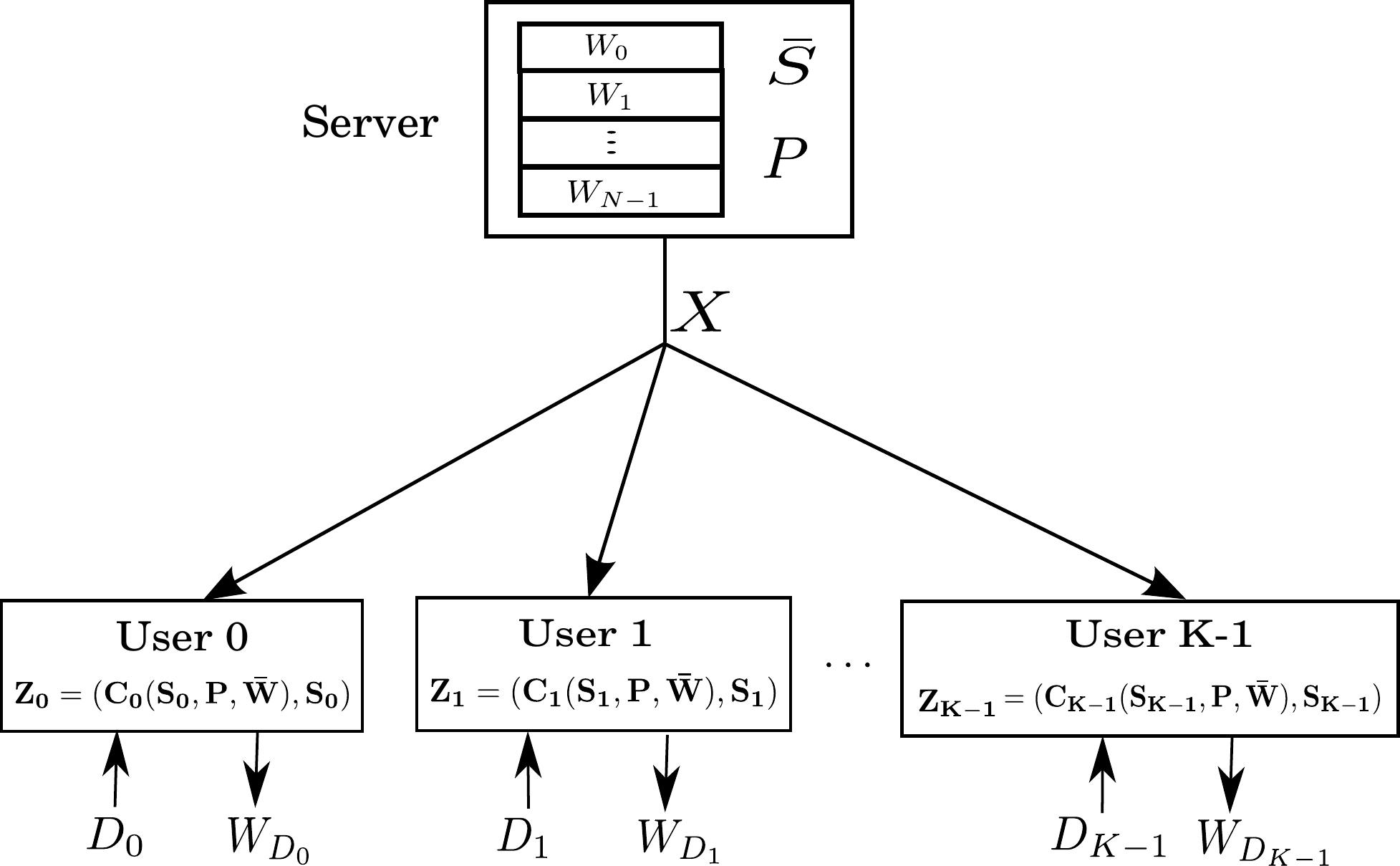}
  \caption{Demand-private coded caching model.}
  \label{Fig_cach_setup}
\end{figure}

{\bf Non-private coded caching scheme:}
An {\it non-private coded caching scheme} consists of the following.

{\it Cache encoding functions:} For $k \in [0:K-1]$, the cache encoding function for the $k$-th user is a map
\begin{align}
C_{k}: {[0:2^F-1]}^N  \rightarrow [0:2^{MF}-1], \label{Def_cach_enc_np}
\end{align} 
and the cache content $Z_k$ is given by $Z_k =C_k(\bar{W})$. 

{\it Broadcast transmission encoding function:} The  transmission
encoding is a map
\begin{align}
E: {[0:2^F-1]}^N \times \cD_0  \times \cdots \times \cD_{K-1}  \rightarrow [0:2^{RF}-1], \label{Def_Tx_enc_np}
\end{align}
and the transmitted message is given by $X=(E(\bar{W}, \bar{D}), \bar{D})$. 

{\it Decoding functions:} User $k$ uses a decoding function
\begin{align}
G_{k}:  \cD_0 \times \cdots \times \cD_{K-1}  \times [0:2^{RF}-1] \times [0:2^{MF}-1]  \rightarrow [0:2^{F}-1]. \label{Def_dec_np}
\end{align}
Let $\cC = \{C_k: k=0,\ldots,K-1\}$ and $\cG =  \{G_k: k=0,\ldots,K-1\}$. Then
the triple $(\cC, E, \cG)$ is called an
$(N,K,M,R)$-non-private scheme if it satisfies 
\begin{align} 
W_{D_{k}} = G_k(\bar{D}, E(\bar{W}, \bar{D}), C_k(\bar{W}))
\label{Eq_dec_cond}
\end{align}
for all values of $\bar{D}$ and $ \bar{W}$.
A memory-rate pair $(M,R)$ is said to be {\em achievable} for the $(N,K)$ coded
caching problem if there exists an $(N,K,M,R)$-non-private scheme  for some
$F$. The {\em  memory-rate trade-off } $R^{*}_{N,K}(M)$ for the non-private coded caching problem
is defined as 
\begin{align}
R^{*}_{N,K}(M)&=\inf\{R: (M,R) \mbox{ is achievable  for $(N,K)$ coded caching problem} \}. \label{Eq_opt_rate_nopriv}
\end{align}

{\bf Private coded caching scheme:} A {\it private coded caching scheme}
consists of the following. 

{\it Cache encoding functions:}
For $k \in [0:K-1]$, the cache encoding function for the $k$-th user is given by
\begin{align}
C_{k} :\cS_k \times \cP \times {[0:2^F-1]}^N  \rightarrow [0:2^{MF}-1], \label{Def_cach_enc}
\end{align} 
and the cache content $Z_k$ is given by 
\mbox{$Z_k =(C_k(S_k, P, \bar{W}), S_k)$}. 

{\it Broadcast transmission encoding function:}
The   transmission encoding functions are
\begin{align*}
&E: {[0:2^F-1]}^N \times \cD_0 \times \cdots \times \cD_{K-1} \times\cP  \times \cS_0 \times \cdots \times \cS_{K-1} \rightarrow [0:2^{RF}-1], \\
&J:\cD_0 \times \cdots \times \cD_{K-1} \times \cP \times \cS_0 \times \cdots \times \cS_{K-1} \rightarrow \cJ.
\end{align*}
The transmitted message $X$ is given by
\begin{align*}
 X=\left(E(\bar{W}, \bar{D}, P, \bar{S}), J(\bar{D}, P, \bar{S}) \right). 
\end{align*}
Here $\log_2 |\cJ|$ is negligible\footnote{ The auxiliary
transmission $J$ essentially captures any additional transmission, that does not
contribute any rate, in addition to the main
payload. Such auxiliary transmissions of negligible rate are used even in non-private
schemes without being formally stated in most work. For example, the
scheme in~\cite{Maddah14} works only if the server additionally transmits the demand vector
in the delivery phase. We have chosen to
formally define such auxiliary transmission here.} compared to file size $F$.

{\it Decoding functions:}
User $k$ has a decoding function
\begin{align}
G_{k} : \cD_k \times \cS_k \times \cJ \times [0:2^{RF}-1] \times [0:2^{MF}-1]  \rightarrow [0:2^{F}-1]. \label{Def_dec}
\end{align}
Let $\cC = \{C_k: k=0,\ldots,K-1\}$ and $\cG =  \{G_k: k=0,\ldots,K-1\}$. The tuple  $(\cC, E, J,\cG)$ is called as an $(N,K,M,R)$-private scheme if it satisfies the following decoding and privacy conditions:
\begin{align*} 
& W_{D_{k}} = G_k\bigl(D_k,S_k,J(\bar{D}, P, \bar{S}, ),  E(\bar{W}, \bar{D},P,\bar{S}), C_k(S_k, P,\bar{W})\bigr), \quad \text{ for }   k \in [0:K-1],\\
& I\left(\ND{k}; D_k,S_k,J(\bar{D}, P, \bar{S}, ),  E(\bar{W}, \bar{D},P,\bar{S}), C_k(S_k, P,\bar{W})\ \right)   = 0, \quad \text{ for }   k \in [0:K-1],
\end{align*}
where $\bar{D}_{\tilde{k}} = (D_0, \ldots, D_{k-1},D_{k+1}, \ldots, D_{K-1})$.
The above conditions are respectively equivalent to
\begin{align} 
H(W_{D_{k}}|Z_k,\TRS, D_k ) & =0, \quad \text{ for }   k \in [0:K-1],
 \label{Eq_decod_cond}\\
I(\ND{k};Z_k,\TRS, D_k )  & = 0, \quad \text{ for }   k \in [0:K-1]. \label{Eq_instant_priv}
\end{align}

A memory-rate pair $(M,R)$ is said to be {\em achievable with demand privacy}
for the $(N,K)$ coded caching problem if there exists an $(N,K,M,R)$-private
scheme  for some  $F$. The {\em memory-rate trade-off with demand privacy} is defined as 
\begin{align}
R^{*p}_{N,K}(M)&=\inf\{R: (M,R) \mbox{ is achievable with demand  privacy for $(N,K)$ coded caching problem} \}. \label{Eq_opt_rate_priv}
\end{align}

\begin{remark}[Different privacy metrics]
	\label{Remark_priv}

 A weaker notion of privacy was considered in~\cite{AravindNCC20, Aravind20} given by
	\begin{align}
	I(D_i ;Z_k, D_k,X ) = 0, \quad  i\neq k. \label{Eq_weak_priv}
	\end{align}
	In words, the privacy condition in~\eqref{Eq_weak_priv} requires that user $k \in [0:K-1]$ should not get any information about $D_i, i\neq k$, but may have some information about the demand vector.
	Note that a scheme that satisfies the privacy condition~\eqref{Eq_instant_priv} also satisfies~\eqref{Eq_weak_priv}.
	The model in~\cite{Yan20} assumed that the users can collude, and  the following stronger notion of privacy metric was considered
	\begin{align}
	I(D_{[0:K-1]\setminus \cS}; Z_{\cS}, D_{\cS}, X|\bar{W}) = 0, \quad \forall \cS \subseteq [0:K-1] \label{Eq_strng_priv_metric}
	\end{align} 
	where $D_{\cS}$ and $Z_{\cS}$ denote the demands and the caches of users in $\cS$, respectively.
	This stronger privacy metric is also satisfied by our Scheme A described in Subsection~\ref{Sec_caching_random} (see Remark~\ref{Rem_SchmA}). In contrast, Schemes B and C, described in Subsections~\ref{Sec_delivery_random} and~\ref{Sec_random_cach_delivry}, respectively, do not satisfy this stronger privacy metric (see  Remark~\ref{Rem_SchemeC}).
\end{remark}

%% file: results1.tex
In this section, we present our results that include our achievability schemes,
the tightness of the memory-rate pairs achievable using these schemes and the exact trade-off for $N\geq K= 2$.  In Subsections~\ref{Sec_caching_random}, \ref{Sec_delivery_random} and \ref{Sec_random_cach_delivry}, we discuss  Schemes A, B and C, respectively and the memory-rate pairs achievable using these schemes.
 We give a comparison of the  memory-rate pairs achievable using Schemes A, B and C  in Subsection~\ref{sec_compare}. In particular, we show that Scheme A outperforms Schemes B and  C for $N \leq K$, while a combination of Schemes B and C outperforms Scheme A for $K < N$.
In Subsection~\ref{Sec_order_optimal}, we discuss the 
tightness of the achievable memory-rate pairs, and  show the order optimality result for all regimes. Finally, we present the
exact memory-rate trade-off  under demand privacy for the case $N\geq K=2$ in Subsection~\ref{sec_exact}.

\subsection{Scheme A}
\label{Sec_caching_random}

It was observed in~\cite{Wan19, Kamath19} that  a demand-private scheme for $N$ files and $K$ users can be obtained using
an existing non-private achievable scheme for $N$ files and $NK$
users as a blackbox. Here every user is associated with a stack of $N$ virtual users in
the non-private caching problem. For
example, demand-private schemes for $N=K=2$ are obtained from the non-private
schemes for $N=2$ and $K=4$. 
We next show that  only certain types of demand vectors
of the non-private scheme are required in the private scheme. 
To this end, we define this particular subset of demand vectors.

Consider a non-private coded caching problem with $N$ files and $NK$ users. 
A demand vector $\bd$ in this problem is an $NK$-length vector, 
where the $j$-th component denotes the demand of user $j$.
Then $\bd$ can also be represented as $K$ subvectors of length $N$ each, i.e.,	
\begin{align}
\bd= [\bd^{(0)},\bd^{(1)},\ldots,\bd^{(K-1)}]
\end{align}
where $\bd^{(i)}\in [0:N-1]^N$ is  an $N$-length vector for all $i \in [0:K-1]$. We now define a \emph{restricted demand subset} $\DRS$.
 \begin{definition}[Restricted Demand Subset $\DRS$]
	\label{Def_dmnd_subst}
	The restricted demand subset $\DRS$ for an $(N,NK)$ coded caching problem
	is the set of all $\bd$ such that $\bd^{(i)}$ is a cyclic shift of the vector $(0, 1,  \ldots, N-1)$ for all $i=0,1, \ldots, K-1$.
\end{definition}
Since $N$ cyclic shifts are possible for each $\bd^{(i)}$, there are a total of $N^K$ demand vectors in $\DRS$. 

For a given $\bd \in \DRS$ and $i\in [0:K-1]$, let $c_i$ denote the number of right cyclic
shifts of $(0,1,\ldots,N-1)$ needed to get $\bd^{(i)}$. Then, 
$\bd \in \DRS$ is  uniquely identified by the vector $\bar{c}(\bd)
:= (c_0, \ldots, c_{K-1})$. For $N=2$ and $NK =4$, the demands in $\DRS$ and
their corresponding $\bar{c}(\bar{d}_s)$ are given in Table~\ref{Tab_2x2}.

\begin{table}[h]
	\begin{center}
		\begin{tabular}{|c|c|c|c|c|}
			\hline
			$D_{0}$ & $D_{1}$ & $D_{2}$ & $D_{3}$ & $\bar{c}(\bar{d}_s)$ \\
			\hline
			$0$ & $1$ & $0$ & $1$ & $(0,0)$\\
			\hline
			$0$ & $1$ & $1$ & $0$ & $(0,1)$\\
			\hline
			$1$ & $0$ & $0$ & $1$ & $(1,0)$\\
			\hline
			$1$ & $0$ & $1$ & $0$ & $(1,1)$\\
			\hline
		\end{tabular}
	\end{center}
	\caption{Restricted Demand Subset $\DRS$ for $N=2$ and $NK=4$.}
	\label{Tab_2x2}
\end{table}

A non-private scheme for an $(N,K)$ coded caching problem that serves all demand vectors in a particular set $\cD\subseteq [0:N-1]^K$, is called a $\cD$-non-private scheme. We have the following theorem.
\begin{theorem}
	\label{Thm_genach}
	If there exists an $(N,NK,M,R)$ \DRS-non-private scheme, then
	there exists an $(N,K,M,R)$-private scheme.
\end{theorem}

The proof of Theorem~\ref{Thm_genach} is given in Subsection~\ref{Sec_proof_Thm_genach}.  The proof follows by showing a  construction of an $(N,K,M,R)$-private scheme using an $(N,NK,M,R)$
\DRS-non-private scheme as a blackbox. The following example shows a construction of $(2,2,\frac{1}{3}, \frac{4}{3})$-private scheme from $(2,4,\frac{1}{3}, \frac{4}{3})$-\DRS-non-private scheme. This particular non-private scheme is from~\cite{Tian2018}.  
It is important to note that the memory-rate pair $(\frac{1}{3}, \frac{4}{3})$ is not achievable for $N=2, K=4$ under no privacy requirement.  Thus, we  observe that there exist memory-rate pairs that are achievable for the \DRS-non-private scheme, but not achievable for the non-private scheme which serves all demands.

\begin{example}
	\label{Ex_cach_random}
	We consider the demand-private coded caching problem for $N=2,K=2, M=1/3$. It was shown in~\cite{Tian2018} that
	for memory $M=1/3$, the optimum non-private rate for $N=2,K=4$ satisfies $R^{*}_{2,4}(1/3) > 4/3$.  Next we give a scheme which achieves a rate $4/3$ under demand privacy for $N=2,K=2, M=1/3$. The other known demand-private schemes also do not achieve $R=4/3$ for $N=2,K=2$. See Fig.~\ref{Fig_schemeC} for reference.

	Let $A$ and $B$ denote the two files. We will now
	give a scheme which achieves a rate $4/3$ for $M=1/3$ with $F=3l$ for some positive integer $l$. We denote the
	3 segments of $A$ and $B$ by $A_1,A_2,A_3$ and $B_1,B_2,B_3$ respectively, of $l$ bits each.  
	First let us consider a $\DRS$-non-private scheme for $N=2$ and $K=4$ from~\cite{Tian2018}. Let
	$C_{i,j}(A,B), i,j=0,1$, as shown in Table~\ref{Table_cache_NK2}, correspond to the cache
	content of user $2i+j$ in the $\DRS$-non-private scheme.  The transmission
	$T_{(p,q)}(A,B), p,q=0,1$, as given in Table~\ref{Tab_Tx}, 
	is chosen for the demand
	$\bar{d} \in \DRS$ such that $ (p,q) = \bar{c}(\bar{d})$. Using
	Tables~\ref{Table_cache_NK2} and \ref{Tab_Tx}, 
	it is easy to verify that the non-private scheme satisfies the decodability condition for demands in \DRS. From this
	scheme, we obtain a demand-private scheme for $N=2,K=2$ as follows. 

	\begin{table}[h]
	\begin{center}
		\begin{tabular}{|c|c|}
			\hline
			Cache & Cache Content\\
			\hline
			$C_{0,0}(A,B)$ & $A_1\oplus B_1$ \\
			\hline
			$C_{0,1}(A,B)$ & $A_3\oplus B_3$\\
			\hline
			$C_{1,0}(A,B)$ & $A_2\oplus B_2$ \\
			\hline
			$C_{1,1}(A,B)$ & $A_1\oplus A_2\oplus A_3\oplus B_1\oplus B_2\oplus B_3$\\
			\hline
		\end{tabular}
\end{center}
		\caption{Choices for the caches of user 0 and user 1.}
		\label{Table_cache_NK2}
	\end{table}
	
	\begin{table}[h]
		\begin{center}
			\begin{tabular}{|c|c|}
			\hline
			$T_{(0,0)}(A,B)$ &  $B_1, B_2, A_3, A_1\oplus A_2\oplus A_3$ \\
			\hline
			$T_{(0,1)}(A,B)$ &  $A_2, A_3, B_1, B_1\oplus B_2\oplus B_3$ \\
			\hline
			$T_{(1,0)}(A,B)$ &  $B_2, B_3, A_1, A_1\oplus A_2\oplus A_3$ \\
			\hline
			$T_{(1,1)}(A,B)$ &  $A_1, A_2, B_3, B_1 \oplus B_2 \oplus B_3 $ \\
			\hline
		\end{tabular}
\end{center}
		\caption{Transmissions for $(2,2, \frac{1}{3}, \frac{4}{3})$-private scheme.}
		\label{Tab_Tx}
	\end{table}

	Let the shared key   $S_k ,k=0,1$ of user $k$ be a uniform binary random variable.
	The cache encoding functions and the  transmission encoding function are denoted as
	\begin{align*}
	C_k(S_k, A,B) & = C_{k, S_k} (A,B) \text{ for } k=0,1, \\
	E(A,B, D_0, D_1, S_0, S_1) & = T_{(D_0\oplus S_0, D_1 \oplus S_1)}(A,B).
	\end{align*}
	User $k$ chooses $C_{k, S_k} (A,B)$ given in
	Table~\ref{Table_cache_NK2} as the cache encoding function.
	In the delivery phase, for given $(S_0,S_1)$ and $(D_0,D_1)$,
	the server broadcasts $T_{(D_0 \oplus S_0,D_1 \oplus S_1)}(A,B)$ as the main payload and $(D_0 \oplus S_0,D_1
	\oplus S_1)$ as the auxiliary transmission. For such a transmission, the decodability follows from the decodability of the chosen non-private scheme.

	Further, the broadcast  transmission will
	not reveal any information about the demand of one user to the other user since
	one particular transmission $T_{(p,q)}(A,B)$ happens for all demand vectors
	$(D_0,D_1)$, and also that $S_i$ acts as one time pad for $D_i$ for each $i=0,1$. Here,
	all the transmissions consist of $4l$ bits (neglecting the 2 bits for $(D_0
	\oplus S_0,D_1 \oplus S_1)$). Since $F=3l$, this scheme achieves a rate $R =4/3$.
\end{example}

Example~\ref{Ex_cach_random} showed that there exists an $(M,R)$ pair that is not achievable for $N$ files and $NK$ users, but it is achievable with demand privacy for $N$ files and $K$ users. This was possible because a \DRS-non-private scheme needs to serve only a subset of demands. Our Scheme A as described later utilizes this fact, and obtains a general scheme for any  parameters $N$ and $K$. Specifically, we show that
 a \DRS-non-private scheme for $N$ files and $NK$ users can be obtained from the non-private scheme given in~\cite{Yu18} for $N$ files and $NK-K+1$ users. The memory-rate pairs achievable using Scheme A are presented in Theorem~\ref{Corl_reduced_usrs}.
We use the following lemma to prove Theorem~\ref{Corl_reduced_usrs}.
\begin{lemma}
	\label{Thm_reduced_usrs}
	For the $(M,R)$ pairs given by
	\begin{align*}
	(M,R)= \left(\frac{Nr}{NK-K+1} , \frac{{NK-K+1\choose r+1}-{NK-K+1-N \choose r+1}}{{NK-K+1\choose r}}\right), \quad \text{ for } r \in \{0,1,\ldots,NK-K\} \label{Eq_YMA_achv_pair}
	\end{align*}
	which are achievable for the non-private coded caching problem with  $N$ files and $NK-K+1$ users by the YMA scheme~\cite{Yu18}, there exists an
	$(N,NK,M,R)$ \DRS-non-private scheme.
\end{lemma}

The proof of Lemma~\ref{Thm_reduced_usrs} can be found in Subsection~\ref{Sec_proof_reducd_usrs}. The proof follows by dividing $NK$ users in the \DRS-non-private scheme into two groups with the first group  containing $K-1$ users and the second group containing  $NK-K+1$ users. Users in the second group follow the prefetching of the YMA scheme while users in the first group follow coded prefetching. In particular, the users in the first group follow the coded prefetching of Type III caching discussed in~\cite{ShaoVZT19}. In the delivery phase, for a given $\bar{d} \in\DRS$, the server chooses the transmission of the YMA scheme corresponds to the demands in the second group of users. Due to the special nature of the demand vectors in \DRS, using this transmission, the demands of all users in the first group can also be served.

\underline{Scheme A:} Scheme A consists of two steps. In the first step, a \DRS-non-private scheme is obtained from the non-private YMA scheme for $N$ files and $NK-K+1$ users. In the second step, an $(N,K, M,R)$-private scheme is obtained using this \DRS-non-private scheme as a blackbox. Scheme A achieves the memory-rate pairs given in the following theorem.

\begin{theorem}
	\label{Corl_reduced_usrs}
There exists an $(N,K,M,R)$-private scheme with the following memory-rate pair:
	\begin{align} 
(M, R)=\left(\frac{Nr}{NK-K+1}, \frac{{NK-K+1\choose r+1} -{NK-K-N+1\choose r+1}}{{NK-K+1\choose r}}\right), \quad \mbox{ for } r=\{0, \ldots, NK-K+1\}.
	\end{align}
\end{theorem}

\begin{proof}
The given memory-rate pair is achievable  by the YMA scheme for $N$ files and $NK-K+1$ users. So, the theorem follows from Theorem~\ref{Thm_reduced_usrs} and Lemma~\ref{Thm_genach}.
\end{proof}

\begin{remark}
	\label{Rem_SchmA}
	If a private scheme is derived from a \DRS-non-private scheme
	using the construction described in the proof of Theorem~\ref{Thm_genach}, then it also satisfies the stronger notion of privacy metric~\eqref{Eq_strng_priv_metric}. 
	This can be shown by replacing $\ND{k}$ by $D_{[0:K-1]\setminus \cS}$, $D_k$ by $D_{\cS}$, and $Z_k$ by $Z_{\cS}$ in the proof of privacy that led to~\eqref{Eq_privcy4}. Since Scheme A is obtained using a \DRS-non-private scheme as a blackbox, it also satisfies the stronger privacy condition~\eqref{Eq_strng_priv_metric}.
\end{remark}

\subsection{Scheme B}
\label{Sec_delivery_random}

Now we describe Scheme B.
For $N \leq K$, Scheme B is trivial, where the caches of all users are populated with the same $M/N$ fraction of each file in the placement phase. In the delivery phase, the uncached parts of all files are transmitted. In this scheme, all users get all files, and the rate of transmission is given by $N(1-M/N) = N-M$. Since the broadcast transmission is independent of the demands, it clearly satisfies the privacy condition~\eqref{Eq_instant_priv}. However, if the number of users is less than the number of files, then this scheme is very wasteful in terms of rate.
For $K<N$,  next we give an outline of Scheme B.  It achieves a rate $K(1-M/N)$ that, in this case, is an improved rate compared to $N-M$.

For $K < N$, let us first consider the non-private scheme~\cite[Example~1]{Maddah14} which achieves rate  $K(1-M/N)$. In this scheme,  all users store the same $M/N$ fraction of each file in the placement phase. In the delivery phase, the server transmits $K$ components where $i$-th component consists of the uncached part of the file demanded by user $i$. However, this scheme does not ensure demand privacy since if the $i$-th component is different from the $j$-th component, $i \neq j$, then user $i$ learns that $D_j \neq D_i$. This clearly violates demand privacy. For $K<N$, the placement phase in Scheme B is the same as that of for $N\leq K$. In the delivery phase, the server transmits $K$ components without violating demand privacy. We next illustrate Scheme B using an example with $N>2$ and $K=2$.

\begin{example}
	\label{Ex_no_coding}
	Let us consider that there are two users and more than two files, i.e., $N>2$ and $K=2$. In the placement phase, each user stores $M/N$ fraction of each file. In  the delivery phase, first let us consider the case of $D_0 \neq D_1$. In this case, the server transmits two components which correspond to the uncached parts of each demanded file. To achieve privacy, the position where the uncached part of $W_{D_0}$ is placed, is selected from one out of two possible choices uniformly at random. The uncached fraction of $W_{D_1}$ is placed in the other position. These positions are conveyed to each user in the auxiliary transmissions. The random variable to convey the position to user $0$ is XOR-ed with shared randomness $S_0$. Since $S_0$ is known only to user $0$, it acts as an one-time pad.  Similarly, $S_1$ helps in protecting the privacy against user $0$. When $D_0 = D_1$, the uncached part of the file is placed in a position chosen randomly. The other position is filled with random bits. Since one user does not have any information about the component from which the other user's demanded file is decoded, and since the files are independent of the demands, this scheme preserves the demand privacy.
\end{example}

For $K <N$, Scheme B is a generalization of the scheme presented in Example~\ref{Ex_no_coding}. Scheme B achieves the memory-rate pairs given in the following theorem.

\begin{theorem}
	\label{th:basic}
	There exists an $(N,K,M,R)$-private scheme with the memory-rate pair $(M,\min \{N,K\}(1-M/N))$.
\end{theorem}
Full description of Scheme B and the proof of Theorem~\ref{th:basic} are provided in Subsection~\ref{Sec_proof_no_coding}.

 \input{schemeC.tex}
 
 \subsection{Comparison of our schemes}
 \label{sec_compare}
 
  \begin{figure}[h]
 	\centering
 	\includegraphics[scale=0.6]{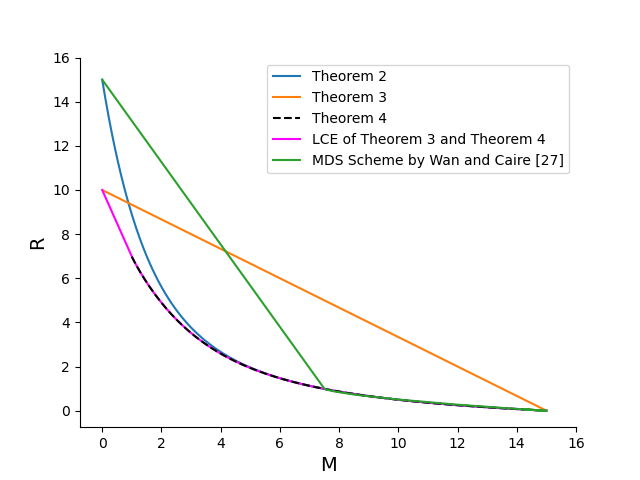}
 	\caption{Comparison of different schemes for $N=15$ and $K=10$. The region given by the lower convex envelop (LCE) of the the points in Theorem~\ref{th:basic} and Theorem~\ref{Thm_PR_SR}   is larger than the region given by the LCE  of the points in Theorem~\ref{Corl_reduced_usrs}.}
 	\label{Fig_N15K10}
 \end{figure}
 
 Now we give a comparison of our schemes.  
 From numerical simulations we observe that, a combination of Schemes B and C outperforms Scheme A for $K<N$, and Scheme A outperforms both Schemes B and C for $N\leq K$, i.e., for $K < N$, the region given by the lower convex envelop (LCE) of the the points in Theorem~\ref{th:basic} and Theorem~\ref{Thm_PR_SR},  is  larger than the region given by  the LCE of the points in  Theorem~\ref{Corl_reduced_usrs}.  Whereas, we observe the opposite for $N \leq K$, i.e., the region given by the LCE of the points in Theorem~\ref{Corl_reduced_usrs} is larger than the region given by the LCE of the points  in Theorem~\ref{th:basic} and Theorem~\ref{Thm_PR_SR}.
 In Fig.~\ref{Fig_N15K10}, we plot the memory-rate pairs achievable using our schemes along with the pairs of achievable using the MDS scheme in~\cite{Wan19} for $N=15$ and $K=10$.
 In Fig.~\ref{Fig_N10K15}, we give a comparison of the memory-rate pairs achievable using different schemes for $N=10$ and $K=15$.

 \begin{figure}[h]
 	\centering
 	\includegraphics[scale=0.6]{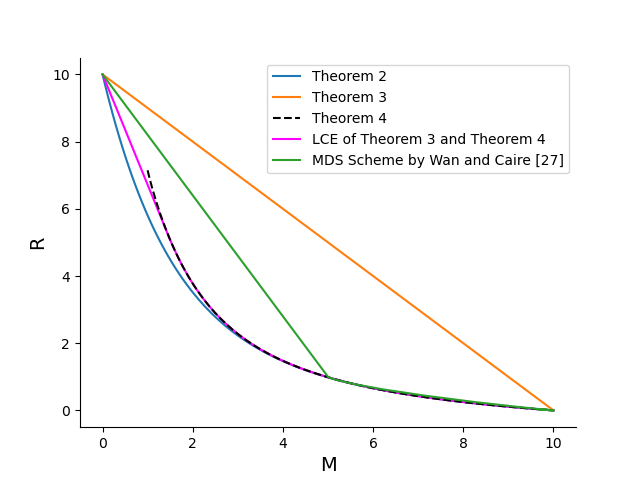}
 	\caption{Comparison of different schemes for $N=10$ and $K=15$. The region given by the lower convex envelop (LCE) of the points in Theorem~\ref{Corl_reduced_usrs} is larger than the region given by the LCE of the points in Theorem~\ref{th:basic} and Theorem~\ref{Thm_PR_SR}.}
 	\label{Fig_N10K15}
 \end{figure}

\subsection{Tightness of the achievable memory-rate pairs}
\label{Sec_order_optimal}

Now we compare the memory-rate pairs achievable using our schemes  with lower bounds on the optimal rates for non-private schemes. Recall that  for $N$ files, $K$ users and memory $M$, $R^{*p}_{N,K}(M)$ and $R^{*}_{N,K}(M)$ denote the optimal private rate and non-private rate, respectively.

\begin{theorem}
	\label{Thm_order}
 Let 
$R^{A}_{N,K}(M)$  denote the LCE of the points in Theorem~\ref{Corl_reduced_usrs}, and let $R^{BC}_{N,K}(M)$ denote the LCE of the points in Theorem~\ref{th:basic} and Theorem~\ref{Thm_PR_SR}. Then, we have
	\begin{enumerate}
		\item 
		For $N \leq K$,
		\begin{align}
		\label{Eq_ratio_bound1}
		\frac{R^{A}_{N,K}(M)}{R^{*}_{N,K}(M)}  \leq 
		\begin{cases}
		4 & \text{ if } M \leq \left(1 - \frac{N}{K}\right)\\
		8 & \text{ if}  \left(1 - \frac{N}{K}\right) \leq M \leq \frac{N}{2}\\
		2 & \text{ if }M \geq \frac{N}{2}.
		\end{cases}
		\end{align} \label{Thm_ordr_part1}
		\item \label{Thm_ordr_part2}
		For $ K < N$,
		\begin{align}
		\label{Eq_ratio_bound2}
		\frac{R^{BC}_{N,K}(M)}{R^{*}_{N,K}(M)} & \leq
		\begin{cases}
		3 & \text{ if } M < \frac{N}{2}\\
		2 & \text{ if } M \geq \frac{N}{2}.
		\end{cases}
		\end{align}
		\item \label{Thm_ordr_part_exct}
		 For all $N$ and $K$, $R^{*p}_{N,K}(M) = R^{*}_{N,K}(M)$ if $M \geq \frac{N(NK-K)}{NK-K+1}$.
	 \end{enumerate}	
  Since $R^{A}_{N,K}(M) \geq R^{*p}_{N,K}(M) \geq R^{*}_{N,K}(M)$, the upper bounds in~\eqref{Eq_ratio_bound1} also hold for the ratios $\frac{R^{A}_{N,K}(M)}{R^{*p}_{N,K}(M)}$ and $\frac{R^{*p}_{N,K}(M)}{R^{*}_{N,K}(M)}$. Similarly,   the upper bounds in~\eqref{Eq_ratio_bound2} also hold for the ratios $\frac{R^{BC}_{N,K}(M)}{R^{*p}_{N,K}(M)}$ and $\frac{R^{*p}_{N,K}(M)}{R^{*}_{N,K}(M)}$.
\end{theorem}

The proof of Theorem~\ref{Thm_order} is presented in Subsection~\ref{Sec_proof_order}. Theorem~\ref{Thm_order} shows that a combination of our schemes gives rates that are always within a constant multiplicative factor from the optimal, i.e., the order optimality result is shown for all cases. We also note that the order optimality result is also obtained in~\cite{Wan19} for all regimes except for the case when $ K < N$ and $M < N/K$. The constant factors in~\eqref{Eq_ratio_bound1} and also the factor $2$ for the regime $K<N, M \geq N/2$  are obtained in~\cite{Wan19}. In contrast, the constant factor $3$ in~\eqref{Eq_ratio_bound2} for the regime $K<N, M<N/2$ shows the order optimality for the case $K<N, M<N/K$, and improves the previously known factor $4$ for the case $K< N, M \leq N/K < N/2$.

 One natural question that arises in demand-private coded
caching is how much cost it incurs due to the extra constraint
of demand privacy. It follows from Theorem~\ref{Thm_order} that the extra cost is always within a constant factor.  However, we note that the extra cost may not be a constant factor for all the regimes under the stronger privacy condition~\eqref{Eq_strng_priv_metric}. For example, when $K<N$ and $M=0$, the optimal non-private rate is $K$.
However, for this case, the optimal private rate under the stronger privacy condition~\eqref{Eq_strng_priv_metric} is shown to be $N$ in~\cite{Yan20}, whereas the optimal private rate under the privacy condition~\eqref{Eq_instant_priv} is $K$ (Theorem~\ref{th:basic}).
Such a difference in rates under these two notions of privacy conditions  also extends for very small memory regimes when $K<N$. 

\subsection{Exact trade-off for  $N\geq K=2$}
\label{sec_exact}
For $N=K=2$, the exact trade-off  under no privacy  was shown in~\cite{Maddah14}. Tian characterized the exact trade-off under no privacy for $N>K=2$ in~\cite{Tian2018}. For $N=K=2$, the non-private trade-off region is characterized by three lines. Whereas, if $N>2, K=2$, then the non-private trade-off region is given by two lines.  We characterize the exact trade-off for $N\geq K=2$ under demand privacy in the following theorem. The characterization shows that the exact trade-off  region of private schemes for $N\geq K=2$ is always given by three lines.

\begin{figure}
	\centering
	\begin{subfigure}{0.4\textwidth}
		\centering
		\includegraphics[scale=0.5]{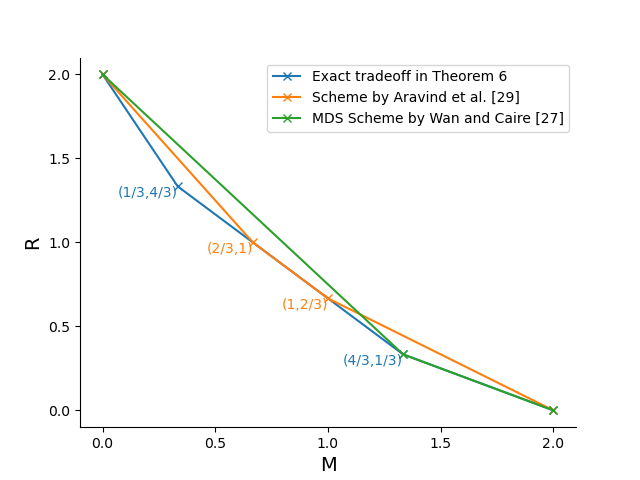}
	\end{subfigure}
\begin{subfigure}{0.4\textwidth}
		\centering
		\includegraphics[scale=0.5]{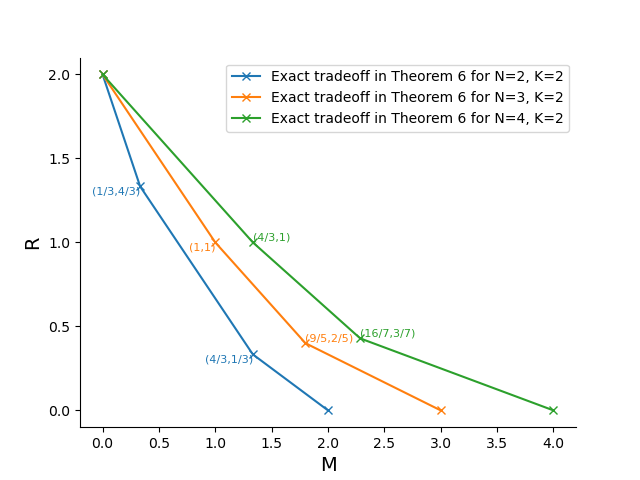}
\end{subfigure}
\caption{The figure on the left gives the exact trade-off with demand privacy for $N=K=2$ and the region given by other known  schemes. The figure on the right gives the exact trade-off with demand privacy for $N=2,3,4$ and $K=2$.}
	\label{Fig_schemeC}
\end{figure}

 \begin{theorem}
	\label{Thm_exact_region}
	\begin{enumerate}
		\item 	Any memory-rate pair $(M,R)$ is achievable with demand privacy for $N=K=2$ if and only if 
		\begin{align}
		2M + R \geq 2, \quad 3M+3R \geq 5, \quad M+2R \geq 2. \label{Eq_N2K2_region}
		\end{align}
		\label{N2K2_region}
		\item Any memory-rate pair $(M,R)$ is achievable with demand privacy for $N>K=2$ if and only if 
		\begin{align}
		3M+NR \geq 2N, \quad  3M+(N+1)R \geq 2N+1, \quad M+NR \geq N. \label{Eq_AnyNK2_region}
		\end{align}
		\label{AnyNK2_region}
	\end{enumerate}
\end{theorem}

 Next we give some outlines of the achievability schemes and the converse to obtain Theorem~\ref{Thm_exact_region}.

\noindent \underline{Outline of converse:}
Any $(M,R)$ pair that is achievable under no privacy requirement needs to satisfy the first and third inequalities in~\eqref{Eq_N2K2_region} for $N=K=2$~\cite{Maddah14}. Similarly, for $N>K=2$, any $(M,R)$ pair satisfies the first and third inequalities in~\eqref{Eq_AnyNK2_region} under no privacy~\cite{Tian2018}. Since any converse bound with no privacy requirement is also a converse bound with privacy. So, to prove the converse result, we need to show only the second inequality in~\eqref{Eq_AnyNK2_region} and  in~\eqref{Eq_N2K2_region}.
Furthermore, observe that substituting $N=2$ in the second inequality in~\eqref{Eq_AnyNK2_region} gives the second inequality in~\eqref{Eq_N2K2_region}. So, to show the converse of Theorem~\ref{Thm_exact_region}, we prove that for $N \geq K=2$, any $(M,R)$ pair under privacy satisfies the second inequality in~\eqref{Eq_AnyNK2_region}.  Full proof of the converse can be found in Subsection~\ref{Sec_proof_exact}.

\noindent \underline{Outline of  achievability:} To show the achievability of the region in \eqref{Eq_N2K2_region}, we use a particular non-private scheme from~\cite{Tian2018}. Using this particular non-private scheme, we can show that for any $(M,R)$ pair in the region by \eqref{Eq_N2K2_region}, there exists an $(2,4, M,R)$-\DRS-non-private scheme. Then, the achievability follows from Theorem~\ref{Thm_genach}. Details can be found in Subsection~\ref{sec_exact_achv}.

To prove  the achievability of the region given by~\eqref{Eq_AnyNK2_region} for $K=2$ and $N>2$, we show  the achievability of two corner points $(\frac{N}{3},1)$ and $(\frac{N^2}{2N-1}, \frac{N-1}{2N-1})$. The corner points of the memory-rate curve given by \eqref{Eq_AnyNK2_region} are $(0,2), (\frac{N}{3},1), (\frac{N^{2}}{2N-1}, \frac{N-1}{2N-1})$ and $(N,0)$. The achievability of  the points $(0,2)$ and $(N,0)$ follows from Theorem~\ref{th:basic}. We propose two  schemes, Scheme D and Scheme E which achieve the pairs  $(\frac{N}{3},1)$ and $(\frac{N^{2}}{2N-1}, \frac{N-1}{2N-1})$, respectively. Scheme D achieves memory-rate pair $(\frac{N}{3},1)$ using uncoded prefetching while Scheme E achieves memory-rate pair $(\frac{N^{2}}{2N-1}, \frac{N-1}{2N-1})$ using coded prefetching. In Example~\ref{Ex_exact_first_point}, we describe Scheme D for $N=3,K=2, M=1$. Then in Example~\ref{Ex_exact_second_point}, we describe  Scheme E for $N=3,K=2, M=2/5$.  General versions of both these schemes for $N>2$ and $K=2$  are provided in Subsection~\ref{sec_exact_achv}.

\input{exact.tex}

%% file: schemeC.tex
\subsection{Scheme C}
\label{Sec_random_cach_delivry}

For $K<N$, the broadcast in scheme A contains symbols which are not necessary for decoding, but are still broadcasted to preserve privacy. In this section, we propose Scheme C which gets rid of such redundant symbols using the idea of permuting the  broadcast symbols as in Scheme B, thus improving the memory-rate trade-off. In Theorem~\ref{Thm_PR_SR}, we give the memory-rate pairs  achievable using Scheme C.
 
\begin{theorem}
	\label{Thm_PR_SR}
There exists an $(N,K,M,R)$-private scheme with the following memory-rate pair:
	\begin{align}
	(M,R) = &\left(\frac{N\sum_{s=t}^{NK-1}{NK-1 \choose s-1}r^{NK-s-1}}{\sum_{s=t}^{NK-1}{NK \choose s} r^{NK-s-1}}, \frac{\sum_{s=t+1}^{NK}[{NK \choose s} -{NK-K\choose s}]r^{NK-s}}{\sum_{s=t}^{NK-1}{NK \choose s} r^{NK-s-1}} \right), \notag \\
	& \qquad \mbox{ for } t=\{1, \ldots, NK-1\}, \; r \in [1,N-1]. \label{Eq_Thm_PR_SR}
	\end{align}
\end{theorem}

 Note that for the memory-rate pairs in Theorem~\ref{Thm_PR_SR}, we have 2 free parameters $t$ and $ r$. By fixing the value of $r$, one can obtain a memory-rate curve by varying the value of $t$. We have observed through numerical computations that the memory-rate curve achieved for $r = r_1$ is better than that for $r=r_2$ if $r_1 > r_2$ (see Fig.~\ref{Fig_schemeC}).
The memory-rate curve for $r< N-1$, although empirically suboptimal compared to $r=N-1$, is useful in showing the order optimality result presented in Theorem~\ref{Thm_order}.

\begin{remark}
	\label{Rem_SchemeC}
	For $K< N$, Scheme B does not satisfy the stronger privacy metric in~\eqref{Eq_strng_priv_metric}. Since Scheme C builds on the ideas of Scheme B, it also does not satisfy this stronger privacy metric.
	The fact that Scheme B does not satisfy~\eqref{Eq_strng_priv_metric} can be intuitively observed from Example~\ref{Ex_no_coding}. For $K=2$, the privacy condition~\eqref{Eq_strng_priv_metric} is achieved if there is no leakage of privacy  after one user gets to know all the files. If one user has all the files, then she can easily verify that the part of the broadcast that she has not used for decoding  is some random bits or a part of a file. Thus, she can infer some knowledge about the demand of the other user in Scheme B.
	This was also observed in~\cite{Yan20} (see~\cite[Example~1]{Yan20}). 
\end{remark}

\begin{figure}
	\centering
	\begin{subfigure}{0.4\textwidth}
		\centering
		\includegraphics[scale=0.5]{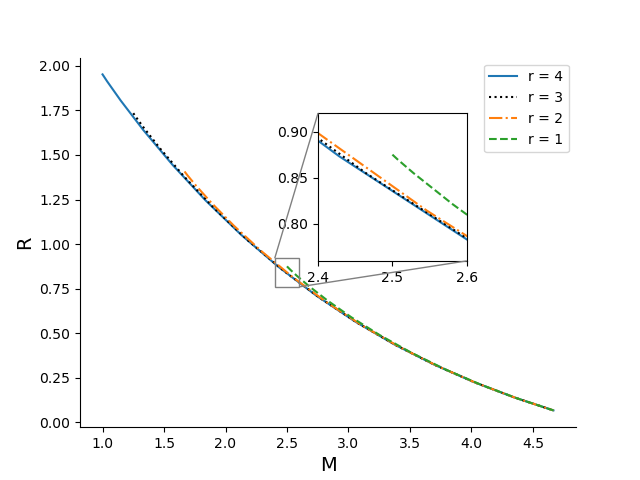}
	\end{subfigure}
\begin{subfigure}{0.4\textwidth}
		\centering
		\includegraphics[scale=0.5]{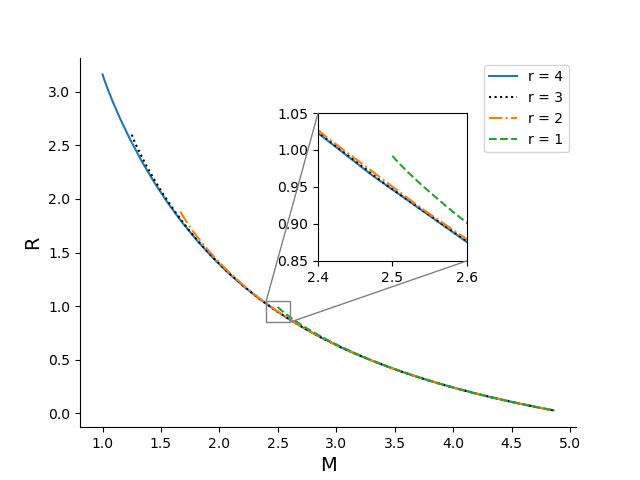}
\end{subfigure}
\caption{Memory-rate pairs in Theorem~\ref{Thm_PR_SR} are plotted for different values of $r$. The first figure is for $N=5,K=3$ and the second one is for $N=5,K=7$.}
	\label{Fig_schemeC}
\end{figure}

Next we illustrate Scheme C for $N=3, K=2$.

\begin{example}
	\label{Ex_PR_SR}
	Let us consider the demand-private coded caching problem for $N=3$ files and $K=2$ users. By choosing $r=2$ and $t=3$ in the expression for memory in Theorem~\ref{Thm_PR_SR}, we get $M=\frac{195}{116}$. The same parameters give $R=\frac{69}{116}$. Next we describe the scheme which achieves this memory-rate pair with $F=116l$ for some positive integer $l$. We partition each file $W_{i}, i \in [0:2]$ into $\sum_{j=t}^{NK-1} {NK \choose j} = \sum_{j=3}^{5} {6 \choose j} = 41$ segments of three different sizes. These segments are grouped into three groups such that all segments in one group have the same size.  The segments are labelled by some subsets of $[0:NK-1] = [0:5]$. The segments of $W_i$ are $W_{i,\cR}$; $\cR \subset [0:5], |\cR|=3,4,5$. These segments are of different sizes, and  these are grouped into 3 groups as 
	\begin{align*}
	\cT^{i}_{5} &=(W_{i,\cR})_{\cR \subset [0:5], |\cR| = 5}, \\ \cT^{i}_{4} &=(W_{i,\cR})_{\cR \subset [0:5], |\cR| = 4}, \\ 
	\cT^{i}_{3} &=(W_{i,\cR})_{\cR \subset [0:5], |\cR| = 3}.
	\end{align*}
	The size of segment $W_{i,\cR}, i\in[0:2]$ is chosen as follows:
	$$
	len(W_{i,\cR}) =  \left\{
	\begin{array}{lcl}
	l & \text{if } |\cR| = 5\\
	\\
	rl=2l & \text{if } |\cR| = 4\\
	\\
	r^2l=4l & \text{if } |\cR| = 3.\\
	\end{array}
	\right.
	$$
	Thus, each segment in $\cT^{i}_{5},\cT^{i}_{4}$ and $\cT^{i}_{3}$ has respectively \(l, 2l\) and \(4l\) bits. Then,  for all $i \in [0:2]$, we have
	\begin{align*}
	len(W_i) &= (|\cT^{i}_{5} |+|\cT^{i}_{4} |\times r+|\cT^{i}_{3}| \times r^2)l \\
	&= (6+15\times 2+20\times4)l \\
	&= 116l.
	\end{align*}
	
	\underline{Caching:} The cache content of user $k \in \{0,1\}$ is determined by the  key $S_k,k=0,1$ which is shared only between the server and user $k$. Shared key $S_k,k=0,1$ is  distributed as $S_k \sim unif\{[0:N-1]\} =unif\{[0:2]\}$. The cache contents of each user is grouped into three parts. The $j^{th}, j=1,2,3$ part of user $k\in \{0,1\}$ is denoted by $\cG_{k,j}$ and is shown in Table~\ref{Tab_cach_PR_SR}. Thus, the number of bits stored at one user is given by $3\left({5 \choose 4} +2\times {5 \choose 3} + 4 \times {5 \choose 2}\right)l = 195l$. Thus, we have $M=\frac{195}{116}$. Other than \(S_k\) the server also places some additional random keys of negligible size in the cache of user \(k\in \{0,1\}\). These will be used as keys for one-time pad in the delivery phase. 
	\begin{table}[h]
	\begin{center}
		\begin{tabular}{|c|c| } 
			\hline 
			$\cG_{k,1}$ & $(W_{i,\cR} | W_{i,\cR}  \in \cT^{i}_{5} \mbox{ and } S_k +3k \in \cR )_{i=0,1,2}$  \\
			\hline
			$\cG_{k,2}$ & $(W_{i,\cR} | W_{i,\cR}  \in \cT^{i}_{4} \mbox{ and } S_k +3k\in \cR )_{i=0,1,2} $ \\
			\hline
			$\cG_{k,3}$ & $(W_{i,\cR} | W_{i,\cR}  \in \cT^{i}_{3} \mbox{ and } S_k +3k \in \cR )_{i=0,1,2} $ \\
			\hline
		\end{tabular}
\end{center}
		\caption{Cache contents of user $k, k = 0,1$.}
		\label{Tab_cach_PR_SR}
	\end{table}
	
	\underline{Delivery:} In the delivery phase, for given demands $(D_{0},D_{1})$, we first construct an expanded demand vector $\bar{d}$ of length $6$ such that $ \bar{d} \in \DRS$ defined in Definition~\ref{Def_dmnd_subst}. The vector \(\bar{d} \) is given by
	$ \bar{d}= (\bar{d}^{(0)}, \bar{d}^{(1)})$, where $\bar{d}^{(k)}, k=0,1$ is obtained by applying $S_{k} \ominus D_{k}$ right cyclic shift to the vector $(0,1,2)$, where $\ominus$ denotes modulo $3$ subtraction. That is, for $k=0,1$, $d_i^{(k)} = i-(S_k-D_k) \mod 3$. 
	Having defined vector \(\bar{d}\), we now define symbols \(Y_{\cR}\) for \(\cR \subset [0:5]\) and \(|\cR| = 4,5,6\) as follows:
	\begin{align*}
	    Y_{\cR} = \bigoplus_{u \in \cR} W_{d_{u}, \cR \setminus \{u\}}
	\end{align*}
	where $d_{u}$ is the $u+1$-th item in $\bar{d}$. In particular, for \(\cR = [0:5]\), we have
	\begin{align*}
	    Y_{[0:5]} = \bigoplus_{u \in [0:5]} W_{d_{u}, [0:5] \setminus \{u\}}.
	\end{align*}
	Symbol \(Y_{[0:5]}\) as defined above is a part of the main payload in the broadcast transmission which needs \(l\) bits.

    To give the other parts of the broadcast, we define symbols \(W_{\cR}\) and \(V_{\cR}\) for  $\cR \subset [0:5] $ and $|\cR|=4,5$ as follows:
	\begin{align*}
	& W_{\cR} = (W_{0,\cR} \oplus W_{1,\cR}, W_{1,\cR} \oplus W_{2,\cR})
	\end{align*}
	and
	\begin{align*}
	V_{\cR} & = Y_{\cR} \oplus W_{\cR} .
	\end{align*}
	Note that for $|\cR|=4, W_{\cR}$ has two parts, each of length $2l$ bits, and $ Y_{\cR}$  has a length of $4l$ bits. We further define sets \(V_4\) and \(V_5\) as follows:
	\begin{align*}
	V_4 = \{V_{\cR} | \cR \cap \{S_0,S_1+3\} \neq \phi, |\cR| = 4\} 
	\end{align*}
	and
	\begin{align*}
	V_5 &= \{V_{\cR} | |\cR| = 5\}.
	\end{align*}
	Observe that \(V_4\) and \(V_5\) contain 14 symbols each of size \(4l\) bits and 6 symbols each of size \(2l\) bits,  respectively.
	The server picks permutation functions \(\pi_4(\cdot) \) and \(\pi_5(\cdot)\) uniformly at random  respectively from the symmetric group\footnote{A symmetric group defined over any set is the group whose elements are all the bijections from the set to itself, and whose group operation is the composition of functions.} of permutations of $[0:13]$ and $[0:5]$ and broadcasts $\pi_4(V_4)$ and $\pi_5(V_5)$. 
	The server does not fully reveal these permutation functions with any of the users. The position of any symbol \(V_{\cR} \in V_i\), \(i=4,5\) in  \(\pi_i(V_i)\) is privately conveyed to user \(k\), if and only if \(S_k +3k \in \cR\). 
	This private transmission of positions is achieved using one-time pads whose keys are deployed in the caches of respective users in the caching phase.  The main payload of the broadcast \(X'\) can be written as
	\[
	X'=(X_0,X_1,X_2)=(Y_{[0:5]},\pi_4(V_4),\pi_5(V_5)).
	\]
	Thus, the total number of transmitted bits  are
	$$
	(1+6 \times 2 + 14 \times 4)l = 69l.
	$$
	So, the rate of transmission is \(\frac{69}{116}.\) Note that \(X'\) is only the main payload. Along with \(X'\), the server also broadcasts some auxiliary transmission $J =(S_0 \ominus D_0, S_1 \ominus D_1, J') = (\bar{S} \ominus \bar{D}, J') $. Here, $J'$ contains the positions of various symbols in \(X_1\) and \(X_2\) encoded using one-time pad as discussed above. Thus, the complete broadcast transmission is $X = (X',J) $.
	
	\underline{Remark:} Here, note that $V_5$ contains all $V_{\cR}$ with $|\cR|=5$. However, $V_4$ does not contain all $V_{\cR}$ with $|\cR|=4$. For example, if $S_0=0$ and $S_1=0$, then $V$ does not contain $V_{\{1,2,4,5\}}$. This is similar to avoiding some redundant transmissions in the leader-based YMA scheme~\cite{Yu18} compared to the scheme in~\cite{Maddah14}. This is the main reason for getting lower rates using this scheme compared to the rates in Theorem~\ref{Corl_reduced_usrs}. 

	\underline{Decoding:} For user $k\in \{0,1\}$, let us  first consider the recovery of segments belonging to $\cT^{D_k}_{i}$, $i = 3,4$. This is done using symbols from $X_1$ and $X_2$. All symbols $W_{D_k, \cR} \in \cT^{D_k}_{i}$ such that $S_k + 3k \in \cR$ (all symbols in set $ \cG_{k,6-i}$) are cached at user \(k\). User $k$ decodes the remaining symbols in $\cT^{D_k}_{i}$, i.e., \(W_{D_k,\cR}\) such that \(|\cR|=i, S_k+3k \notin \cR\) and \(\cR \subset [0:5]\) as follows:
	\begin{align}
    \widehat{W}_{D_k,\cR} = V_{\cR^{+}} \oplus W_{\cR^{+}}
	\oplus \left(\bigoplus_{u\in {\cR}} W_{d_{u},\cR^{+} \setminus \{u\}}\right) \label{eq:decodform}
	\end{align} 
	where $\cR^{+} = \{S_k+3k\} \cup \cR$. Here, $V_{\cR^{+}}$ is a part of $\pi_{i+1}(V_{i+1})$ and its position in $\pi_{i+1}(V_{i+1})$ has been revealed to user \(k\) since $S_k+3k \in \cR^{+}$. The symbols $W_{\cR^{+}}$ and $W_{d_{u},\cR^{+} \setminus \{u\}}$ in \eqref{eq:decodform} can be recovered from her cache. Substituting for $V_{\cR^{+}}$ in \eqref{eq:decodform} yields
	\begin{align}
	\widehat{W}_{D_k,\cR} &= Y_{\cR^{+}} \oplus W_{\cR^{+}} \oplus W_{\cR^{+}}  \oplus \left(\bigoplus_{u\in {\cR}} W_{d_{u},\cR^{+} \setminus \{u\}}\right)\nonumber \\
	&= \bigoplus_{u \in \cR^{+}} W_{d_{u}, \cR^{+} \setminus \{u\}} \oplus \left(\bigoplus_{u\in {\cR}} W_{d_{u},\cR^{+} \setminus \{u\}}\right) \nonumber \\
	&= W_{d_{S_k+3k}, \cR } \nonumber \\
	&= W_{D_k, \cR }.  \label{eq:decodform2}
	\end{align}
	Since user \(k\) has all segments in $\cT^{D_k}_{3}$ and $\cT^{D_k}_{4}$, we consider the recovery of symbols in $\cT^{D_k}_{5}$.
	In the first part $\cG_{k,1}$ of cache, user \(k\) does not have one segment of  $\cT^{D_k}_{5}$, namely $W_{D_k, [0:5]\setminus \{S_k + 3k\}}$ . User $k$ decodes this segment as
	\begin{align*}
	    \widehat{W}_{D_k, [0:5]\setminus \{S_k + 3k\}} = Y_{[0:5]} \oplus \left(\bigoplus_{u\in {[0:5]\setminus \{S_k + 3k\}}} W_{d_{u},[0:5] \setminus \{u\}}\right).
	\end{align*}
	Observe that \(Y_{[0:5]}\) is broadcasted by the server while each symbol $W_{d_{u},[0:5] \setminus \{u\}}$ is a part of $\cG_{k,1}$, and hence a part of the cache of user \(k\). Thus, user \(k\) can compute $\widehat{W}_{D_k, [0:5]\setminus \{S_k + 3k\}}$. Using~\eqref{eq:decodform2}, it can be shown that $\widehat{W}_{D_k, [0:5]\setminus \{S_k + 3k\}} = W_{D_k, [0:5]\setminus \{S_k + 3k\}} $. Thus, user $k$ can retrieve all symbols belonging to each of the three groups of file $W_{D_k}$ and she can recover this file by concatenating these symbols.
	 
	\underline{Privacy:} To show the demand-privacy for user $k \in \{0,1\}$, we first define \(\tilde{k} = (k + 1)\) mod \(2\). Since $I(D_{\tilde{k}};Z_{k},D_{k}) = 0$, the privacy condition $I(D_{\tilde{k}};X,Z_k,D_k) = 0$ follows by showing that $I(X;D_{\tilde{k}}|Z_{k},D_{k}) = 0$. To that end, we
	divide all symbols that are a part of the main payload into two sets, $X'_k$ and $\tilde{X}'_k$ which are defined as follows:
	\begin{align*}
	X'_{k} & = \{Y_{[0:NK-1]}\} \cup \{V_{\cR}| S_{k}+3k \in \cR, V_{\cR} \in X' \}, \\  
	\tilde{X}'_k & = X'\setminus X'_{k} .
	\end{align*}
	Note that the positions in $X'$ of all symbols belonging to $X'_{k}$ is known to user $k$ while the positions of symbols belonging to $\tilde{X}'_k$ are not known. It can be shown that all symbols in $\tilde{X}'_k$ appear like a sequence of random bits to user $k$. This is because for some set $\cR$, $\cR \subset [0:5],|\cR| = 4,5$, the server broadcasts $V_{\cR}$ instead of $Y_{\cR}$. The symbol $W_{\cR}$ essentially hides the message $Y_{\cR}$ from all users that do not belong to set $\cR$. Further, it can also be shown that 
	\begin{align}
	H(X'_{k}|W_{D_k},Z_k,\bar{S}\ominus \bar{D})=0. \label{eq:dummypriv} 
	\end{align}
	It is easy to see that, $(W_{D_k},Z_k,\bar{S}\ominus \bar{D})$ does not reveal any information about $D_{\tilde{k}}$ which in combination with \eqref{eq:dummypriv} ensures privacy. 
\end{example} 

%% file: exact.tex
\begin{example}
	\label{Ex_exact_first_point}
 We describe  Scheme D for $N=3$ and $K=2$ which achieves rate \( 1\) for \(M=\frac{N}{3} = 1\). File $W_{i}, i \in [0:2]$ is divided into 3 disjoint  parts of equal size, i.e., \mbox{$W_{i} = (W_{i,0},W_{i,1},W_{i,2})$}. 
 
 \underline{Caching: }The server picks 2 independent permutations \(\pi_{0}\) and \(\pi_{1}\) uniformly at random from the symmetric group of permutations of [0:2]. The server places \(\pi_{0}(W_{0,0},W_{1,0},W_{2,0})\)  and \(\pi_{1}(W_{0,1},W_{1,1},W_{2,1})\) in the caches of user 0 and user 1, respectively. Each of these permutation functions \(\pi_{0}\) and \(\pi_{1}\) are unknown to both the users. Some additional random bits are also shared with each user in the caching phase.
 
 \underline{Delivery:} The server picks permutation  \(\pi_{2}\)  uniformly at random from the symmetric group of permutations of [0:2] which is independent of \(\pi_{0}\), \(\pi_{1}\).
 The main payload \(X'\) is given by
 \begin{align*}
 X' = 
 \begin{cases}
 \pi_{2}(W_{D_{0},1} \oplus W_{D_{1},0}, W_{D_{0},2}, W_{D_{1},2}) \qquad \qquad \qquad \text{if } D_{0}\neq  D_{1}\\
\pi_{2}(W_{D_{0},1} \oplus W_{m,0}, W_{D_{0},2}, W_{D_{1},0} \oplus W_{m,1})  \qquad \quad \text{if } D_{0} =D_{1}
 \end{cases}
 \end{align*}
where \(m = (D_{0}+1)\) mod $3$. 
 \comment{
 The main payload $X'$ alone is not enough for decoding because none of the permutation functions are known to them.
}
To enable decoding at each user, the server also transmits some auxiliary transmission $J=(J_1,J_2,J_3)$ of negligible rate. Each $J_j,j = 1,2,3$ can be further divided into 2 parts, i.e., $J_j =(J_{j,0},J_{j,1})$, where $J_{j,k},k\in \{0,1\}$ is meant for user $k$ for all $j=1,2,3$. Using a one-time pad which uses the pre-shared random bits, the server ensures that $J_{j,k}$ can be decoded only by user \(k\) and it is kept secret from the other user.  These parts are used as follows:
 \begin{enumerate}
     \item  $J_{1,k}$ conveys the position of \(W_{D_{k},k}\) in user $k$'s cache.
     \item  $J_{2,k}$ gives the positions of the coded and uncoded parts of \(X'\) involving $W_{D_k}$ to user $k$. Specifically, $J_{2,k}$ reveals the positions of \(W_{D_{0},1} \oplus W_{D_{1},0}\) and \(W_{D_{k},2}\) to user $k$ when \(D_{0} \neq D_{1}\), and the positions of \(W_{D_{k},\tilde{k}} \oplus W_{m,k}\) and \(W_{D_{k},2}\) when \(D_{0} = D_{1}\), where $\tilde{k} = (k+1)$ mod 2.
     \item  $J_{3,k}$ discloses the position of \(W_{D_{\tilde{k}},k}\) if \(D_{0} \neq D_{1}\) and \(W_{m,k}\) if \(D_{0} = D_{1}\) in her cache to user $k$.
 \end{enumerate}
 
 \underline{Decoding:} User $k$ decodes \(W_{D_{k}}\) as follows. \(W_{D_{k},k}\) can be obtained from the cache since  she knows its position from $J_{1,k}$. User $k$ recovers \(W_{D_{k},2}\) from the delivery since she knows its position in \(X'\) from $J_{2,k}$. The remaining segment \(W_{D_{k},\tilde{k}}\) is available in coded form in $X'$. The segment that \(W_{D_{k},\tilde{k}}\) is XOR-ed with, is available in the cache of user $k$, and its position in the cache is revealed by $J_{3,k}$. Thus, user $k$ retrieves all three segments of file $W_{D_k}$.
 
 \underline{Privacy:} Now we give an outline of how $D_{1}$ is kept secret from user 0. From the transmission, we can observe that for both the cases, i.e., \(D_{0} \neq D_{1}\) and \(D_{0} = D_{1}\), user 0 receives \(W_{D_{0},2}\) in the uncoded form and \(W_{D_{0},1}\) coded with another symbol. Also, in both the cases, the remaining symbol is like a sequence of \(\frac{F}{3}\) random bits to user 0, because it contains either \(W_{D_{1},2}\) or \(W_{m,1}\) which she doesn't  have access to. Thus, even though the structure of the broadcast is different in the two cases, any user cannot  differentiate between them. Further, given \(J_{1,0}\), any of the remaining 2 symbols can occupy the remaining 2 positions in the cache with equal likelihood. Thus, although user 0 can use one of these symbols, i.e., the  symbol XOR-ed with \(W_{D_{0},1}\), for decoding with the help of \(J_{3,0}\), the identity of the symbol is not known because \(J_{3,0}\) only discloses the position of that symbol in user 0's cache. Due to the symmetry of the scheme, similar privacy arguments apply for user 1.
\end{example}

\begin{example}
	\label{Ex_exact_second_point}
	Now we describe Scheme E for $N=3$ and $K=2$ which achieves rate \( \frac{N-1}{2N-1} = \frac{2}{5}\), for \(M=\frac{N^2}{2N-1} = \frac{9}{5}\). File $W_{i}, i \in [0:2]$ is partitioned into \(2N-1 = 5\)  parts of equal size, i.e., $W_{i} = (W_{i,1},W_{i,2},W_{i,3},W_{i,4},W_{i,5})$. 
	We encode each file \(W_{i}\) using a \((3N-2, 2N-1) = (7,5)\) MDS code (i.e. each file is split into \(5\) pieces of size \(\frac{F}{5}\) bits each, which are then encoded using \((7,5)\) MDS code such that each of \(7\) MDS coded symbols has \(\frac{F}{5}\) bits). Each file can be reconstructed using any \(5\) MDS symbols. The \(7\) symbols of file \(W_{i}\) are denoted by \(F_{i,0}, F_{i,0,0},F_{i,0,1},F_{i,1,0},F_{i,1,1},F_{i,2,0}\) and \(F_{i,2,1}\).	Further, we define tuples \(\cL_{0},\cL_{1}\) and \(\cL_{2}\) as follows:
	$$   \cL_{j} = (F_{i,0}, F_{i,j,0}, F_{i,j,1})_{i \in [0:2]},   \quad \forall j \in [0:2].\\
	$$
	
	\underline{Caching:} The server picks 2 independent permutation functions \(\pi_{0}\) and \(\pi_{1}\) uniformly at random from the symmetric group of permutations of [0:8]. The server also picks a random number \(U_{0}\) which is uniformly distributed in \(\{0,1,2\}\) and places \(\pi_{0}(\cL_{U_0})\) in the cache of user 0. Similarly, the server picks random number \(U_1\) which is uniformly distributed in \(\{0,1,2\}\backslash \{U_0\}\) and places \(\pi_{1}(\cL_{U_1})\) in user 1's cache.
	Similar to Scheme D, each of these permutation functions \(\pi_{0}\) and \(\pi_{1}\) are private to both the users. Unlike the permutation functions, \(U_k\) is shared with user $k \in \{0,1\}$ by placing it in her cache and thus kept secret from the other user.
	
	\underline{Delivery:} The main payload, \(X'\) is given by
	$$
	X' = \left\{
	\begin{array}{lcl}
	(F_{D_{0},U_1,0} \oplus F_{D_{1},U_0,0}, F_{D_{0},U_1,1} \oplus F_{D_{1},U_0,1}) & \text{if } D_{0}\neq D_{1}\\
	\\
	(F_{D_{0},V,0} \oplus F_{m_{1},0}, F_{D_{0},V,1} \oplus F_{m_{2},0}) & \text{if } D_{0} = D_{1}
	\end{array}
	\right.
	$$
	where, \(m_{t} = (D_{0}+t)\) mod \(3\), and \(V = [0:2]\backslash\{U_0,U_1\}\).
	
	\comment{
		As in the previous example, the main payload $X'$ alone is not enough for decoding because the permutation functions are private. 
	}
	To enable decoding at each user, the server also transmits some auxiliary transmission $J=(J_1,J_2,J_3)$ of negligible rate. Each $J_j,j = 1,2,3$ can be further divided into 2 parts, i.e., $J_j =(J_{j,0},J_{j,1})$. where $J_{j,k},k\in \{0,1\}$ is meant for user $k$ for all $j=1,2,3$. Using a one-time pad, the server ensures that $J_{j,k}$ can be decoded only by user \(k\) and it is kept secret from the other user.  These parts are used as follows:
	\begin{enumerate}
		\item  $J_{1,k}$ conveys the positions of \(F_{D_{k},0},F_{D_{k},U_k,0}\) and \(F_{D_{k},U_k,1}\) in user $k$'s cache.
		\item  $J_{2,k}$ discloses the positions of symbols \(F_{D_{\tilde{k}},U_k,0}\)
		and \(F_{D_{\tilde{k}},U_k,1}\) in user $k$'s cache 
		if \(D_{0} \neq D_{1}\), where \(\tilde{k} = (k+1)\) mod 2. If \(D_{0} = D_{1}\), then $J_{2,k}$ reveals the positions of \(F_{m_{1},0}\)
		and \(F_{m_{2},0}\) in the cache for user \(k\).
		\item  $J_{3,k}$ gives the value of the random variable \(T_{k}\) which takes the value \(U_{\tilde{k}}\) if \(D_{0}\neq D_{1}\), and \(V\) if \(D_{0} = D_{1}\).
	\end{enumerate}
	
	\underline{Decoding:} Let us consider the decoding of \(W_{D_{k}}\) at user $k$. First, the user retrieves \(F_{D_{k},0},F_{D_{k},U_k,0}\) and \(F_{D_{k},U_k,1}\) directly from the cache since its positions are obtained from $J_{1,k}$. The positions of \(F_{D_{\tilde{k}},U_k,0}\) and \(F_{D_{\tilde{k}},U_k,1}\) when  \(D_{0} \neq D_{1}\), and the positions of \(F_{m_{1},0}\) and \(F_{m_{2},0}\) when \(D_{0} = D_{1}\), are available through $J_{2,k}$. Thus, using \(X'\) user $k$ can recover \(F_{D_{k},U_{\tilde{k}},0}\) and \(F_{D_{k},U_{\tilde{k}},1}\), or \(F_{D_{k},V,0}\) and \(F_{D_{k},V,1}\), accordingly. Note that because user $k$ does not know whether \(D_{0} \neq D_{1}\) or \(D_{0} = D_{1}\), she also does not know whether she has recovered \(F_{D_{k},U_{\tilde{k}},0}\) and \(F_{D_{k},U_{\tilde{k}},1}\) or \(F_{D_{k},V,0}\) and \(F_{D_{k},V,1}\). This information is available through $J_{3,k}$ which gives the value of \(T_{k}\). Thus, user $k$ has knowledge of \(F_{D_{k},T_{k},0}\) and \(F_{D_{k},T_{k},1}\). Since \(T_{k} \neq U_{k}\), user $k$ has access to 5 distinct symbols of the MDS code namely, \(F_{D_{k},0},F_{D_{k},U_k,0},F_{D_{k},U_k,1},F_{D_{k},T_{k},0}\) and \(F_{D_{k},T_{k},1}\). Thus, the file \(W_{D_{k}}\) can be retrieved.
	
	\underline{Privacy:} Now we describe how $D_{1}$ remains private to user 0. It is important to note that, from the knowledge of the tuple \((U_0,D_{0},T_{0})\), user 0 cannot find out the demand of user 1. For example, if \((U_0,D_{0},T_{0}) = (0,0,1)\), \((U_1,D_{1})\) can take 3 distinct values namely, \((1,1)\), \((1,2)\) and \((2,0)\), all 3 possibilities being equally likely, implying that even after knowing \((U_0,D_{0},T_{0})\), \(D_{1}\) can take all 3 values with equal likelihood. 
	Since \(\pi_0\) is not shared with user 0, and since the other auxiliary transmissions ($J_{1,0}$ and $J_{2,0}$) reveal only the positions of the 5 relevant symbols in user 0's cache, they are independent of \(D_1\) and can take any of the \(9 \choose {5}\) possible values depending on \(\pi_0\). Hence, it is clear that they do not reveal the demand of user 1.  These arguments are crucial in ensuring the demand-privacy of user 1 and similar arguments hold for the other user as well due to the symmetry of this scheme.
\end{example}

%% file: proofs.tex
\subsection{Proof of Theorem~\ref{Thm_genach}}
\label{Sec_proof_Thm_genach}

Let us consider any $(N,NK,M,R)$ \DRS-non-private scheme. 
Let $C_k^{(np)}, k\in [0:NK-1]$ be the cache encoding functions, $E^{(np)}$ be the
broadcast encoding function, and $G_k^{(np)}, k\in [0:NK-1]$ be the decoding
functions for the given $(N,NK,M,R)$ \DRS-non-private scheme.
We now present a construction of an $(N,K,M,R)$-private scheme from
the given $(N,NK,M,R)$ \DRS-non-private scheme. 

\underline{Caching:} For $k\in [0:K-1]$ and $S_k\in [0:N-1]$, the $k$-th user's cache
encoding function is given by
\begin{eqnarray}
&& C_k (S_k, \bar{W}) := C_{kN+S_k}^{(np)} (\bar{W}).
\label{eq:cachesch}
\end{eqnarray}
The $k$-th user's cache encoding function is taken to be the same
as that of the $S_k$-th user in the $k$-th stack in the corresponding $(N,NK)$
caching problem.
The cache content is given by $Z_k=(C_k (S_k, \bar{W}), S_k)$.

\underline{Delivery:} To define the broadcast encoding, we need some
new notations and definitions.
Let $\Psi: [0:N-1]^N \rightarrow [0:N-1]^N$ denote the cyclic shift operator,
such that $\Psi (t_1,t_2,\ldots,t_N)=(t_N,t_1,\ldots,t_{N-1})$.
Let us denote a vector $\mathbb{I}:=(0,1,\ldots,N-1)$.
Let us also define 
\begin{align*}
\bar{S}\ominus \bar{D}:=(S_0\ominus D_0, S_1\ominus D_1, \ldots, S_{K-1}\ominus D_{K-1})
\end{align*}
where $S_k\ominus D_k$ denotes the difference of $S_k$ and $D_k$ modulo $N$.
For a given $\bar{D}\in [0:N-1]^K$, we define an expanded demand vector
for the non-private problem as:
\begin{align*}
\bar{D}^{(np)}(\bar{D},\bar{S})=(\Psi^{S_0\ominus D_0}(\mathbb{I}),
\ldots,\Psi^{S_{K-1}\ominus D_{K-1}}(\mathbb{I}))
\end{align*}
where $\Psi^i$ denotes
the $i$-times cyclic shift operator.

The broadcast encoding function for the $(N,K,M,R)$-private scheme is defined by
\begin{eqnarray}
&& E (\bar{W},\bar{D},\bar{S}) := E^{(np)}(\bar{W},\bar{D}^{(np)}(\bar{D},\bar{S})).
\label{eq:trencsch}
\end{eqnarray}
Let us denote $X_1=E (\bar{W},\bar{D},\bar{S})$. In the private scheme,
the server transmits $X=(X_1, \bar{S}\ominus \bar{D})$.

\underline{Decoding:}
User $k\in[0:K-1]$ uses the decoding function of the $(kN+S_k)$-th user in the
non-private scheme, i.e., 
\begin{align}
G_k(D_k, S_k, \bar{S}\ominus \bar{D}, X_1,Z_k)
&:=G_{kN+S_k}^{(np)}(\bar{D}^{(np)}(\bar{D},\bar{S}), X_1, Z_k).
\label{eq:decsch}
\end{align}
Here the decoder computes $\bar{D}^{(np)}(\bar{D},\bar{S})$ from
$\bar{S}\ominus \bar{D}$.

From \eqref{eq:cachesch}, \eqref{eq:trencsch}, and \eqref{eq:decsch}, it is clear that the decoder of the $k$-th user outputs the same file requested by the $S_k$-th virtual user of the $k$-th stack in the non-private scheme.
The index of the output file is the $(kN+S_k)$-th component in
$\bar{D}^{(np)}(\bar{D},\bar{S})$, i.e., $S_k\ominus (S_k\ominus D_k)=D_k$. 
Thus, the $k$-th user recovers its desired file.

\underline{Proof of privacy:}
The proof of privacy essentially follows from the fact that $S_i$ acts as
one-time pad for $D_i$  which prevents any user $j \neq i$ getting any
information about $D_i$.
We now show that the derived $(N,K,M,R)$-private scheme satisfies the privacy condition~\eqref{Eq_instant_priv}. 
First we show that $I(\ND{k}; Z_k,D_k,\TRS | \bar{W}) =0$.
\begin{align}
 I(\ND{k}; Z_k,D_k,\TRS | \bar{W}) 
& = H(Z_k,D_k,\TRS | \bar{W}) - H(Z_k,D_k,\TRS | \bar{W}, \ND{k})  \nonumber \\
& \overset{(a)}{=}  H(S_k,D_k, \bar{S}\ominus \bar{D}, \bar{D}^{(np)}(\bar{D},\bar{S})  | \bar{W}) - H(S_k,D_k, \bar{S}\ominus \bar{D}, \bar{D}^{(np)}(\bar{D},\bar{S}) | \bar{W},  \ND{k}) \nonumber  \\
& \overset{(b)}{=}  H(S_k,D_k, \bar{S}\ominus \bar{D}  | \bar{W}) - H(S_k,D_k, \bar{S}\ominus \bar{D} | \bar{W},  \ND{k})  \nonumber\\
& \overset{(c)}{=}  H(S_k,D_k,\bar{S}\ominus \bar{D}) - H(S_k,D_k,\bar{S}\ominus \bar{D}  |  \ND{k}) \nonumber \\
& \overset{(d)}{=}  H(S_k,D_k,\bar{S}\ominus \bar{D}) - H(S_k,D_k, \bar{S}\ominus \bar{D}) \nonumber \\
& =0. \label{Eq_privcy4}
\end{align}
Here, $(a)$ follows since $X=(X_1, \bar{S} \ominus \bar{D}), Z_k = (C_k(S_k,\bar{W}),S_k)$, and also 
due to~\eqref{eq:trencsch}. In $(b)$, we used that $H(\bar{D}^{(np)}(\bar{D},\bar{S})| \bar{S}\ominus \bar{D}) =0$, and $(c)$ follows since $(S_k,D_k, \bar{S}\ominus \bar{D}, \ND{k})$ is independent of $\bar{W}$.
We get $(d)$ since $S_i \ominus D_i$ is independent of $D_i$ for all $i \in [0:K-1]$.
Using the fact that demands and files are independent, 
we get 
the following from~\eqref{Eq_privcy4} 
\begin{align*}
I(\ND{k}; Z_k,D_k,\TRS , \bar{W}) & = I(\ND{k};  \bar{W}) + I(\ND{k}; Z_k,D_k,\TRS | \bar{W}) \\
& = 0.
\end{align*}
This shows the derived scheme satisfies the privacy condition 
$I(\ND{k}; Z_k,D_k,\TRS) =0$.

The size of the cache in the $(N,K,M,R)$-private scheme differs only by the size of the shared key from the $(N,NK,M,R)$ \DRS-non-private scheme. For large enough file size $2^F$, this difference is negligible. Furthermore, we can observe that the rate of transmission in $(N,K,M,R)$-private scheme is the same as  that of the $(N,NK,M,R)$ \DRS-non-private scheme. This proves Theorem~\ref{Thm_genach}.

\subsection{Proof of Lemma~\ref{Thm_reduced_usrs}}
\label{Sec_proof_reducd_usrs}

Consider any
\begin{align*}
(M,R)= \left(\frac{Nr}{NK-K+1} , \frac{{NK-K+1\choose r+1}-{NK-K+1-N \choose r+1}}{{NK-K+1\choose r}}\right), \quad \text{ for } r \in \{0,1,\ldots,NK-K\}
\end{align*}
which is achievable for $N$ files and $NK-K+1$ users by the YMA scheme.  We will construct  an $(N,NK,M,R)$ \DRS-non-private scheme with these $(M,R)$ pairs. We denote the set of $NK$ users by $\cU = \cU_1\cup\cU_2$,
where 
\begin{align*}
\cU_1 := \{u'_1,u'_2, \ldots, u'_{K-1}\},
\end{align*}
and
\begin{align*}
\cU_2 := \{u_0,u_1, \ldots, u_{NK-K}\}.
\end{align*}
The users are partitioned into $K$ subsets/stacks $\cK_i,i\in [0:K-1]$ as given below.
\begin{align*}
\cK_0 &= \{u_j | j \in [0:N-1] \},\\
\cK_i &=\{u_j | i(N-1)+1\leq j \leq (i+1)N-1 \} \cup \{u'_i\} \quad \mbox{ for } 1 \leq i\leq K-1.
\end{align*}
The stacks of users are shown in Fig.~\ref{Fig_users}. 
We denote the demand of $u_i \in \cU_2$ by $d_i$, and the demand of $u'_i\in \cU_1$ by $d'_i$. Similarly, we denote the cache of 	$u_i \in \cU_2$ by $Z_i$, and the cache of user $u'_i\in \cU_1$ by $Z'_{i}$.
 We further define 
\begin{align*}
\cK'_i &:= \{j | u_j \in \cK_i \} \quad \text{for } i = 0,1,2,\ldots,K-1.\\
\cV_i  & := \cK'_{0} \cup \cK'_{i}, \quad \text{for } i = 1,2,\ldots,K-1.
\end{align*}

\begin{figure}[htb]
	\centering
	\includegraphics[scale=0.5]{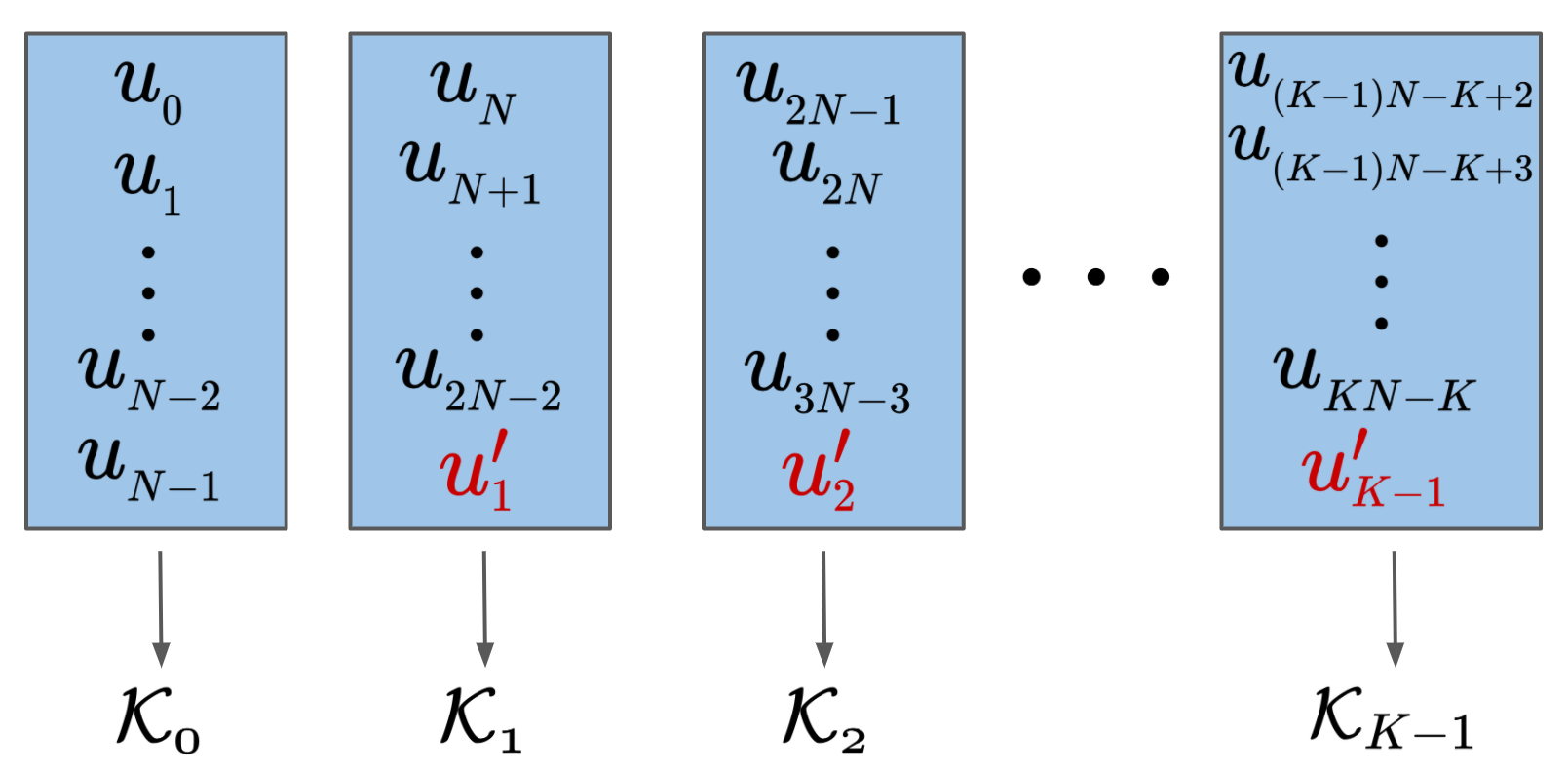}
	\caption{The stacks of users are shown vertically. The users from $\cU_2$ are shown in black color, whereas those from $\cU_1$ are shown in red color.}
	\label{Fig_users}
\end{figure}

\underline{Caching:} 
The users in $\cU_2$ has  the same prefetching as that of the users in the YMA scheme. 
Let $Z_m^{\text{YMA}},m=0,\ldots,NK-K$ denote the cache content of $m$-th user in the YMA scheme. The cache content of user $ u_j \in \cU_2$ is given by
\begin{align*}
Z_j = Z_{j}^{\text{YMA}}.
\end{align*}
To explain the prefetching of users in $\cU_1$, let
\begin{align*}
\cT := [0:NK-K].
\end{align*}
 In the YMA scheme, 
each file is divided into ${NK-K+1\choose r}$ subfiles and file $W_{i}, i \in [0:N-1]$ is given by
$$W_{i}=(W_{i,\cR})_{\cR \subset \cT, |\cR| = r}.$$
For $ \cS \subset \cT$ such that $|\cS| = r-1$, we  define  
 \begin{align}
 Z^{j}_{i, \cS} := \bigoplus_{u \in \cV_i  \setminus \cS \cap \cV_i } W_{j, \{u\} \cup \cS},  \text{ for } i \in [0:K-1] \setminus \{0\}, \; j \in [0:N-1]. \label{Eq_cach_subset}
 \end{align}
The cache content of user $u'_i \in \cU_1$ is given by
\begin{align*}
Z'_{i} = \left(Z^{j}_{i, \cS} | \cS \subset \cT \setminus \{0\}, |\cS| = r-1\right)_{ j \in [0:N-1]}.
\end{align*}
Since the size of one $Z^{j}_{i, \cS}$ is $\frac{F}{{NK-K+1\choose r}}$
and since there are ${NK-K\choose r-1}$ possible sets $\cS$, the number of bits stored at user $u'_i$  is given by
\begin{align}
 len(Z'_{i}) & = \frac{N F{NK-K\choose r-1}}{{NK-K+1\choose r}} \nonumber \\
 &= \frac{NFr}{NK-K+1} \nonumber\\
 & = FM.\nonumber
\end{align}

\underline{Delivery:} For a given $\bar{d} \in \DRS$, let $\bar{d}_1 = (d'_{i})_{u'_i \in \cU_1}$ and $\bar{d}_2 = (d_i)_{u_i \in \cU_2}$. The server chooses the transmission of the YMA scheme for $\cU_2$ under demand $\bar{d}_2$.
 The broadcast transmission is described using
\begin{align}
Y_{\cR}=\bigoplus_{u \in \cR} W_{d_{u}, \cR \setminus \{u\}}, \quad  \cR \subset \cT  \text{ such that } |\cR| = r+1. \label{Eq_YMA_symbols}
\end{align}
The broadcast transmission $X$ is given by 
\begin{align*}
X = \{Y_{\cR}| \cR \cap \cK_{0} \neq \phi \}.
\end{align*}
Note that in the above broadcast transmission, stack of users $\cK_0$ corresponds to the set of  ``leaders" described in the YMA scheme since each file $W_{i}, i \in [0:N-1]$ is demanded by exactly one user in $\cK_{0}$.
The size of each symbol $Y_{\cR} \in X$ is $\frac{F}{{NK-K+1\choose r}}$ and $X$ contains ${NK-K+1 \choose r+1} - {NK-K+1-N\choose r+1}$ such symbols. Thus, 
\begin{align*}
len(X) & = \frac{\left({NK-K+1 \choose r+1} - {NK-K+1-N\choose r+1}\right)F}{{NK-K+1\choose r}} \\
& = RF.
\end{align*}

\underline{Decoding: }  For all user in $\cU_2$, the decodability follows from the decodability of the YMA scheme. The decoding of users in $\cU_1$ is as follows.

\begin{remark}
	\label{Rem_SchemeA}
From the YMA scheme, we know that all symbols $Y_{\cR}$ such that $\cR \subset [0:NK-1]$ and $|\cR| = r+1$ can be recovered from $X$. Similarly the following lemma states that although all symbols $Z^{j}_{i, \cS}$ given in \eqref{Eq_cach_subset} are not a part of the cache of user $u'_{i} \in \cU_{1}$, each of these symbols can still be recovered from the cache contents.
\end{remark}

\begin{lemma}
 \label{Lem_subset_users}
	For $i \in [0:K-1]\setminus \{0\}$, all symbols $Z^{j}_{i, \cS} $, where $\cS \subset \cT , \;|\cS| = r-1$, and $j\in[0:N-1]$, can be recovered from the cache content $Z'_{i}$ of user $u'_i$.
\end{lemma}
\begin{proof}
	See Appendix~\ref{Sec_proof_subset_lem}.
\end{proof}
 Next we show how user $u'_i \in \cU_1$ obtains $W_{d'_{i}, {\cR}}$ for all $\cR$ satisfying $\cR \subset \cT$ and $|\cR| = r$. Each $W_{d'_{i}, {\cR}}$ can be written as 
\allowdisplaybreaks
\begin{align}
W_{d'_{i}, {\cR}} & \overset{\mathrm{(a)}}{=} \bigoplus_{u \in \cV_{i}} W_{d_{u},\cR} \nonumber \\ 
&= \bigoplus_{u \in \cV_i \cap \cR} W_{d_{u},\cR} \oplus \bigoplus_{u \in \cV_i \backslash \cV_i \cap \cR} W_{d_{u},\cR} \nonumber \\
&\overset{\mathrm{(b)}}{=} \bigoplus_{u \in \cV_i \cap \cR} W_{d_{u},\cR} \oplus \bigoplus_{u \in \cV_i \backslash \cV_i \cap \cR} \bigg\{ Y_{\{u\}  \cup {\cR}} \oplus \bigoplus_{t\in {\cR}} W_{d_{t}, \{u\}  \cup {\cR} \backslash\{t\}} \bigg\} \nonumber\\
&= \bigoplus_{u \in \cV_i \cap \cR} W_{d_{u},\cR} \oplus \bigoplus_{u \in \cV_i \backslash \cV_i \cap \cR} Y_{\{u\} \cup {\cR}} \oplus \bigoplus_{u \in \cV_i \backslash \cV_i \cap \cR} \; \bigoplus_{t\in {\cR}} W_{d_{t}, \{u\}  \cup {\cR} \backslash\{t\}}  \nonumber\\
&= \bigoplus_{u \in \cV_i \cap \cR} W_{d_{u},\cR} \oplus \bigoplus_{u \in \cV_i \backslash \cV_i \cap \cR} Y_{\{u\} \cup {\cR}} \oplus \bigoplus_{t \in \cV_i \backslash \cV_i \cap \cR} \; \bigoplus_{u\in {\cR}} W_{d_{u}, \{t\}\cup {\cR} \backslash\{u\}}  \nonumber\\
&= \bigoplus_{u \in \cV_i \cap \cR} W_{d_{u},\cR} \oplus \bigoplus_{u \in \cV_i \backslash \cV_i \cap \cR} Y_{\{u\} \cup {\cR}}  \nonumber \\
& \quad \oplus \bigoplus_{t \in \cV_i \backslash \cV_i \cap \cR} \; \bigg\{\bigoplus_{u\in {\cV_i \cap \cR}} W_{d_{u}, \{t\}\cup {\cR} \backslash\{u\}} \oplus \bigoplus_{u\in {\cR \backslash \cV_i \cap \cR}} W_{d_{u}, \{t\}\cup {\cR} \backslash\{u\}} \bigg\}  \nonumber\\
&= \bigoplus_{u \in \cV_i \cap \cR} W_{d_{u},\cR} \oplus \bigoplus_{u \in \cV_i \backslash \cV_i \cap \cR} Y_{\{u\} \cup {\cR}} \oplus \bigoplus_{t \in \cV_i \backslash \cV_i \cap \cR} \; \bigoplus_{u\in {\cV_i \cap \cR}} W_{d_{u}, \{t\}\cup {\cR} \backslash\{u\}} \nonumber \\
& \quad \oplus \bigoplus_{t \in \cV_i \backslash \cV_i \cap \cR} \; \bigoplus_{u\in {\cR \backslash \cV_i \cap \cR}} W_{d_{u}, \{t\}\cup {\cR} \backslash\{u\}} \nonumber\\
&= \bigoplus_{u \in \cV_i \cap \cR} W_{d_{u},\cR}  \oplus \bigoplus_{u\in {\cV_i \cap \cR}} \; \bigoplus_{t \in \cV_i \backslash \cV_i \cap \cR} W_{d_{u},\{t\} \cup {\cR} \backslash\{u\}} \oplus \bigoplus_{ u \in \cV_i \backslash \cV_i \cap \cR} Y_{\{u\} \cup {\cR}}  \nonumber \\
& \quad \oplus \bigoplus_{u\in {\cR \backslash \cV_i \cap \cR}} \; \bigoplus_{t \in \cV_i \backslash \cV_i \cap \cR} W_{d_{u},\{t\} \cup {\cR} \backslash\{u\}} \nonumber\\
&= \bigoplus_{u \in \cV_i \cap \cR} \bigg\{ W_{d_{u},\cR}  \oplus \bigoplus_{t \in \cV_i \backslash \cV_i \cap \cR} W_{d_{u},\{t\}\cup {\cR} \backslash\{u\}}\bigg\} \oplus \bigoplus_{u \in \cV_i \backslash \cV_i \cap \cR} Y_{\{u\} \cup {\cR}} \nonumber \\
& \quad \oplus \bigoplus_{u\in {\cR \backslash \cV_i \cap \cR}} \; \bigoplus_{t \in \cV_i \backslash \cV_i \cap (\cR\backslash \{u\})} W_{d_{u}, \{t\}\cup {\cR} \backslash\{u\}} \nonumber\\
&= \bigoplus_{u \in \cV_i \cap \cR} \bigg\{\bigoplus_{t \in \cV_i \backslash \cV_i \cap (\cR\backslash \{u\})} W_{d_{u}, \{t\}\cup {\cR} \backslash\{u\}}\bigg\} \oplus \bigoplus_{u \in \cV_i \backslash \cV_i \cap \cR} Y_{\{u\} \cup {\cR}} \nonumber \\
& \quad \oplus \bigoplus_{u\in {\cR \backslash \cV_i \cap \cR}} \; \bigoplus_{t \in \cV_i \backslash \cV_i \cap (\cR\backslash \{u\})} W_{d_{u}, \{t\}\cup {\cR} \backslash\{u\}} \nonumber\\
&\overset{\mathrm{(c)}}{=} \bigoplus_{u \in \cV_i \cap \cR} Z^{d_{u}}_{i,\cR\backslash \{u\}} \oplus \bigoplus_{u \in \cV_i \backslash \cV_i \cap \cR} Y_{\{u\} \cup {\cR}} \quad \oplus \bigoplus_{u\in {\cR \backslash \cV_i \cap \cR}} Z^{d_{u}}_{i,\cR\backslash \{u\}} \nonumber
\end{align} 
where $(a)$ follows since all  $W_{d_u,\cR}, u \in \cV_i$  but $W_{d'_{i}, \cR}$ appear twice  in the summation on the RHS of $(a)$ which is due to the structure of demands in  \DRS. Further, $(b)$ follows from the definition of $Y_{\{u\} \cup \cR}$.
The symbols in the first and third terms of $(c)$ can be obtained from the cache of the user due to Lemma ~\ref{Lem_subset_users}, and the symbols in the second term can be obtained from the delivery part because of Remark~\ref{Rem_SchemeA}.
Hence, the decodability of user $u'_i\in \cU_1$ follows.  This completes the proof of Theorem~\ref{Thm_reduced_usrs}.

\subsection{Proof of Theorem~\ref{th:basic}}
\label{Sec_proof_no_coding}

In the placement phase, the caches of all
users are populated with the same $M/N$ fraction of each file. 
Let each file $W_i$ be split in two parts: cached part $W^{(c)}_i$
of length $FM/N$, and uncached part $W^{(u)}_i$
of length $F(1-M/N)$.
The cache contents of all the users are the same, and given by
$Z_k=(Z^{(0)},Z^{(1)},\ldots,Z^{(N-1)})$, where
\begin{align*}
Z^{(i)}=W^{(c)}_i, \quad \text{for }i=0,1,\ldots,N-1.
\end{align*}
To describe the delivery phase, we consider two cases:

\noindent
\underline{Case 1: $N \leq K$}

For $N\leq K$, the server broadcasts the remaining $(1-M/N)$ fraction of each
file. 
This scheme achieves privacy because the transmission does not depend
on the demands of the users.

\noindent
\underline{Case 2: $K < N$}

Let $D_0,D_1,\ldots,D_{K-1}$ be the demands of the users.
The transmission $X$ has two parts $(X',J)$, where 
$X'=(X'_0,X'_1,\ldots,X'_{K-1})$ is the main payload, and $J$ is
the auxiliary transmission of negligible rate which helps each user find the corresponding decoding function. 
For each $i$, $X'_i$
is either $W^{(u)}_{D_j}$ for some $j$ or random bits
of the same length. In particular,
the position of $W^{(u)}_{D_j}$ in $X'$ is denoted by a random variable
$P_j\in [0:K-1]$. 
The random variables $P_0,P_1,\ldots,P_{K-1}$ are defined inductively
as
\begin{align*}
P_i = & \begin{cases} 
P_j & \text{if } D_i=D_j  \text{ for some } j<i \\
\sim unif([0:K-1]\setminus\{P_0,P_1,\ldots,P_{i-1}\})
& \text{if } D_i\neq D_j \text{ for all } j<i.
\end{cases}
\end{align*}
Note that each demanded (uncached) file is transmitted only in one component
of the transmission so that one user can not possibly detect the same file
(as its own demand) being transmitted in another component and thus infer
that the corresponding other user also has the same demand.

The keys $S_0,S_1,\ldots,S_{K-1}\in [0:K-1]$ are chosen i.i.d. and uniformly
distributed. 
The transmission is then given by
\begin{align*}
X'_{j} = \begin{cases} W^{(u)}_{D_i} & \text{ if } j=P_i   \text{ for some }i\in [0:K-1]\\
\sim unif\left( \{0,1\}^{F(1-M/N)}\right) & \text{ otherwise}
\end{cases}
\end{align*}
and
\begin{align*}
J = (P_0\oplus_K S_0, P_1\oplus_K S_1,\ldots,P_{K-1}\oplus_K S_{K-1})
\end{align*}
where $\oplus_K$ denotes the addition modulo $K$ operation.
Since user $k$ knows $S_k$, it can find $P_k$ from $J$. It
then can find $X'_{P_k}=W^{(u)}_{D_k}$, and thus $W_{D_k}=(Z^{(D_k)},X'_{P_k})$.

Next we show that this scheme also satisfies the privacy condition.
Let us denote $Q_i=P_i\oplus_K S_i$ for the ease of writing.
\begin{align}
I(\ND{k};X,  D_k,Z_k) & =I(\ND{k}; X'_0,\ldots,X'_{K-1}, Q_0,Q_1,\ldots,Q_{K-1},  D_k, S_k, W^{(c)}_{0}, \ldots,W^{(c)}_{N-1} ) \notag \\
&\overset{(a)}{=} I(\ND{k};  Q_0, \ldots,Q_{K-1}, D_k, S_k) \notag \\
& = I(\ND{k};  Q_0,  \ldots,Q_{k-1},Q_{k+1},\ldots,Q_{K-1}, 
D_k, S_k,P_k)\notag\\
& \overset{(b)}{=} 0 \notag
\end{align}
where $(a)$ follows because $(X'_0,\ldots,X'_{K-1},W^{(c)}_{0},
\ldots,W^{(c)}_{N-1} )$ is uniformly distributed in $\{0,1\}^{MF+FK(1-M/N)}$,
and is independent of $(\ND{k},Q_0, \ldots,Q_{K-1},D_k, S_k)$,
and $(b)$ follows because all the random variables in the
mutual information are independent.
In this scheme, the number of bits broadcasted is $FK(1-M/N)$ 
as the bits transmitted for communicating $J$ is negligible for large $F$.
Thus, the scheme achieves  rate $K(1-M/N)$.

\subsection{Proof of Theorem~\ref{Thm_PR_SR}}
\label{Sec_proof_PR_SR}
First we explain the scheme that achieves the memory-rate pairs given in Theorem ~\ref{Thm_PR_SR} . We further show that this scheme also preserves privacy.

For $t \in \{1,2,\ldots,NK-1\}$, we partition 
file $W_i,i\in [0:N-1]$ into $\sum_{l=t}^{NK-1} {NK \choose l}$  segments of $(NK-t)$ different sizes. These segments are grouped   into $(NK-t)$ groups such that all segments in the same group have the same size.  The segments are labelled by some subsets of $[0:NK-1]$. The segments of $W_i$ are $W_{i,\cR}$; $\cR \subset [0:NK-1], |\cR| \in \{t, t+1, \ldots, NK-1\}$. These $(NK-t)$ groups are given as
\begin{align*}
\cT_{|\cR|}^{i} = \{W_{i, \cR}| \cR \subset [0:NK-1] \} \quad \text{ for } t \leq |\cR| \leq NK-1.
\end{align*}
Thus, file $W_i,i\in[0:N-1]$ is given as 
\begin{align*}
W_i = \left(\cT_{|\cR|}^{i}\right)_{t \leq |\cR| \leq NK-1}.
\end{align*}
All elements of each file in one group have same size and elements of different groups have different sizes. For $|\cR_1| < |\cR_2|$, size of an element in $\cT_{|\cR_1|}^{i} $ is $r^{|\cR_2|-|\cR_1|}$ times the size of an element in $\cT_{|\cR_2|}^{i}$ for parameter $r \in [1,N-1]$. Hence, for $i \in [0:N-1]$ and $\cR \subset [0:NK-1], t \leq |\cR| \leq NK-1$, we have
$$ len(W_{i,\cR}) = \frac{r^{NK-|\cR|-1}}{\sum_{s=t}^{NK-1}{NK \choose s} r^{NK-s-1}}F.$$ 
Sum of all these segments is $F$ since  for any $r>0$, we have
\begin{align*}
len(W_i) & = \sum_{|\cR|=t}^{NK-1} \frac{ {NK \choose |\cR|}  r^{NK-|\cR|-1}}{\sum_{s=t}^{NK-1}{NK \choose s} r^{NK-s-1}}F \\
 &=\frac{  \sum_{|\cR|=t}^{NK-1} {NK \choose |\cR|}  r^{NK-|\cR|-1}}{\sum_{s=t}^{NK-1}{NK \choose s} r^{NK-s-1}}F \\
& = F. \label{Eq_fraction_sum}
\end{align*}

\underline{Caching:} 
The cache content $Z_k$ of user $k\in [0:K-1]$ has two components: the main load $Z'_k$  and sub load $Z''_k$. The main load $Z'_k$ is grouped into $NK-t$ groups similar to the way we partition the file. The groups are indexed by the cardinality of $\cR \subset [0:NK-1]$, where $t  \leq |\cR|\leq NK-1$. The group indexed by $|\cR|$ of user $k$ is denoted by $\cG_{k,|\cR|}$. Its content is determined by random variable $S_k \sim unif[0:N-1]$ which is shared between user $k$ and the server, and it is given by
\begin{align*}
\cG_{k,|\cR|} = \{W_{i,\cR}| W_{i,\cR} \in \cT_{|\cR|}^i \text{ and } S_k+kN \in \cR\}_{i\in[0:N-1]}.
\end{align*}
Then the main load $Z'_k$ is given by
\begin{align*}
Z'_k := (\cG_{k,|\cR|})_{t \leq |\cR| \leq NK-1}.
\end{align*}
Since there are $N{NK-1 \choose |\cR|-1}$ elements in $\cG_{k,|\cR|}$,
we obtain the size of the main load as
\begin{align*}
len(Z'_{k}) &= \sum_{|\cR|=t}^{NK-1}|\cG_{k,|\cR|}| \frac{r^{NK-|\cR|-1}F}{\sum_{s=t}^{NK-1}{NK \choose s} r^{NK-s-1}} \\
&= \frac{N\sum_{s=t}^{NK-1}{NK-1 \choose s-1}r^{NK-s-1}F}{\sum_{s=t}^{NK-1}{NK \choose s} r^{NK-s-1}} \\
&= MF.
\end{align*}

Now we define the sub load $Z''_k$ which is 
of negligible size compared to the file size $F$. To this end, we first define
$$
\cL := \{ kN+S_{k}|0 \leq k \leq K-1\},
$$
and
$$
\tau := \{ \cR | \cR \subset [0:NK-1], \cR \cap \cL \neq \phi, \text{ and } t+1\leq |\cR| \leq NK-1\}.
$$
The server generates independent symbols $S'_{\cR}$ for all $\cR\in\tau$, where each $S'_{\cR}\sim unif\{[0:\kappa_{|\cR|}-1]\}$, with $\kappa_{s}$ defined as 
\begin{align*}
\kappa_{s} = {NK \choose s} - {NK-K \choose s}, \quad \mbox{ for } s \in \{t+1,t+2,\ldots,NK-1\}.
\end{align*}
For all $\cR \in \tau$, $S'_{\cR}$ is cached at user $k$ if and only if $k \in \cR$. Then, the sub load $Z''_k$ is given by
\begin{align}
Z''_k:= \left(\{S'_{\cR} | \cR \in \tau \text{ and } (kN+S_{k})\cap \cR \neq \phi\}, S_k\right).
\end{align}
The cache content $Z_k$ is the concatenation of $Z'_k$ and $Z''_k$, i.e., $Z_k = (Z'_k,Z''_k)$.

\underline{Delivery:}
For a given demand vector $(D_0,\ldots,D_{K-1})$, the server first constructs an expanded demand vector $\bar{d}$ of $NK$-length. We write it as $K$ vectors of $N$ length each, as follows:
\begin{align}
\bar{d} = \left[\bar{d}^{(0)}, \bar{d}^{(1)}, \ldots, \bar{d}^{(K-1)}\right] \label{Eq_def_barD}
\end{align}
where $\bar{d}^{(k)}, k\in [0:K-1]$ is the vector obtained by applying $S_{k} \ominus D_{k}$ cyclic shift to the vector $(0,1,\ldots,N-1)$. Here $\ominus$ denotes modulo $N$ subtraction. That is, for $k \in [0:K-1]$, $d_i^{(k)} = i-(S_k-D_k) \mod N$. We also define
\begin{align}
\bar{S} \ominus \bar{D}:=\left(S_{0} \ominus D_{0}, S_{1} \ominus D_{1}, \ldots, S_{K-1} \ominus D_{K-1}\right).\label{Eq_def_SandD}
\end{align}
To explain the broadcast transmission, we define symbols $Y_{\cR}$ 
for $\cR \subset [0:NK-1]$ and $t+1 \leq |\cR| \leq NK$ as follows:
\begin{align*}
Y_{\cR} := \bigoplus_{u \in \cR} W_{d_{u}, \cR \setminus \{u\}}
\end{align*}
where $d_u$ is the $(u+1)$-th item of $\bar{d}$, and for all $\cR \in \tau$, we define symbols $W_{\cR}$
\begin{equation}\label{eq:key}
W_{\cR} := (W_{0,\cR} \oplus W_{1,\cR},W_{1,\cR} \oplus W_{2,\cR},\ldots ,W_{N-2,\cR} \oplus W_{N-1,\cR}).
\end{equation}
If the size of $W_{\cR}$ is $F'$ bits, we denote the first $rF'/(N-1)$ bits of $W_{\cR}$ by $W^{r}_{\cR}$, where $r \in [1,N-1]$. Further, we also define
\begin{align}
V_{\cR} &:= Y_{\cR}\oplus W^{r}_{\cR}, \quad \text{ for } r \in [1,N-1], \enspace \cR \in \tau,
\end{align}
and 
\begin{align}
V_{|\cR|} &:=\{V_{\cR}| \cR \in \tau\}. \label{Eq_def_setV}
\end{align}
Also, set $V$ is defined as the concatenation of all sets defined in~\eqref{Eq_def_setV}, i.e.,
\begin{align*}
V &:= (V_{s})_{t+1 \leq s \leq NK-1}.
\end{align*}
The server picks permutation functions $\left(\pi_{t+1}(\cdot),\pi_{t+2}(\cdot),\ldots,\pi_{NK-1}(\cdot)\right )$, where $\pi_{i}(\cdot) $ is picked uniformly at random from the symmetric group of permutations of $[0:\kappa_{i}-1]$ for \(i \in \{t+1,t+2,\ldots,NK-1\}\). These permutation functions are not fully shared with any of the users. The main payload is given by 
$$X' = (X'_{t+1},X'_{t+2},\ldots,X'_{NK-1},Y_{[0:NK-1]}) = (\pi_{t+1}(V_{t+1}),\pi_{t+2}(V_{t+2}),\ldots,\pi_{NK-1}(V_{NK-1}),Y_{[0:NK-1]}).$$ 
Rate of transmission is calculated as follows.
For $t+1 \leq |\cR| \leq NK-1$, the server transmits ${NK \choose |\cR| } - {NK -K \choose |\cR| } $ number of symbols $V_{\cR}$, and the server also transmits $Y_{[0:NK-1]}$. Then, the total number of bits transmitted in the main payload are given by
\begin{align*}
len(X') &= \sum_{s=t+1}^{NK}\frac{[{NK \choose s} -{NK-K\choose s}]r^{NK-s}F}{\sum_{s=t}^{NK-1}{NK \choose s} r^{NK-s-1}} \\
& = RF.
\end{align*}

Along with $X'$, the server also broadcasts some auxiliary transmission  $J$ of negligible rate, given by
\begin{align}
J &=(\{S'_{\cR} \oplus \alpha_{|\cR|,\cR}|\cR \in \tau\},  \bar{S} \ominus \bar{D})  \nonumber\\
&= (J',  \bar{S} \ominus \bar{D}). \label{Eq_aux_part}
\end{align}
Here, $\alpha_{|\cR|,\cR}$ denotes the position of $V_{\cR}$ in $\pi_{|\cR|}(V_{|\cR|})$ for $\cR \in \tau$. The private keys ensure that the location of any symbol $V_{\cR}$ is shared with user $k$ if and only if $S_k + kN \in \cR$.  For large file sizes, the size of the auxiliary transmission is negligible. The broadcasted message, $X$ can thus be given as $X=(X',J)$.

\underline{Decoding:} 
Now we explain how user $k \in [0:K-1]$ decodes the  segments that are missing from each group in her cache.
We can observe that the group $\cG_{k,NK-1}$ in the cache of user $k$ has all the elements of $\cT_{NK-1}^{D_k}$ except one. This missing element $W_{D_{k},[0:NK-1] \setminus \{S_k + kN\}} $ can be decoded as
\begin{align*}
\widehat{W}_{D_{k},[0:NK-1] \setminus \{S_k + kN\}} = Y_{[0:NK-1]} \oplus \left(\bigoplus_{u\in {[0:NK-1]\setminus \{S_k + kN\}}} W_{d_{u},{[0:NK-1]} \setminus \{u\}}\right).
\end{align*}
Observe that \(Y_{[0:NK-1]}\) is broadcasted by the server while each symbol $W_{d_{u},[0:NK-1] \setminus \{u\}}$ is a part of $\cG_{k,NK-1}$ and hence a part of the cache of user \(k\). Thus, user \(k\) can compute $\widehat{W}_{D_k, [0:NK-1]\setminus \{S_k + 3k\}}$. It follows that
\begin{align*}
\widehat{W}_{D_{k},[0:NK-1] \setminus \{S_k + kN\}} & = Y_{[0:NK-1]} \oplus \left(\bigoplus_{u\in {[0:NK-1]\setminus \{S_k + kN\}}} W_{d_{u},{[0:NK-1]} \setminus \{u\}}\right) \\
& = \bigoplus_{u \in [0:NK-1]} W_{d_{u}, [0:NK-1] \setminus \{u\}} \oplus \left(\bigoplus_{u\in {[0:NK-1]\setminus \{S_k + kN\}}} W_{d_{u},{[0:NK-1]} \setminus \{u\}}\right) \\
& = W_{d_{S_k + kN},[0:NK-1] \setminus \{S_k + kN\}} \\
& \overset{(a)}{=} W_{D_k, [0:NK-1] \setminus \{S_k + kN\}}.
\end{align*}
Here \((a)\) follows because $d_{S_k + kN} = (S_k + kN - (S_{k} - D_{k}))$ mod \(N\) \(= D_k\). Since user $k$ has all the segments in $\cT^{D_k}_{NK-1}$, we explain how user $k$ can obtain all symbols in any  set $\cT^{D_k}_j$, where $t \leq j \leq NK-2$. All symbols $W_{D_k, \cR} \in \cT^{D_k}_j$ such that $S_k + kN \in \cR$  form the group $\cG_{k,j}$ and hence are a part of her cache. All the remaining symbols $W_{D_k, \cR} \in \cT^{D_k}_j$ satisfying $S_k + kN \notin \cR$ can be decoded by user $k$ as follows:
\begin{align*}
    \widehat{W}_{D_{k}, {\cR}} & = X'_{{ |\cR^{+} |},t} \oplus W^{r}_{\cR^{+}}  \oplus \left(\bigoplus_{u\in {\cR}} W_{d_{u},\cR^{+} \setminus \{u\}}\right)
\end{align*}
where $\cR^{+} = \{S_k +kN \} \cup {\cR}$, $t = \alpha_{{|\cR^{+} |},\cR^{+}}$ and $X'_{{|\cR^{+} |},t}$ denotes the symbol in the $t$-th position of $X'_{{|\cR^{+} |}}$. Here, $X'_{{|\cR^{+}|}}$ is a part of the broadcast. User $k$ can recover $t$ using the auxiliary transmission as $t = S'_{\cR^{+}} \oplus (S'_{\cR^{+}} \oplus \alpha_{|\cR^{+}|,\cR^{+}})$ because $S'_{\cR^{+}}$ is part of her cache. All symbols in the second and third terms can also be recovered from the cache of user $k$. Thus, user \(k\) can compute $\widehat{W}_{D_{k}, {\cR}}$. Thus, we obtain
\begin{align*}
    \widehat{W}_{D_{k}, {\cR}} & = X'_{{ |\cR^{+} |},t} \oplus W^{r}_{\cR^{+}}  \oplus \left(\bigoplus_{u\in {\cR}} W_{d_{u},\cR^{+} \setminus \{u\}}\right) \\
    & = V_{\{S_k +kN \}\cup \cR }\oplus W^{r}_{\{S_k +kN \} \cup \cR}  \oplus \left(\bigoplus_{u\in {\cR}} W_{d_{u},\{S_k +kN \} \cup {\cR} \setminus \{u\}}\right)\\
    & = Y_{\{S_k +kN \}\cup \cR }\oplus W^{r}_{\{S_k +kN \} \cup \cR}\oplus W^{r}_{\{S_k +kN \} \cup \cR}  \oplus \left(\bigoplus_{u\in {\cR}} W_{d_{u},\{S_k +kN \} \cup {\cR} \setminus \{u\}}\right)\\
     & = \bigoplus_{u \in \{S_k +kN \} \cup {\cR}} W_{d_{u}, \{S_k +kN \} \cup {\cR} \setminus \{u\}} \oplus \left(\bigoplus_{u\in {\cR}} W_{d_{u},\{S_k +kN \} \cup {\cR} \setminus \{u\}}\right)\\
     & = W_{D_k, \cR}
\end{align*}
which shows that user $k$ can recover all symbols in  $\cT^{D_k}_j$ for $t \leq j \leq NK-1$. This completes the proof for decodability.

\underline{Proof of privacy:} We show that
\begin{align}
I(X; \ND{k}|Z_{k}, D_k) =0, \quad  \forall k \in [0:K-1] \label{Eq_proof_priv0}
\end{align}
which implies the privacy condition $I(\ND{k} ; X,Z_{k}, D_k) =0$, since $I(\ND{k} ; Z_{k}, D_k) =0$.
To show~\eqref{Eq_proof_priv0}, we first define 
\begin{align}
B_{k} := \{\alpha_{|\cR|,\cR}| \cR \subset [0:NK-1],  kN+S_{k} \cap \cR \neq \phi,    t+1 \leq |\cR| \leq NK-1\}.
\end{align}
We also divide $J'$ given in~\eqref{Eq_aux_part} into two parts, $J' = (J'_k, \Tilde{J'}_k)$, where $J'_k$  is the part $J$ which can be accessed by user $k$ while $\Tilde{J'}_k$ is the remaining part. These are defined as follows:
\begin{align*}
J'_k & := \{S'_{\cR} \oplus \alpha_{|\cR|, \cR} | \cR \in \tau, kN+S_k \in \cR  \},\\
\Tilde{J'}_k & := J'\setminus J'_k.
\end{align*}
Then, we have 
\begin{align}
I(X; \ND{k}|Z_{k}, D_k) & = I(X' , J; \ND{k}|Z_{k}, D_k) \nonumber \\
& = I(X' , J'_k,  \Tilde{J'}_k,  \bar{S} \ominus \bar{D}; \ND{k}|Z_{k}, D_k) \nonumber \\
&  \overset{(a)}{=}  I(X', \bar{S} \ominus \bar{D}, J'_k, B_k; \ND{k}|Z_{k}, D_k) \nonumber \\
&  \overset{(b)}{=}  I(X', \bar{S} \ominus \bar{D}, B_k; \ND{k}|Z_{k}, D_k) \nonumber \\
& = I( \bar{S} \ominus \bar{D}, B_k; \ND{k}|Z_{k}, D_k) +    I(X' ; \ND{k}|Z_{k}, D_k,  \bar{S} \ominus \bar{D}, B_k)  \nonumber \\
& \overset{(c)}{=}  I(X' ; \ND{k}|Z_{k}, D_k,  \bar{S} \ominus \bar{D},B_k) \nonumber \\
& =   I(Y_{[0:NK-1]}, \pi_{t+1}(V_{t+1}),\pi_{t+2}(V_{t+2}),\ldots,\pi_{NK-1}(V_{NK-1}); \ND{k}|Z_{k}, D_k, B_{k}, \bar{S} \ominus \bar{D}) \nonumber\\
& \overset{(d)}{=}  I(Y_{[0:NK-1]}, V; \ND{k}|Z_{k}, D_k, B_{k}, \bar{S} \ominus \bar{D}).
\label{Eq_proof_priv1}
\end{align}
Here, $(a)$ follows since $B_k$ is a function of $(Z_k,J'_k)$ and $ \Tilde{J'}_k$ is independent of all other random variables on the RHS of $(a)$, and $(b)$ follows since $J'_k$ is a function of $(Z_k, B_k)$. Further, $(c)$ follows  since $( \bar{S} \ominus \bar{D}, B_k)$ is independent of other random variables, and $(d)$ follows due to the fact that the permutations $\left(\pi_{t+1}(\cdot),\pi_{t+2}(\cdot),\ldots,\pi_{NK-1}(\cdot)\right )$ are independent of all other random variables.
Next we show that the RHS of~\eqref{Eq_proof_priv1} is zero.
To this end, we first divide the set $V$, defined in~\eqref{Eq_def_setV}, into two parts: the first part $X_{k}$ contains the symbols in $V$ whose positions are known to user $k$, and the second part $\Tilde{X}_{k}$ contains the remaining symbols in $V$, i.e.,
\begin{align*}
X_{k} & := \{V_{\cR}|(kN+S_{k}) \cap \cR \neq \phi, \cR \in \tau\}, \\
\Tilde{X}_{k} & := V \setminus X_{k}.
\end{align*}
Set \(\Tilde{X}_{k}\) can be further divided into more groups labelled by \(\Tilde{X}_{k,|\cR|}\), where \(t+1 \leq |\cR| \leq NK-1\), as follows:
$$
\Tilde{X}_{k,|\cR|} = \{V_{\cR}|(kN+S_{k}) \cap \cR = \phi, \cR \in \tau \}.
$$
Then, we get
\begin{align}
&I(Y_{[0:NK-1]}, V; \ND{k}| Z_{k},D_k, B_{k}, \bar{S} \ominus \bar{D} )\\ 
& = I(Y_{[0:NK-1]}, X_{k}, \Tilde{X}_{k}; \ND{k} | Z_{k},D_k, B_{k}, \bar{S} \ominus \bar{D}) \nonumber \\
& = I(Y_{[0:NK-1]},X_{k} ; \ND{k} | Z_{k},D_k, B_{k}, \bar{S} \ominus \bar{D} ) + I(\Tilde{X}_{k}; \ND{k}|Y_{[0:NK-1]},X_{k}, Z_{k},D_k, B_{k}, \bar{S} \ominus \bar{D}). \nonumber \\
& \overset{(a)}{=} I(Y_{[0:NK-1]},X_{k}, W_{D_k} ; \ND{k} | Z_{k},D_k, B_{k}, \bar{S} \ominus \bar{D} ) + I(\Tilde{X}_{k}; \ND{k}|Y_{[0:NK-1]},X_{k}, Z_{k}, D_k,B_{k}, \bar{S} \ominus \bar{D}, W_{D_k}). \nonumber \\
&  \overset{(b)}{=}  I( W_{D_k} ; \ND{k} | Z_{k},D_k, B_{k}, \bar{S} \ominus \bar{D}) + I(\Tilde{X}_{k}; \ND{k}| Z_{k}, D_k, B_{k}, \bar{S} \ominus \bar{D}, W_{D_k}). \nonumber \\
& = I(W_{D_k} ; \ND{k}|Z_{k},D_k, B_{k}, \bar{S} \ominus \bar{D}) +\sum_{i=t+1}^{NK-1} I(\Tilde{X}_{k,i};  \ND{k}|Z_{k},D_k, B_{k}, \bar{S} \ominus \bar{D}, W_{D_k},\Tilde{X}_{k,t+1}, \ldots, \Tilde{X}_{k,i-1}) \label{eq_pause14} 
\end{align}
where in~\eqref{eq_pause14}, we used  $\Tilde{X}_{k,i} =\phi$ for $i < t+1$. Here, $(a)$ follows because we have seen in the  decodability section that $W_{D_{k}}$ can be recoverd from $(\bar{S} \ominus \bar{D}, Z_{k}, X_{k}, Y_{[0:NK-1]})$, and $(b)$ follows since each \( V_{\cR} \in X_{k}\) can be written as
\begin{align*}
V_{\cR} & = Y_{\cR} \oplus W^{r}_{\cR} \\
& = \bigoplus_{u \in \cR} W_{d_{u}, \cR \backslash \{u\}} \oplus W^{r}_{\cR} \\
& = \bigoplus_{u \in \cR \backslash (S_{k}+kN)} W_{d_{u}, \cR \backslash \{u\}} \oplus W_{d_{S_{k}+kN}, \cR \backslash \{S_{k}+kN\}} \oplus W^{r}_{\cR}.
\end{align*}
Here, the first and third terms can be recovered from $Z_k$ and the second term is a part of $W_{D_k}$ since $d_{S_k+kN} = D_k$. Similarly, we have
\begin{align*}
Y_{[0:NK-1]} & = \bigoplus_{u\in {[0:NK-1]}} W_{d_{u},{[0:NK-1]} \backslash \{u\}}\\
& = \bigoplus_{u\in {[0:NK-1] \backslash (S_{k}+kN)}} W_{d_{u},{[0:NK-1]} \backslash \{u\}} \oplus W_{d_{S_{k}+kN}, [0:NK-1] \backslash \{S_{k}+kN\}}.
\end{align*}
Here, all symbols in the first term are a part of $Z_k$ while the second term is a part of $W_{D_k}$ because $d_{S_k+kN} = D_k$. Thus, $(X_{k}, Y_{[0:NK-1]})$ is a function of $(\bar{S} \ominus \bar{D}, Z_{k}, W_{D_{k}})$ which completes the argument for $(b)$.

Next, we show that each term on the RHS of~\eqref{eq_pause14} is zero. First, we consider the terms $I(\Tilde{X}_{k,i};Z_{k}, B_{k}, \bar{S} \ominus \bar{D}, \ND{k}, W_{D_k}|Y_{[0:NK-1]},X_{k},\Tilde{X}_{k,i-1}, \ldots, \Tilde{X}_{k,t+1})$ for $t+1 \leq i \leq NK-1$. For simplicity of notation, we define set $\tau_{k,i}$ as follows:
\[
\tau_{k,i} = \{\cR \in \tau, \cR \cap (kN+S_{k}) = \phi, |\cR| = i\}.
\]
For $k \in [0:K-1]$ and $ t+1 \leq i \leq NK-1$, we get
\begin{align}
& I(\Tilde{X}_{k,i};  \ND{k}|Z_{k},D_k, B_{k}, \bar{S} \ominus \bar{D}, W_{D_k}, \Tilde{X}_{k,t+1}, \ldots\Tilde{X}_{k,i-1} ) \nonumber \\
& = I((V_{\cR})_{\cR \in \tau_{k,i}};  \ND{k}|Z_{k},D_k, B_{k}, \bar{S} \ominus \bar{D}, W_{D_k}, \Tilde{X}_{k,t+1}, \ldots\Tilde{X}_{k,i-1} ) \nonumber \\
& = I((Y_{\cR} \oplus W_{\cR})_{\cR \in \tau_{k,i}};  \ND{k}|Z_{k},D_k, B_{k}, \bar{S} \ominus \bar{D}, W_{D_k}, \Tilde{X}_{k,t+1}, \ldots\Tilde{X}_{k,i-1} ) \nonumber \\
& = I((Y_{\cR} \oplus (W_{0,\cR} \oplus W_{1,\cR},...,W_{N-2,\cR} \oplus W_{N-1,\cR}))_{\cR \in \tau_{k,i}};  \ND{k}|Z_{k},D_k, B_{k}, \bar{S} \ominus \bar{D}, W_{D_k}, \Tilde{X}_{k,t+1}, \ldots\Tilde{X}_{k,i-1} ) \nonumber \\
& = 0. \label{eq:pause11} 
\end{align}
Here, \eqref{eq:pause11} follows because each symbol $(W_{i,\cR})$, $i \in [0:N-1]$ is non-overlapping with $(\Tilde{X}_{k,t+1}, \ldots,\Tilde{X}_{k,i-1}, Y_{k,i})$ and because $W_{D_k}$ contains only one symbol in $W_{\cR}$, namely $W_{D_k, \cR}$. We can also see that the first term on the RHS of~\eqref{eq_pause14} is zero, i.e.,
\begin{align}
I(W_{D_k} ; \ND{k} | Z_{k}, D_k,B_{k}, \bar{S} \ominus \bar{D} ) = 0. \label{eq:obs13}
\end{align}
because $W_{D_k}$ is independent of $\ND{k}$.
Thus, from~\eqref{eq:obs13} and \eqref{eq:pause11}, we obtain
\begin{align*}
I(Y_{[0:NK-1]}, V; \bar{D}| Z_{k},D_k, B_{k}, \bar{S} \ominus \bar{D} ) = 0.
\end{align*}
This together with~\eqref{Eq_proof_priv1} implies~\eqref{Eq_proof_priv0}.

\input{proofs2.tex}
\input{proofs_exact.tex}

%% file: proofs2.tex
\subsection{Proof of Theorem~\ref{Thm_order}}
\label{Sec_proof_order}
To prove the theorem, we first give some notations and inequalities.
For parameter $r_2 = \frac{KM}{N}$, let
\begin{align}
R^{\text{YMA}}_{N,K}\left(\frac{Nr_2}{K}\right) & = \frac{{K \choose r_2+1} - {K-\min(N,K) \choose r_2+1}}{{K \choose r_2}}, \quad \mbox{ for } r_2 \in \{0, 1,\ldots, K\}, \label{Eq_rate1}\\
R^{\text{MAN}}_{N,K}\left(M\right)  &= K\left(1-\frac{M}{N}\right)\min\left(\frac{1}{1+\frac{KM}{N}},\frac{N}{K}\right), \quad \mbox{ for } M \in \{0, N/K,2N/K,\ldots, N\}.  \label{Eq_rate2}
\end{align}
Furthermore, let $R^{\text{YMA, c}}_{N,K}(M)$ and $R^{\text{MAN, c}}_{N,K}(M)$ denote the lower convex envelop of the points in~\eqref{Eq_rate1} and~\eqref{Eq_rate2}, respectively. Recall that  $\Rto $  and $\Rm$ denote the optimal rate with privacy and without privacy as defined in \eqref{Eq_opt_rate_priv} and  \eqref{Eq_opt_rate_nopriv}, respectively. 
Then we have the following inequalities which hold for all $M\geq 0$:
\begin{align}
\Rm  \stackrel{(a)}{\leq} \Rto \stackrel{(b)}{\leq} \Rtc = R^{\text{YMA, c}}_{N,NK-K+1}(M) \leq R^{\text{YMA, c}}_{N,NK}(M)
\stackrel{(c)}{\leq} R^{\text{MAN, c}}_{N,NK}(M) \label{Eq_rate_ineq1}
\end{align}
where $(a)$ follows from the fact that the optimal rate required with demand privacy is  larger than that of without privacy, $(b)$ follows since an achievable rate is lower-bounded by the optimal rate, and
$(c)$ was shown in~\cite{Yu18}. 

\subsubsection{Proof of Part~\ref{Thm_ordr_part1}), ($N \leq  K$)}

We first prove that
\begin{align}
\frac{R^{A}_{N,K}(M)}{R^{*}_{N,K}(M)}  \leq 
\begin{cases}
4 & \text{ if } M \leq \left(1 - \frac{N}{K}\right)\\
8 & \text{ if}  \left(1 - \frac{N}{K}\right) \leq M \leq \frac{N}{2}.
\end{cases} \label{Eq_Rp_Rm_bound}
\end{align}
To this end, we show that
\begin{align}
\frac{ R^{\text{MAN, c}}_{N,NK}(M)}{\Rm }   \leq 
\begin{cases}
4 & \text{ if } M \leq \left(1 - \frac{N}{K}\right)\\
8 & \text{ if}  \left(1 - \frac{N}{K}\right) \leq M \leq \frac{N}{2}.
\end{cases} \label{Eq_ratio_bnd1}
\end{align}
Then the result follows from~\eqref{Eq_rate_ineq1}. 
We first consider the ratio $\frac{\RMaNx}{\RMaN}$ for $M\in \{ 0, N/K, 2N/K, \ldots, N \}$.
We have
\begin{align}
\frac{\RMaNx}{\RMaN} = \frac{N\min\left(\frac{1}{1+KM},\frac{1}{K}\right)}{\min\left(\frac{1}{1+\frac{KM}{N}},\frac{N}{K}\right)}, \quad M\in \{ 0, N/K, 2N/K, \ldots, N \}. \label{ratio}
\end{align}
We consider the following three cases.

\noindent \underline {Case 1: $M \in [0 , 1-\frac{N}{K}]$}

\noindent We first find $\min\left(\frac{1}{1+KM},\frac{1}{K}\right)$ and $\min\left(\frac{1}{1+\frac{KM}{N}},\frac{1}{K}\right) $. 
\begin{align*}
\frac{1}{1+KM} & \geq \frac{1}{1+K(1 - N/K)}\\
& = \frac{1}{K-N+1}\\
& > \frac{1}{K}, \quad \mbox{ for } N >1.
\end{align*}
So, $\min\left(\frac{1}{1+KM},\frac{1}{K}\right) = \frac{1}{K}$. Further,
\begin{align*}
\frac{1}{1+\frac{KM}{N}} &  \geq \frac{1}{1+\frac{K}{N}(1-N/K)}\\
& = \frac{N}{K}.
\end{align*}
Thus, $\min\left(\frac{1}{1+\frac{KM}{N}},\frac{N}{K}\right) = \frac{N}{K}$. Hence~\eqref{ratio} gives 1.

\noindent \underline{Case 2: $M \in [1-\frac{N}{K}, 1-\frac{1}{K}]$}

\noindent In this case, we get
\begin{align*}
\min\left(\frac{1}{1+KM},\frac{1}{K}\right) = \frac{1}{K}, 
\end{align*}
and
\begin{align*}
\min\left(\frac{1}{1+\frac{KM}{N}},\frac{N}{K}\right) =\frac{1}{1+\frac{KM}{N}}.
\end{align*}
Then from~\eqref{ratio},  it follows that
\begin{align*}
\frac{ \RMaN}{ R^{\text{MAN}}_{N,K}\left(M\right) }  & =\frac{N}{K}\left(1+\frac{KM}{N}\right)\\
&=\frac{N}{K}+M \\ 
& \leq 2
\end{align*}
where the last inequality follows since $\frac{N}{K} \leq 1$ and $M \leq 1$.

\noindent  \underline{Case 3: $M \in [ 1-\frac{1}{K}, N]$}

\noindent In this case, we obtain
\begin{align*}
\min\left(\frac{1}{1+KM},\frac{1}{K}\right) = \frac{1}{1+KM}, \quad \mbox{ if } 1-\frac{1}{K} \leq M \leq N
\end{align*}
and
\begin{align*}
\min\left(\frac{1}{1+\frac{KM}{N}},\frac{N}{K}\right) =\frac{1}{1+\frac{KM}{N}}, \quad \mbox{ if } 1-\frac{1}{K} \leq M \leq N.
\end{align*}

Then from~\eqref{ratio}, we get the following   
\begin{align}
\frac{ \RMaNx}{ \RMaN} & = \frac{N}{1+KM}\left(1+\frac{KM}{N}\right) \notag \\
&=\frac{N+KM}{1+KM}\notag \\
&=\frac{N-1}{1+KM}+1. \label{gap_temp2}  
\end{align}
Further,
\begin{align*}
M \geq 1-\frac{1}{K} & \implies KM \geq K-1,\\
& \implies KM \geq N -1 \; (\text{Since $K\geq N$}),\\
& \implies KM +1 \geq N-1,\\
& \implies \frac{N-1}{1+KM} \leq 1.
\end{align*} 
Then, we obtain $\frac{ \RMaNx}{ \RMaN} \leq 2$ from~\eqref{gap_temp2}.

Let $R^{\text{MAN, lin}}_{N,K}(M)$ denote the region obtained  by linearly interpolating the adjacent memory points given in~\eqref{Eq_rate2}. Similarly, $R^{\text{MAN, lin}}_{N,NK}(M)$ denotes the linear interpolation of the points $R^{\text{MAN}}_{N,NK}(M), M \in \{0, N/K,2N/K, \ldots, N \}$. Then, it follows from the above three cases that
\begin{align}
\frac{R^{\text{MAN, lin}}_{N,NK}(M)}{ R^{\text{MAN,lin}}_{N,K}(M)}  \leq 
\begin{cases}
1& \text{ if } M \leq \left(1 - \frac{N}{K}\right)\\
2 & \text{ if}  \left(1 - \frac{N}{K}\right) \leq M \leq \frac{N}{2}.
\end{cases} \label{Eq_1ratio1}
\end{align}
Next we need  the following lemma.
\begin{lemma}
	\label{Lem_bound_lin}
	For $N\leq K$, the following holds:
	\begin{align}
	\frac{R^{\text{MAN, lin}}_{N,K}(M) }{\Rm } \leq 4, \quad \mbox{for } 0 \leq M \leq \frac{N}{2}. \label{Eq_1ratio2}
	\end{align}
\end{lemma}
\begin{proof}
	See Appendix~\ref{Sec_append}.
\end{proof}
Since  $ R^{\text{MAN, c}}_{N,NK}(M) \leq R^{\text{MAN, lin}}_{N,NK}(M)$, \eqref{Eq_ratio_bnd1} follows from~\eqref{Eq_1ratio1} and~\eqref{Eq_1ratio2}. This further implies~\eqref{Eq_Rp_Rm_bound}. 

Now it remains to prove 
\begin{align*}
\frac{R^{A}_{N,K}(M)}{R^{*}_{N,K}(M)} &  \leq 2, \quad \text{ if }M \geq \frac{N}{2}.
\end{align*}
By substituting \(r_{2} = \lfloor{NK/2}\rfloor\) in~\eqref{Eq_rate1} for \(N\) files and \(NK\) users, we get
\begin{align*} \label{}
R^{\text{YMA}}_{N,NK}\left (\frac{\lfloor{NK/2}\rfloor}{K} \right ) & \leq \frac{{NK \choose \lfloor{NK/2}\rfloor +1}}{{NK \choose \lfloor{NK/2}\rfloor}} \\
&= \frac{NK-\lfloor{NK/2}\rfloor}{\lfloor{NK/2}\rfloor+1}\\
& = \frac{NK+1}{\lfloor{NK/2}\rfloor+1} - 1\\
& \leq \frac{NK+1}{NK/2-1/2+1} - 1 \\
&= 1.
\end{align*}
Since $\frac{\lfloor{NK/2}\rfloor}{K} \leq \frac{N}{2}$, we have $R^{\text{YMA, c}}_{N,NK}(N/2) \leq 1$. Also, $R^{\text{YMA, c}}_{N,NK}(N) = 0$. Thus, for $N/2 \leq M \leq N $, it follows that
\begin{equation} \label{eq:achline}
R^{\text{YMA, c}}_{N,NK}(M) \leq 2 \left (1 - \frac{M}{N} \right).
\end{equation}
The cutset bounds on the rates without privacy gives that
\begin{equation}\label{eq:cutset}
R^{*}_{N,K}(M) \geq \left (1 - \frac{M}{N} \right).  
\end{equation}
From~\eqref{eq:achline} by~\eqref{eq:cutset}, we obtain
\begin{align*}
\frac{R^{\text{YMA, c}}_{N,NK}(M)}{R^{*}_{N,K}(M)}  \leq 
2, \quad \text{ for }M \geq \frac{N}{2}.
\end{align*}
From~\eqref{Eq_rate_ineq1}, we have $R^{A}_{N,K}(M) \leq R^{\text{YMA, c}}_{N,NK}(M)$ which then implies that
\begin{align*}
\frac{R^{A}_{N,K}(M)}{R^{*}_{N,K}(M)}  \leq 
2, \quad \text{ for }M \geq \frac{N}{2}.
\end{align*}
This completes the proof of  Part~\ref{Thm_ordr_part1}).

\subsubsection{Proof of Part~\ref{Thm_ordr_part2}), $K < N$}

 We  denote the memory corresponding to parameters \(r = r_{0}\) and \(t=t_{0}\) in (\ref{Eq_Thm_PR_SR}) by \(M_{r_{0},t_{0}}\). First we consider the  memory regime \(M \leq N/2\). Substituting \(t=1\) in \eqref{Eq_Thm_PR_SR}, we get the achievability of the following memory-rate pairs
\begin{align} 
(M_{r,1},R) &= \left(\frac{N\sum_{s=1}^{NK-1}{NK-1 \choose s-1}r^{NK-s-1}}{\sum_{s=1}^{NK-1}{NK \choose s} r^{NK-s-1}}, \frac{\sum_{s=2}^{NK}[{NK \choose s} -{NK-K\choose s}]r^{NK-s}}{\sum_{s=1}^{NK-1}{NK \choose s} r^{NK-s-1}} \right) \nonumber\\
&= \left(\frac{N\sum_{s=1}^{NK-1}{NK-1 \choose s-1}r^{NK-s}}{\sum_{s=1}^{NK-1}{NK \choose s} r^{NK-s}}, \frac{\sum_{s=2}^{NK}{NK \choose s}r^{NK-s+1} - \sum_{s=2}^{NK}{NK-K\choose s}r^{NK-s+1}}{\sum_{s=1}^{NK-1}{NK \choose s} r^{NK-s}} \right)\nonumber \\
&= \left(\frac{N\sum_{s=0}^{NK-2}{NK-1 \choose s}r^{NK-s-1}}{\sum_{s=1}^{NK-1}{NK \choose s} r^{NK-s}}, \frac{\sum_{s=2}^{NK}{NK \choose s}r^{NK-s+1} - \sum_{s=2}^{NK}{NK-K\choose s}r^{NK-s+1}}{\sum_{s=1}^{NK-1}{NK \choose s} r^{NK-s}} \right) \nonumber\\
&= \left(\frac{N(\sum_{s=0}^{NK-1}{NK-1 \choose s}r^{NK-s-1}-1)}{\sum_{s=0}^{NK}{NK \choose s} r^{NK-s}-r^{NK}-1}, \frac{\sum_{s=0}^{NK}{NK \choose s}r^{NK-s+1} - \sum_{s=0}^{NK}{NK-K\choose s}r^{NK-s+1} -Kr^{NK}}{\sum_{s=1}^{NK-1}{NK \choose s} r^{NK-s} } \right) \nonumber\\
&= \left(\frac{N((r+1)^{NK-1} - 1)}{(r+1)^{NK} - r^{NK} - 1}, \frac{r((r+1)^{NK}-Kr^{NK-1}-(r+1)^{NK-K}r^{K})}{((r+1)^{NK} - r^{NK} - 1)} \right). \label{eq:simpl}
\end{align} 
We first show that $M_{r,1}$ in \eqref{eq:simpl} satisfies the following
\begin{align}
M_{r,1} &= \frac{N((r+1)^{NK-1} - 1)}{(r+1)^{NK} - r^{NK} - 1} \nonumber \\
& = \frac{N}{r+1}\left (1 -  \frac{r-r^{NK}}{(r+1)^{NK} - r^{NK} - 1} \right) \nonumber \\
& \geq \frac{N}{r+1} \label{eq:interm1}
\end{align}
where the last inequality follows since $\left (1 -  \frac{r-r^{NK}}{(r+1)^{NK} - r^{NK} - 1} \right) \geq 1$.
Using the fact that all points on the line joining \((0,K)\) and \((M_{r,1},R)\) are also achievable, for $M \leq M_{r,1}$ we get
\begin{align}
R^{BC}_{N,K}(M)&\leq \left(\frac{R-K}{M_{r,1}}\right)M + K\\
&= \left(\frac{(r+1)^{NK}r - (r+1)^{NK-K}r^{K+1} - K(r+1)^{NK}+K}{N((r+1)^{NK-1} - 1)}\right)M + K. \label{eq:interm2}
\end{align} 
Now we substitute \(r=K/s - 1\) for some integer \(s\) in the interval \([1 ,\lfloor{K/2}\rfloor]\). Note that,  \(Ns/K = N/(r+1) \leq M_{r,1}\), where the inequality follows from \eqref{eq:interm1}. Thus, \eqref{eq:interm2} holds for \(M=Ns/K\) and we obtain
\allowdisplaybreaks
\begin{align}
R^{BC}_{N,K}(Ns/K) & \leq \left(\frac{(\frac{K}{s})^{NK}\left(\frac{K}{s}-1\right) - (\frac{K}{s})^{NK-K}\left(\frac{K}{s}-1\right)^{K+1} - K(\frac{K}{s})^{NK}+K}{N((\frac{K}{s})^{NK-1} - 1)}\right)\frac{Ns}{K} + K \nonumber\\
& = \left(\frac{K^{NK}(K-s) - K^{NK-K}(K-s)^{K+1} - sK^{NK+1}+Ks^{NK+1}}{N(K^{NK-1}s - s^{NK})}\right)\frac{N}{K} + K \nonumber\\
&= \left(\frac{K^{NK-1}(K-s) - K^{NK-K-1}(K-s)^{K+1} - sK^{NK}+s^{NK+1}}{(K^{NK-1}s - s^{NK})}\right) + K \nonumber \\
&= \left(\frac{K^{NK-1}(K-s) - K^{NK-K-1}(K-s)^{K+1} - Ks^{NK}+s^{NK+1}}{(K^{NK-1}s - s^{NK})}\right). \label{eq:mainineq}
\end{align}
Note that \(R^{\text{YMA}}_{N,K}(Ns/K) = (K-s)/(s+1)\). Dividing \eqref{eq:mainineq} by \(R^{\text{YMA}}_{N,K}(Ns/K)\) yields
\begin{align}
\frac{R^{BC}_{N,K}(Ns/K)}{R^{\text{YMA}}_{N,K}(Ns/K)}&= (s+1)\left(\frac{K^{NK-1}(K-s) - K^{NK-K-1}(K-s)^{K+1} - Ks^{NK}+s^{NK+1}}{(K^{NK-1}s - s^{NK})(K-s)}\right)\nonumber \\
&= (s+1)\left(\frac{K^{NK-1} - K^{NK-K-1}(K-s)^{K} - s^{NK}}{(K^{NK-1}s - s^{NK})}\right) \nonumber \\
&\leq \frac{(s+1)}{s}\left(\frac{K^{NK-1} - K^{NK-K-1}(K-s)^{K} - s^{NK-1}}{(K^{NK-1} - s^{NK-1})}\right) \nonumber \\
&= \frac{(s+1)}{s}\left(1 - \frac{K^{NK-K-1}(K-s)^{K}}{(K^{NK-1} - s^{NK-1})}\right) \nonumber \\
&\leq \frac{(s+1)}{s}\left(1 - \frac{K^{NK-K-1}(K-s)^{K}}{K^{NK-1}}\right) \nonumber \\
&= \frac{(s+1)}{s}\left(1 - \left(1 - \frac{s}{K}  \right)^{K}\right). \label{eq:fin1}
\end{align}
Now we need to compute the maximum value of the expression on the RHS of \eqref{eq:fin1} for \(K \geq 2\) and \(1 \leq s \leq \lfloor{K/2}\rfloor\). Note that both \(s\) and \(K\) are integers. For \(s\) fixed, \(K \geq 2s\) satisfies all constarints. Observe that for \(s\) fixed the function \((1 - \frac{s}{K}  )^{K}\) is increasing in \(K\). Thus, the RHS of \eqref{eq:fin1} is decreasing in \(K\). Since we want to compute the maxima, we substitute \(K=2s\). Thus, it follows that 
\begin{align}
\frac{R^{BC}_{N,K}(Ns/K)}{R^{\text{YMA}}_{N,K}(Ns/K)}& \leq 
\frac{(s+1)}{s}\left(1 - \left(\frac{1}{2}  \right)^{2s}\right). \label{eq:fin2} 
\end{align}
The expression on the RHS of \eqref{eq:fin2} takes value $3/2$ when \(s=1\). For \(s \geq 2\), \((s+1)/s \leq 3/2\). So the maxima is $3/2$ and attained at \(s=1,   K=2\). Thus, we have
\begin{equation} \label{eq:ineq1}
\frac{R^{BC}_{N,K}(Ns/K)}{R^{\text{YMA}}_{N,K}(Ns/K)} \leq 3/2 \qquad \forall s \in \{1,2,\ldots, \lfloor{K/2}\rfloor\}.
\end{equation}
Substituting \(t=1\) and \(r=1\) in \eqref{Eq_Thm_PR_SR}, we get the following memory-rate pair
\begin{equation}\label{eq:(1,1)}
(M_{1,1},R) = \left (\frac{N}{2},\frac{2^{NK}-2^{NK-K}-K}{2^{NK}-2} \right).
\nonumber
\end{equation}
We know that \(R^{\text{YMA}}_{N,K}(N/2) \geq K/(K+2)\). Thus,
\begin{align}
\frac{R^{BC}_{N,K}(N/2)}{R^{\text{YMA}}_{N,K}(N/2)}& \leq \frac{(K+2)(2^{NK}-2^{NK-K}-K)}{K(2^{NK}-2)} \nonumber \\
& \leq \frac{(K+2)}{K}\left (\frac{(2^{NK}-2^{NK-K}-2)}{(2^{NK}-2)}\right ) \nonumber \\
& = \frac{(K+2)}{K}\left (1 - \frac{2^{NK-K}}{(2^{NK}-2)}\right ) \nonumber \\
& \leq \frac{(K+2)}{K}\left (1 - \frac{1}{2^{K}}\right). \label{eq:fin3}
\end{align}
The maximum value of the expression on the RHS of \eqref{eq:fin3}) is 3/2 and is attained at \(K=2\). The analysis for computing this maxima is exactly the same as the one we used for \eqref{eq:fin2}.
Thus,
\begin{equation}\label{eq:ineq2}
\frac{R^{BC}_{N,K}(N/2)}{R^{\text{YMA}}_{N,K}(N/2)} \leq 3/2.
\end{equation}
It was shown in~\cite{Yu18} that
\begin{equation}\label{eq:ineq3}
\frac{R^{\text{YMA}}_{N,K}(M)}{R^{*}_{N,K}(M)} \leq 2.
\end{equation}
For \(M \leq N/2\), \(R^{\text{YMA}}_{N,K}(M)\) is the linear extrapolation of the points \(R^{\text{YMA}}_{N,K}(M')\) where \(M' \in \{0,N/K,2N/K,\ldots,N/2\}\). Thus, using~\eqref{eq:ineq1}, \eqref{eq:ineq2} and \eqref{eq:ineq3}, we conclude that,
\begin{equation}
\frac{R^{BC}_{N,K}(M)}{R^{*p}_{N,K}(M)} \leq 3 \quad \text{for } M \leq N/2.
\end{equation}
Now let us consider the memory regime $M \geq N/2$.
All memory-rate points on the line joining \((N/2,R^{BC}_{N,K}(M)(N/2))\) and \((N,0)\) are achievable. Moreover from~\eqref{eq:(1,1)}, it is clear that \(R^{BC}_{N,K}(N/2) \leq 1\). So,
\begin{equation}\label{eq:ineq4}
R^{BC}_{N,K}(M)(M) \leq 2 - \frac{2}{N}M, \quad \text{for } M \geq N/2.
\end{equation}
Using the cut-set bounds, we have the lower bound on the non-private rate
\begin{equation}\label{eq:ineq5}
R^{*}_{N,K}(M) \geq \left (1 - \frac{M}{N}\right).
\end{equation}
Since $R^{*p}_{N,K}(M) \geq R^{*}_{N,K}(M) $, from~\eqref{eq:ineq4} and~\eqref{eq:ineq5}, it follows that
\begin{equation}
\frac{R^{BC}_{N,K}(M)}{R^{*p}_{N,K}(M)} \leq 2, \quad \text{for } M \geq N/2 \nonumber.
\end{equation}
This completes the proof of Part~\ref{Thm_ordr_part2}).

\subsubsection{Proof of Part~\ref{Thm_ordr_part_exct})}
On substituting \(r=NK-K\) and \(r=NK-K+1\) in \eqref{Eq_YMA_achv_pair} we get memory-rate trade-off points \((\frac{N(NK-K)}{NK-K+1}, \frac{1}{NK-K+1})\) and \((N,0)\), respectively. Observe that both these points lie on the line given by \eqref{eq:ineq5}. This shows that \(R^{*p}_{N,K}(M) = R^{*}_{N,K}(M)\) for \(M \geq \frac{N(NK-K)}{NK-K+1}\).

%% file: proofs_exact.tex
\subsection{Proof of Converse for Theorem~\ref{Thm_exact_region}}
\label{Sec_proof_exact}
As discussed in Section~\ref{sec_results}, any $(M,R)$ pair that is achievable under no privacy requirement needs to satisfy the first and third inequalities in~\eqref{Eq_N2K2_region} for $N=K=2$~\cite{Maddah14}. Similarly, for $N>2$ and $K=2$, any feasible $(M, R)$ pair under no privacy constraint is required to satisfy the first and third inequalities in~\eqref{Eq_AnyNK2_region} ~\cite{Tian2018}. Substituting $N=2$ in the second inequality of \eqref{Eq_AnyNK2_region} gives us the second inequality of \eqref{Eq_N2K2_region}. Thus,  we only need to prove that the inequality \(3M+(N+1)R \geq 2N+1\) holds for \(N \geq 2\) and \(K=2\) and we give a proof for the same in this subsection.
To show that any feasible $(M,R)$ pair
satisfies this inequality, we use the following lemma on some conditional distributions. 
\begin{lemma}
\label{Lem_eqiuv_distrbn2}
Let $\tilde{k} = (k+1) \mod 2$ for $k=0,1$. Then for all \(i \in [0:N-1], i' \in [0:N-1]\) and \(j \in [0:N-1]\) any demand-private scheme for $K=2$ satisfies the following for user $k$, with $k \in \{0,1\}$ :
\begin{align}
( X, Z_k, W_{j}|D_k=j) & \sim ( X, Z_k, W_{j}|D_{\tilde{k}}=i, D_k=j) \notag\\
&  \sim  (X, Z_k, W_{j}|D_{\tilde{k}} =i', D_k=j). \label{Eq_priv_cond3}
\end{align}
\end{lemma}
\begin{proof}
	See Appendix~\ref{Sec_lemma_proof}.	
\end{proof}

Throughout this proof, for simplicity, we denote \((X|D_{0}=d_{0}, D_{1}=d_{1})\) by \( X_{d_{0},d_{1}}\). We also define \(W_{[0:N-1]}=(W_{0},W_{1},\ldots,W_{N-1})\), \( j_{i}=(j \oplus i) \) mod \(N\) and \(X'_{j}=(X_{j,j_{1}},X_{j,j_{2}},\ldots,X_{j,j_{N-1}})\). Then, we have
\allowdisplaybreaks
\begin{align}
& \sum_{j=0}^{N-1}(NH(Z_{0}) + H(Z_{1}) + 2H(X_{j,j}) + \sum_{i\neq j}H(X_{j,i})) \nonumber\\
& \geq \sum_{j=0}^{N-1}( H(Z_{0},X_{j,j})+H(Z_{1},X_{j,j})+\sum_{i\neq j}H(X_{j,i},Z_{0}))\nonumber\\
& \overset{\mathrm{(a)}}{=} \sum_{j=0}^{N-1}(H(Z_{0},X_{j,j},W_{j})+H(Z_{1},X_{j,j},W_{j})+\sum_{i\neq j}H(X_{j,i},Z_{0},W_{j}))\nonumber\\
& = \sum_{j=0}^{N-1}(H(Z_{1},X_{j,j},W_{j}) + H(X_{j,j}|Z_{0},W_{j})+\sum_{i\neq j}H(X_{j,i}|Z_{0},W_{j})+NH(Z_{0},W_{j}))\nonumber\\
& \geq \sum_{j=0}^{N-1}(H(Z_{1},X_{j,j},W_{j}) + H(X'_{j},X_{j,j},Z_{0},W_{j})+(N-1)H(Z_{0},W_{j}))\nonumber\\
& = \sum_{j=0}^{N-1}(H(Z_{1}|X_{j,j},W_{j}) + H(X'_{j},Z_{0}|X_{j,j},W_{j})+(N-1)H(Z_{0},W_{j})+2H(X_{j,j},W_{j}))\nonumber\\
& \geq \sum_{j=0}^{N-1}(H(X'_{j},Z_{0},Z_{1},X_{j,j},W_{j})+(N-1)H(Z_{0},W_{j})+H(X_{j,j},W_{j}))\nonumber\\
& \overset{\mathrm{(b)}}{=} \sum_{j=0}^{N-1}(H(X'_{j},Z_{0},Z_{1},X_{j,j},W_{[0:N-1]})+(N-1)H(Z_{0},W_{j})+H(X_{j_{1},j},W_{j}))\nonumber\\
& \geq \sum_{j=0}^{N-1}(H(W_{[0:N-1]})+(N-2)H(Z_{0},W_{j})+H(Z_{0}|W_{j})+H(X_{j_{1},j}|W_{j})+2H(W_{j}))\nonumber\\
& \geq \sum_{j=0}^{N-1}(H(W_{[0:N-1]})+(N-2)H(Z_{0},W_{j})+H(Z_{0},X_{j_{1},j},W_{j})+H(W_{j}))\nonumber\\
& \overset{\mathrm{(c)}}{\geq} \sum_{j=0}^{N-1}(H(W_{[0:N-1]})+(N-2)H(Z_{0},W_{j})+H(Z_{0},W_{j},W_{j_{1}})+H(W_{j}))\nonumber\\
& = N(N+1)F + \sum_{j=0}^{N-1}(N-2)H(Z_{0},W_{j})+\sum_{j=0}^{N-1}H(Z_{0},W_{j},W_{j_{1}})\nonumber\\
& \overset{\mathrm{(d)}}{=} N(N+1)F + \sum_{j=0}^{N-1}\sum_{i=2}^{N-1}H(Z_{0},W_{j_{i}})+\sum_{j=0}^{N-1}H(Z_{0},W_{j},W_{j_{1}})\nonumber\\
& \geq N(N+1)F +
\sum_{j=0}^{N-1}(\sum_{i=2}^{N-1}H(W_{j_{i}}|Z_{0})+H(W_{j},W_{j_{1}}|Z_{0})
+(N-1)H(Z_{0})))\nonumber\\
& \geq N(N+1)F + \sum_{j=0}^{N-1}(H(W_{[0:N-1]})
+(N-2)H(Z_{0}))\nonumber\\
& = N(2N+1)F + N(N-2)H(Z_{0})\label{eq:e}
\end{align}
where (a) and (c) follow from the decodability criteria; (b) follows from decodability criteria and Lemma~\ref{Lem_eqiuv_distrbn2}; (d) follows by rearranging the terms of first summation and the definition \(j_i=(j \oplus i)\) mod \(N\). Cancelling out the common terms of \(H(Z_{0})\) from both sides in \eqref{eq:e}, we obtain
\begin{align}
&2NH(Z_{0}) + NH(Z_{1}) + \sum_{j=0}^{N-1}( 2H(X_{j,j}) + \sum_{i\neq j}H(X_{j,i}))  \geq N(2N+1)F \nonumber 
\end{align}
which implies that
\begin{align}
&3M+(N+1)R \geq 2N+1. \nonumber
\end{align}
This follows since $H(Z_i) \leq MF$ and $H(X_{j,i}) \leq RF$ for $i \in [0:N-1]$, $j \in [0:N-1]$ by definition. The proof is thus complete.

\subsection{Proof of Achievability for Theorem~\ref{Thm_exact_region}}
\label{sec_exact_achv}

\subsubsection{Achievability of the region in~\eqref{Eq_N2K2_region} for $N=K=2$}

We show that any memory-rate pair $(M,R)$ satisfying the below inequalities is achievable under demand privacy for $N=K=2$:
\begin{align*}
2M + R \geq 2, \quad 3M+3R \geq 5, \quad M+2R \geq 2.
\end{align*}
To this end, we use the concept of \emph{demand type} introduced in~\cite{Tian2018}.
\begin{definition}[Demand Types]
	In $(N,K)$-non-private coded caching problem, for a given demand vector
	$\bar{d}$, let $t_i$ denote the number of users requesting file $i$, where
	$i=0, \ldots, N-1$. Demand type of $\bar{d}$, denoted by $T(\bar{d})$, is defined
	as the $N$-length vector 
	$T(\bar{d}):=\bar{t} = (t_1, \ldots, t_N)$. The type class of
	$\bar{t}$ is defined as $\cD_{\bar{t}}=\{\bd|T(\bd)=\bar{t}\}$. 
\end{definition}

Clearly, the restricted demand subset $\DRS$ defined in Definition~\ref{Def_dmnd_subst} is a subset of the type class
$(K,K,\ldots,K)$, i.e., 
\begin{eqnarray}
&&\DRS \subseteq \cD_{(K,K,\ldots,K)}.\label{eq:drstype}
\end{eqnarray}
Indeed, for $\cD_1\subset \cD_2$, a
$\cD_2$-non-private scheme is also a $\cD_1$-non-private scheme. Thus, we have the following proposition.

\begin{proposition}
	\label{Cor_Type}
	If there exists an $(N,NK,M,R)$ $\cD_{(K,K,\ldots,K)}$-non-private scheme, then
	there exists an $(N,K,M,R)$-private scheme.
\end{proposition}
\begin{proof}
	As mentioned before, we have $\DRS \subseteq \cD_{(K,K,\ldots,K)}$. So, an $(N,NK,M,R)$ $\cD_{(K,K,\ldots,K)}$-non-private scheme is also an $(N,NK,M,R)$ $\DRS$-non-private scheme. Then, the proposition follows from Theorem~\ref{Thm_genach}.
\end{proof}

It was shown in~\cite[Proposition~7]{Tian2018} that
the region given by~\eqref{Eq_N2K2_region} is an achievable  for Type $(2,2)$ for the
$N=2,K=4$  coded caching problem without demand privacy. So it follows from Proposition~\ref{Cor_Type} that the same region is achievable under demand privacy for $N=K=2$. 

\begin{remark}
	The corner points of the region in~\eqref{Eq_N2K2_region} are $(0,2)$ $(\frac{1}{3},\frac{4}{3})$, $(\frac{4}{3}, \frac{1}{3})$ and $(2,0)$. The achievability of the pairs $(0,2)$ and $(2,0)$ for Type $(2,2)$ in $N=2,K=4$ non-private  coded caching
	problem is trivial. The achievability of the pairs 
$(\frac{1}{3},\frac{4}{3})$ and $(\frac{4}{3}, \frac{1}{3})$ were shown in~\cite{Tian2018}.  In Example~\ref{Ex_cach_random}, we showed that using the non-private scheme that achieves the memory-rate pair
$(\frac{1}{3},\frac{4}{3})$, we can achieve the same pair with demand privacy for $N=K=2$. Similarly, we can also achieve the pair $(\frac{4}{3}, \frac{1}{3})$. We further note that the  pair $(\frac{4}{3}, \frac{1}{3})$ is also achievable by the MDS scheme  in~\cite{Wan19}.
\end{remark}

\subsubsection{Achievability of the region in~\eqref{Eq_AnyNK2_region} for $N > K=2$}
We show that any rate-memory pair satisfying the below inequalities is achievable under demand privacy for $N > 2$ and $K=2$:
\begin{align*}
3M+NR \geq 2N, \quad 3M+(N+1)R \geq 2N+1, \quad M+NR \geq N. \label{Eq_AnyNK2_region}
\end{align*}
The corner points of this rate-memory curve are \((0,2), (\frac{N}{3},1), (\frac{N^{2}}{2N-1}, \frac{N-1}{2N-1})\) and \((N,0)\). The achievability of \((0,2) \) and \((N,0)\) was shown in Theorem~\ref{th:basic} while that of \((\frac{N}{3},1) \) and \((\frac{N^{2}}{2N-1}, \frac{N-1}{2N-1})\) is proved next. The achievability of the entire region then follows by memory sharing. Throughout this subsection, for simplicity we define $\tilde{k} = (k + 1)$ mod 2, where $k \in \{0,1\}$.

\underline{\textbf{Achievability of $(\frac{N}{3},1)$:}}
Now we describe Scheme D for $N>2$ files and 2 users which generalizes the ideas discussed in Example~\ref{Ex_exact_first_point}. Scheme D achieves rate 1 for memory $\frac{N}{3}$. We first give an outline of Scheme D before describing it in detail.

 In Scheme D, the server partitions each file into three symbols of equal size.  The first symbols of all files are cached at user 0 and the second symbols are cached at user 1. So, each user has $N$ symbols of $F/3$ bits in her cache. The server randomly permutes all these $N$ symbols before caching at each user. Thus,  the users do not know the position of each symbol in their own cache.  In the delivery phase, the server reveals the position of the symbol of the demanded file that is available in her cache, through auxiliary transmission. Thus, she needs two more symbols to recover the entire file, which are obtained from the broadcast.
 The main payload of the broadcast consists of three symbols each of size $F/3$ bits. Each user uses two out of these three symbols to recover its demanded file in 
 the two cases of $D_0 = D_1$ and $D_0 \neq D_1$. Out of  the two symbols that each user uses to recover the file,
 one symbol is coded (XOR-ed with a symbol available in the cache) and the other one is uncoded in both the cases. The remaining symbol in the broadcast, which the user does not use for decoding,  appears as a sequence of random bits to the user. This symmetry helps in achieving the privacy. Next we formally describe Scheme D.

\underline{Caching:} The server breaks each file $W_{i}, i\in[0:N-1]$ into 3 disjoint parts of equal size, i.e.,
\(W_i=(W_{i,0},W_{i,1},W_{i,2})\). We define \(Z'_{0}\) and \(Z'_{1}\) as follows:
\begin{align*}
Z'_{0}& := (W_{i,0})_{i \in [0:N-1]} \\
Z'_{1} & := (W_{i,1})_{i \in [0:N-1]}.
\end{align*}
Let \(\pi_0 \) and \(\pi_1\) be two  permutation functions which are independent and uniformly distributed  in the symmetric group of permutations of $[0:N-1]$. Further, for $k \in \{0,1\}$, let
\begin{align*}
Z''_{k} & := (Z''_{k,0},Z''_{k,1}, \ldots,Z''_{k,N-1}) \\
& = \pi_{k}(Z'_{k}).
\end{align*}
The server places  $Z''_{k}$ in the cache of user $k\in \{0,1\}$ along with  4 symbols \((S_{k,1},S_{k,2},P_{k,1},P_{k,2})\)  of negligible size, where  
\begin{align*}
&S_{k,j}\sim unif\{[0:N-1]\}, \\
&P_{k,j} \sim unif\{[0:2]\}, \quad \text{for } k\in \{0,1\},\; j=1,2.
\end{align*}
 These 4 symbols are used  in the delivery phase. Thus,  the cache of user \(k\), \(Z_{k}\) is given by
\[
Z_{k} = (Z''_{k},S_{k,1},S_{k,2},P_{k,1},P_{k,2} ).
\]
Observe that $Z'_{k}$ consists of $N$ symbols each containing $\frac{F}{3}$ bits, which gives\footnote{ where $o(F)$ is some function of $F$ such that $\lim_{F \to \infty} \frac{o(F)}{F} = 0$. \label{note1}}
\begin{align*}
len(Z_k) &= \frac{NF}{3} + o(F)\\
& = MF + o(F).
\end{align*}
Note that in this caching scheme, the server does not fully reveal the permutation functions \(\pi_0 \) and \(\pi_1\) with any user.

\underline{Delivery:} To describe the delivery, we first define
$$
    X' := \left\{
       \begin{array}{lcl}
        (W_{D_{0},1} \oplus W_{D_{1},0}, W_{D_{0},2}, W_{D_{1},2}) & \text{if } D_{0}\neq D_{1}\\
        \\
         (W_{D_{0},1} \oplus W_{m,0}, W_{D_{0},2}, W_{D_{0},0} \oplus W_{m,1}) & \text{if } D_{0} = D_{1}
        \end{array}
        \right.
$$
where \(m = (D_{0}+1)\) mod \(N\). The server picks a permutation function \(\pi_{2}\) uniformly at random from the symmetric group of permutations of $[0:2]$ and includes \(\pi_{2}(X')\) in the transmission. The permutation \(\pi_{2}\) is not fully revealed to any of the users. In addition to \(\pi_{2}(X')\), to recover the demanded files, users need some more information, which can be delivered with negligible rate. The entire broadcast is given by 
\begin{align*}
X &= (\pi_{2}(X'),J_1,J_2,J_3)\\
& = (X'',J_1,J_2,J_3)\\
&  = (X''_{0},X''_{1},X''_{2},J_1,J_2,J_3). 
\end{align*}
Here, \((J_1,J_2,J_3)\) are auxiliary transmissions which contain the extra information. The auxiliary transmission $J_{1}$ is given by \(J_{1}= (J_{1,0},J_{1,1})\), where
 $J_{1,k}$ is defined as
\begin{align}
J_{1,k}  := S_{k,1} \oplus \pi_{k}(D_{k}). \label{eq:s1J1}
\end{align}
Recall that $\pi_{k}(D_k)$ gives the position of $W_{D_k,k}$ in $Z''_{k}$ while $S_{k,1}$, $k \in \{0,1\}$ is a part of $Z_{k}$.

To define auxiliary transmission \(J_2\), we first define random variables \(T_{k,j}\) for \(k \in \{0,1\}\), \(j \in \{1,2\}\) as follows:
\begin{align} \label{eq:impobsT}
(T_{0,1},T_{0,2},T_{1,1},T_{1,2}) := \left\{
\begin{array}{lcl}
 (\pi_{2}(0), \pi_{2}(1), \pi_{2}(0), \pi_{2}(2)) & \text{if } D_{0}\neq D_{1}\\
\\
(\pi_{2}(0), \pi_{2}(1), \pi_{2}(2), \pi_{2}(1)) & \text{if } D_{0}=D_{1}.
 \end{array}
 \right.
\end{align}
Note that $\pi_{2}(i)$ gives the position of the $i$-th symbol of $X'$ in $\pi_2(X')$. 
Auxiliary transmission $J_{2}$ is given by \(J_{2} = (J_{2,0},J_{2,1}) = (J_{2,0,0},J_{2,0,1},J_{2,1,0},J_{2,1,1})\), where $ J_{2,k,j}$, $j \in \{0,1\}$ 
is defined as
\begin{align}
J_{2,k,j} & := P_{k,j+1} \oplus T_{k,j+1}. \label{eq:s1J2}
\end{align}
Recall that symbols $P_{k,1}$ and $P_{k,2}$ are part of $Z_k$.

Auxiliary transmission $J_3$ is given by \(J_{3} = (J_{3,0},J_{3,1})\),  
where \(J_{3,k}\) is defined as
\begin{align}
J_{3,k} := S_{k,2} \oplus \pi_{k}(p_{k}) \label{eq:s1J3}
\end{align}
and $p_k$ is given as 
\begin{align} \label{eq:impobsT}
p_k := \left\{
\begin{array}{lcl}
 D_{\tilde{k}} & \text{if } D_{0}\neq D_{1}\\
\\
m & \text{if } D_{0}=D_{1}.
 \end{array}
 \right.
\end{align}
Recall that \(m = (D_{0}+1)\) mod \(N\), and $S_{0,2}$ and $S_{1,2}$ are part of $Z_0$ and $Z_1$, respectively.

Observe that $X'$ contains 3 symbols, each of size $\frac{F}{3}$ bits, which gives
\begin{align*}
len(X) &= \frac{3F}{3} + o(F) \\
&= RF + o(F).
\end{align*}

\underline{Decoding:} 
 We discuss the decoding of file $W_{D_k}$ for user $ k=0,1$. 
User $k$ first recovers \(\pi_{k}(D_{k}),T_{k,1},T_{k,2}\) and \(\pi_{k}(p_{k})\) from \(J_{1},J_{2}\) and \(J_{3}\) as follows:
\begin{align*}
& J_{1,k} \oplus S_{k,1} = \pi_{k}(D_{k}) \oplus S_{k,1} \oplus S_{k,1} = \pi_{k}(D_{k}) \qquad \text{(using \eqref{eq:s1J1})} \\
& J_{2,k,0}\oplus P_{k,1} = T_{k,1} \oplus P_{k,1} \oplus P_{k,1} = T_{k,1} \qquad \qquad \text{(using \eqref{eq:s1J2})} \\
& J_{2,k,1}\oplus P_{k,2} = T_{k,2}\oplus P_{k,2} \oplus P_{k,2} = T_{k,2} \qquad \qquad \text{(using \eqref{eq:s1J2})} \\
& J_{3,k} \oplus S_{k,2} = \pi_{k}(p_{k}) \oplus S_{k,2} \oplus S_{k,2} = \pi_{k}(p_{k}) \qquad \enspace \text{(using \eqref{eq:s1J3})}.
\end{align*}
User $k$ decodes the 3 parts of \(W_{D_{k}}\) as follows:
\begin{align*}
& \widehat{W}_{D_{k},k} = Z''_{k, \pi_{k}(D_{k})} \\
& \widehat{W}_{D_{k},2} = X''_{T_{k,2}} \\
& \widehat{W}_{D_{k},\tilde{k}} = X''_{T_{k,1}} \oplus Z''_{k,\pi_{k}(p_{k})}.
\end{align*}
User $k$ can recover each of $\widehat{W}_{D_{k},k}, \widehat{W}_{D_{k},2}$ and $\widehat{W}_{D_{k},\tilde{k}}$, where $\tilde{k} = (k+1)$ mod 2, because she has access to each symbol $X''_{i}$, $i \in [0:2]$ from the broadcast while all symbols $Z''_{k,j}$, $j \in [0:N-1]$ are available in her cache.

Observe that $\widehat{W}_{D_{k},k} = {W}_{D_{k},k}$ and $\widehat{W}_{D_{k},2} = {W}_{D_{k},2}$ by definition of $\pi_{k}(D_{k})$ and $T_{k,2}$, respectively. We show that  $\widehat{W}_{D_{k},\tilde{k}} = W_{D_k,\tilde{k}}$ by considering the following two cases: 

\noindent Case 1: $D_0 \neq D_1$
\begin{align*}
 \widehat{W}_{D_{k},\tilde{k}} & = X''_{T_{k,1}} \oplus Z''_{k,\pi_{k}(p_{k})} \\
& =  X''_{T_{k,1}} \oplus Z''_{k,\pi_{k}(D_{\tilde{k}})} \qquad \text{(using \eqref{eq:impobsT})} \\ 
& = W_{D_k,\tilde{k}} \oplus W_{D_{\tilde{k}},k} \oplus W_{D_{\tilde{k}},k}\\
& = W_{D_k,\tilde{k}}.
\end{align*}

\noindent Case 2: $D_0 = D_1$
\begin{align*}
\widehat{W}_{D_{k},\tilde{k}} & = X''_{T_{k,1}} \oplus Z''_{k,\pi_{k}(p_{k})}\\ 
& = X''_{T_{k,1}} \oplus Z''_{k,\pi_{k}(m)} \qquad \text{(using \eqref{eq:impobsT})} \\ 
& = W_{D_k,\tilde{k}} \oplus W_{m,k} \oplus W_{m,k}\\ 
& = W_{D_k,\tilde{k}}.
\end{align*}
Having retrieved these 3 segments of $W_{D_k}$, user $k$ recovers  \(W_{D_{k}}\) by concatenating \(W_{D_{k},0},W_{D_{k},1}\) and \(W_{D_{k},2}\) in that order.

\underline{Proof of privacy:} We now prove that Scheme D is demand-private for user $k \in \{0,1\}$, i.e., \eqref{Eq_instant_priv} holds true for this scheme. Recall that $\tilde{k}$ is defined as $\tilde{k} = (k+1)$ mod 2. Then the following sequence of equalities holds true:
\begin{align}
    I(D_{\tilde{k}};X,Z_k,D_k)
    & = I(D_{\tilde{k}};X|Z_{k},D_{k}) + I(D_{\tilde{k}};Z_k|D_{k}) +  I(D_{k};D_{\tilde{k}}) \nonumber \\
    & \overset{\mathrm{(a)}}{=} I(X;D_{\tilde{k}}|Z_{k},D_{k}) \nonumber \\
    &=I(\pi_{2}(X'),J_1,J_2,J_3;D_{\tilde{k}}|\pi_{k}(Z'_{k}), S_{k,1},S_{k,2},P_{k,1},P_{k,2},D_{k}) \nonumber  \\
    &\overset{\mathrm{(b)}}{=}I(J_1,J_2,J_3;D_{\tilde{k}}| S_{k,1},S_{k,2},P_{k,1},P_{k,2},D_{k})\nonumber \\
    &=I(J_{1,0},J_{1,1},J_{2,0},J_{2,1},J_{3,0},J_{3,1};D_{\tilde{k}}| S_{k,1},S_{k,2},P_{k,1},P_{k,2},D_{k}) \nonumber  \\
    & \overset{\mathrm{(c)}}{=} I(J_{1,k},J_{2,k},J_{3,k};D_{\tilde{k}}| S_{k,1},S_{k,2},P_{k,1},P_{k,2},D_{k})\nonumber \\
    & = I((S_{k,1} \oplus \pi_{k}(D_k)),(P_{k,1} \oplus T_{k,1}, P_{k,2} \oplus T_{k,2} ), (S_{k,2} \oplus \pi_{k}(p_k));D_{\tilde{k}}| S_{k,1},S_{k,2},P_{k,1},P_{k,2},D_{k}) \nonumber \\
    &=I(\pi_{k}(D_k) , T_{k,1}, T_{k,2}, \pi_{k}(p_k);D_{\tilde{k}}|S_{k,1},S_{k,2},P_{k,1},P_{k,2},D_{k}) \nonumber \\ 
    &\overset{\mathrm{(d)}}{=}I(\pi_{k}(D_k) , T_{k,1}, T_{k,2}, \pi_{k}(p_k);D_{\tilde{k}}|D_{k}) 
     \label{eq:pause1}
\end{align}
where (a) follows because $Z_k$ is independent of $(D_0,D_1)$ and also $D_0$ and $D_1$ are independent; for (b) note that for any fixed value of \((J_1,J_2,J_3,D_k,D_{\tilde{k}},S_{k,1},S_{k,2},P_{k,1},P_{k,2})\), we have
\begin{align*}
 (X', Z'_k|J_1,J_2,J_3,D_k,D_{\tilde{k}},S_{k,1},S_{k,2},P_{k,1},P_{k,2}) \sim  unif\{0,1\}^{\left(1+\frac{N}{3}\right)F}
\end{align*}
 which holds because for both cases $D_0 \neq D_1$ and $D_0 = D_1$, the symbols in $X'$ and $Z_k$ are independent. Hence, $X', Z_k $ and \((J_1,J_2,J_3,D_k,D_{\tilde{k}},S_{k,1},S_{k,2},P_{k,1},P_{k,2})\) are independent which gives (b); (c) follows because
$(S_{\tilde{k},1},S_{\tilde{k},2},$ $ P_{\tilde{k},1},P_{\tilde{k},2})$ which are one-time pads for symbols in $(J_{1,\tilde{k}},J_{2,\tilde{k}},J_{3,\tilde{k}})$ are independent of all other random variables; (d) follows because $(S_{k,1},S_{k,2},P_{k,1},P_{k,2})$ are independent of all other random variables.

Now we show that the RHS of~\eqref{eq:pause1} is 0.
From the definition in~\eqref{eq:impobsT}, we have that, for all fixed values of $D_0$, $D_1$, $\pi_0$, $\pi_1$, and $t_1, t_2 \in [0:2]$,
\begin{equation} \label{eq:impnote}
\Pr(T_{k,1}=t_{1},T_{k,2}=t_{2}|D_{0}, D_{1}, \pi_{0},\pi_{1}) = \Pr(T_{k,1}=t_{1},T_{k,2}=t_{2}) = \left\{
       \begin{array}{lr}
       \frac{1}{6}, \quad  & \text{if } t_{1}\neq t_{2}\\
       \\
       0, \quad  & \text{if } t_{1} = t_{2}.\\
        \end{array}
        \right.
\end{equation} 
Hence, $(T_{k,1},T_{k,2})$ is independent of $(D_{0}, D_{1}, \pi_{0},\pi_{1})$. Also, from definition we know that $(\pi_{k}(D_k),\pi_{k}(p_k))$ is a function of \((D_{0},D_{1},\pi_{k})\), which implies the independence of $(T_{k,1},T_{k,2})$ and $(\pi_{k}(D_k),\pi_{k}(p_k),D_0, D_1)$. Then, it follows from~\eqref{eq:pause1} that
\begin{align}
I(D_{\tilde{k}};X,Z_k,D_k) & =  I(\pi_{k}(D_k) , \pi_{k}(p_k);D_{\tilde{k}}|D_{k})\nonumber \\
    & = I(\pi_{k}(p_k);D_{\tilde{k}}|D_{k},\pi_{k}(D_k) +  I(\pi_{k}(D_k);D_{\tilde{k}}|D_{k}) \nonumber \\
    &\overset{\mathrm{(a)}}{=} 0 \nonumber
\end{align}
where (a) follows because for any fixed value of $(D_k,D_{\tilde{k}},\pi_{k}(D_k))$, we have
\begin{align*}
(\pi_{k}(p_k)|D_{k},D_{\tilde{k}},\pi_{k}(D_k)) \sim  (\pi_{k}(p_k)|D_{k},\pi_{k}(D_k)) \sim unif\{[0:N-1] \setminus \{\pi_{k}(D_k)\}\}
\end{align*}
and 
\begin{align*}
(\pi_{k}(D_k)|D_{k},D_{\tilde{k}}) \sim (\pi_{k}(D_k)|D_{k}) \sim unif\{[0:N-1]\}.
\end{align*}
This completes the proof of privacy.
\newline

\noindent \underline{\textbf{Achievability of \((\frac{N^{2}}{2N-1}, \frac{N-1}{2N-1})\):}}
 Now we describe Scheme E for $N>2$ files and 2 users which achieves rate $\frac{N-1}{2N-1}$ for memory $\frac{N^{2}}{2N-1}$.  File \(W_{i}\), $i \in [0:N-1]$ is partitioned  into \(2N-1\) disjoint parts of equal size, i.e., $W_i$ is given by
\(W_i =(W_{i,0},W_{i,1},\ldots,W_{i,2N-2})\). File \(W_{i}\) is then encoded using a \((3N-2,2N-1)\) MDS code such that each of \((3N-2)\) coded symbols has \(\frac{F}{2N-1}\) bits. Each file can be reconstructed using any \((2N-1)\) coded symbols. One of the \((3N-2)\) coded symbols of file \(W_{i}\) is denoted by  \(F_{i,0}\) while the remaining $(3N-3)$ symbols are denoted by \(F_{i,j,k}\), where \(j \in [0:2]\) and \(k \in [0:N-2]\). Next we give an outline of Scheme E.

In Scheme E, $N$ out of $3N-2$ symbols of each file are cached at each user. Out of these $N$ symbols, one symbol is common with the other user and remaining $N-1$ symbols are distinct from the other user. Similar to Scheme D, the server randomly  permutes these $N^2$ symbols before caching at each user. The main payload of the broadcast consists of $N-1$ symbols. To decode the demanded file, each user needs $2N-1$  symbols. The server reveals the positions of the $N$ symbols of the requested file that are available in the cache of each user, through auxiliary transmission. Both users obtain additional $N-1$ symbols from the broadcast in the two cases of $D_0=D_1$ and $D_0 \neq D_1$. This symmetry is a crucial point in preserving privacy. We formally describe Scheme E next.

\underline{Caching:} To give the cache contents of the users, we first  define tuples \(\cL_{0},\cL_{1}\) and \(\cL_{2}\) as follows:
$$   \cL_{j} = 
        (F_{i,0}, F_{i,j,1}, F_{i,j,2},\ldots,F_{i,j,N-1})_{i \in [0:N-1]}, \quad \forall j \in [0:2].
$$
The server randomly picks any 2 of  \(\cL_{0},\cL_{1}\) and \(\cL_{2}\)  and places one of them in the cache of user 0 after applying a random permutation, and places the other in the cache of user 1 after applying another independent permutation. To describe this process formally, we first define 
\begin{align*}
&U_0 \sim unif\{[0:2]\} \\
& U_1 \sim unif\{[0:2]\backslash \{ U_0\} \}.
\end{align*}
Let \(\pi_{0}\) and \(\pi_{1}\) be two independent and uniformly distributed permutation functions in the symmetric group of permutations of $[0:N^{2}-1]$. Further, for $k\in \{0,1\}$, let
\begin{align} 
Z'_{k} & = (Z'_{k,0},Z'_{k,1}, \ldots,Z'_{k,N^2-1}) = \pi_{k}(\cL_{U_k}). \label{eq:cache}
\end{align}
The server places $Z'_{k}$ and $U_k$ in the cache of user $k$ and also the symbols \((S_{k,1},S_{k,2},\ldots,S_{k,2N-1},P_{k})\), where 
\begin{align*}
 S_{k,j} & \sim unif\{[0:N^{2}-1]\}, \\
  P_{k} & \sim unif\{[0:2]\}, \qquad  \forall k \in \{0,1\}, \enspace j \in [0:2N-1]\backslash \{0\}. 
\end{align*}
These symbols are used in the delivery phase. Thus, the cache of user \(k\), \(Z_{k}\) is given by
\[
Z_{k} = (Z'_{k},S_{k,1},S_{k,2},\ldots,S_{k,2N-1},P_{k}, U_k ).
\]
Note that the main payload consists of $N$ coded symbols of each file and each symbol has $\frac{F}{2N-1}$ bits. Thus, we have
\begin{align*}
len(Z_k) &= \frac{N^2 F}{2N-1} + o(F) \\
&= MF + o(F).
\end{align*}

\underline{Delivery:} To describe the delivery phase, we first define
\begin{align} \label{eq:X'}
    X' = (X'_{0},X'_{1},\ldots,X'_{N-2}) = \left\{
       \begin{array}{lcl}
        (F_{D_{0},U_1,t} \oplus F_{D_{1},U_0,t})_{t \in [0:N-1]\backslash \{0\}} & \text{if } D_{0}\neq D_{1}\\
        \\
        (F_{D_{0},V,t} \oplus F_{m_{t},0})_{t \in [0:N-1]\backslash \{0\}} & \text{if } D_{0} = D_{1}
        \end{array}
        \right.
\end{align}
where \(m_{t} = (D_{0}+t)\) mod \(N\), and \(V = [0:2]\backslash\{U_0,U_1\}\).
The transmitted message $X$ is given by
\[
X=(X',J_1,J_2,J_3)
\]
where \((J_1,J_2,J_3)\) are the auxiliary transmissions.
Next we describe these auxiliary transmissions.

To describe $J_{1}$, we define
\begin{align}
C^{k}_{i,j} := \pi_{k}(Ni+j),  \quad \forall  i\in [0:N-1], \; j\in [0:N-1], \; k \in \{0,1\}. \label{eq:pos}
\end{align}
Thus, \(C^{k}_{i,0}\) and \(C^{k}_{i,j}\) respectively give the positions of \(F_{i,0}\) and \(F_{i,U_k,j}\) in \(Z'_{k}\).
The auxiliary transmission $J_1$ is given by \(J_1 = (J_{1,0},J_{1,1})\),  
where 
\begin{align}
J_{1,k} & = (J_{1,k,j})_{j \in [0:N-1]} := (S_{k,j+1} \oplus C^{k}_{D_{k},j})_{j \in [0:N-1]}, \quad k\in \{0,1\}. \label{eq:J1}
\end{align}
Here, $\oplus$ denotes addition modulo $N^2$ and also 
note that $S_{k,j+1}$ are part of $Z_k$.

Auxiliary transmission $J_{2}$ is defined as \(J_{2}=(J_{2,0},J_{2,1})\), where  
\begin{align}
J_{2,k} = (J_{2,k,j})_{j \in [0:N-2]} := (S_{k,N+1+j} \oplus H_{k,j+1})_{j \in [0:N-2]}, \quad k \in \{0,1\} \label{eq:J2}
\end{align}
with  \(H_{k,j} \in [0:N^2-1]\) defined by
\begin{align} \label{eq:defh}
H_{k,j} := \left\{
\begin{array}{lcl}
C^{k}_{D_{\tilde{k}},j} & \text{if } D_{0}\neq D_{1}\\
\\
C^{k}_{m_{j},0} & \text{if } D_{0} = D_{1}.
\end{array}
\right.
\end{align}
 Note that, for \(k \in \{0,1\}\) and \(j \in [0:N-1]\backslash \{0\}\), $S_{k,N+1+j} \in [0:N^{2}-1]$ are part of $Z_k$.

Finally, the auxiliary transmission $J_{3}$ is defined as \(J_{3}= (J_{3,0},J_{3,1})\), where
\begin{align}
J_{3,k} &  := (P_{k} \oplus T_{k}). \label{eq:J3}
\end{align}
Here, $P_k$ is a part of $Z_k$, and \((T_{0}, T_{1})\) is defined as
$$
    (T_{0}, T_{1}) := \left\{
       \begin{array}{lcl}
        (U_1, U_0) & \text{if } D_{0}\neq D_{1}\\
        \\
        (V, V) & \text{if } D_{0} = D_{1}.
        \end{array}
        \right.
$$
Observe that the main payload $X'$ consists of $(N-1)$ symbols of $\frac{F}{2N-1}$ bits each. Thus, we have
\begin{align*}
len(X) &= \frac{(N-1)F}{2N-1} + o(F) \\
&= RF + o(F).
\end{align*}

\underline{Decoding:} Now we describe the decoding of file $W_{D_k}$ at user $k \in \{0,1\}$. 
For $i \in [0:N-1], j \in [0:N-1] \backslash \{0\}$, user $k$  recovers \((C^{k}_{D_{k},0}, C^{k}_{D_{k},1},\ldots, C^{k}_{D_{k},N-1})\), \((H_{k,1}, H_{k,2},\ldots,H_{k,N-1})\) and \(T_{k}\) from \(J_{1},J_{2}\) and \(J_{3}\), respectively as follows:
\begin{align*}
& J_{1,k,i} \oplus S_{k,i+1} = C^{k}_{D_{k},i} \oplus S_{k,i+1} \oplus S_{k,i+1} = C^{k}_{D_{k},i} \qquad \quad \enspace \text{(using~\eqref{eq:J1})} \\
& J_{2,k,j-1} \oplus S_{k,N+j} = H_{k,j} \oplus S_{k,N+j} \oplus S_{k,N+j} = H_{k,j} \qquad \text{(using~\eqref{eq:J2})} \\
& J_{3,k} \oplus P_{k} = T_{k} \oplus P_{k} \oplus P_{k} = T_{k} \qquad \qquad \qquad \qquad \qquad \enspace \hspace{0.8mm} \text{(using~\eqref{eq:J3})}.
\end{align*} 
The coded symbols of \(W_{D_{k}}\), namely, \(F_{D_{k},0}\) and $F_{D_{k},U_k,j}, j \in [0:N-1]\setminus\{0\} $ are stored in the cache of user $k$, but their positions are unknown to the user. Using \((C^{k}_{D_{k},0}, C^{k}_{D_{k},1},\ldots,C^{k}_{D_{k},N-1})\), user $k$ can  recover these symbols as
\begin{align*}
\widehat{F}_{D_{k},0} &= Z'_{k,C^{k}_{D_{k},0}}, \\ \widehat{F}_{D_{k},U_k,j} & = Z'_{k,C^{k}_{D_{k},j}}, \quad \text{for } j \in [0:N-1] \backslash \{0\}.
\end{align*}

Observe that, by the definition of $C^{k}_{D_{k},0}$ and $ C^{k}_{D_{k},j}$, we get $\widehat{F}_{D_{k},0} = {F}_{D_{k},0}$ and $ \widehat{F}_{D_{k},U_k,j} = {F}_{D_{k},U_k,j}$. Now that user $k$ has recovered $N$ coded symbols of $W_{D_k}$, we show how it recovers $(N-1)$ more symbols namely, \(F_{D_{k},T_{k},j}\), $j \in [0:N-1]\setminus\{0\} $. 
Symbol \(F_{D_{k},T_{k},j}\) can be recovered from the main payload using \((H_{k,1}, H_{k,2},\ldots,H_{k,N-1})\) as follows:
$$
\widehat{F}_{D_{k},T_{k},j} = X'_{j-1} \oplus Z'_{k,H_{k,j}}, \quad \text{for } j \in [0:N-1] \backslash \{0\}.
$$
To show that $\widehat{F}_{D_{k},T_{k},j} =  {F}_{D_{k},T_{k},j}$, we consider the following two cases:
\newline
Case 1: $D_0 \neq D_1$
\begin{align*}
\widehat{F}_{D_{k},T_{k},j} &= X'_{j-1} \oplus Z'_{k,H_{k,j}}\\ 
& = F_{D_{0},U_1,j} \oplus F_{D_{1},U_0,j} \oplus Z'_{k,C^{k}_{D_{\tilde{k}},j}} \qquad \text{(using \eqref{eq:X'} and \eqref{eq:defh})}\\ 
&= F_{D_{0},U_1,j} \oplus F_{D_{1},U_0,j} \oplus F_{D_{\tilde{k}},U_k,j} \qquad   \text{(using \eqref{eq:pos} and \eqref{eq:cache})}\\ 
&= F_{D_{k},U_{\tilde{k}},j}\\ 
&= F_{D_{k},T_{k},j}.
\end{align*}
Case 2: $D_0 = D_1$
\begin{align*}
\widehat{F}_{D_{k},T_{k},j} &= X'_{j-1} \oplus Z'_{k,H_{k,j}} \\ &= F_{D_{0},V,j} \oplus F_{m_{j},0} \oplus Z'_{k,C^{k}_{m_{j},0}} \qquad \text{(using \eqref{eq:X'} and \eqref{eq:defh})} \\ 
&= F_{D_{0},V,j} \oplus F_{m_{j},0} \oplus F_{m_{j},0} \qquad \quad \text{(using \eqref{eq:pos} and \eqref{eq:cache})} \\ 
&= F_{D_{0},V,j} \\ &= F_{D_{k},T_{k},j}.
\end{align*}
Since \(T_k \neq U_k\), user $k$ has retrieved \(2N-1\) distinct symbols of the MDS code. Using these, user $k$ can decode file \(W_{D_{k}}\).

\underline{Proof of privacy: } Now we prove that  Scheme E is demand-private for user $k=0,1$, i.e., \eqref{Eq_instant_priv} holds true for this scheme.  Recall that $\tilde{k}$ is defined as $\tilde{k} = (k+1)$ mod 2.  Then the following sequence of equalities holds true:
\begin{align}
    &I(D_{\tilde{k}};X,Z_k,D_k) \nonumber \\
    & = I(D_{\tilde{k}};X|Z_{k},D_{k}) + I(D_{\tilde{k}};Z_k|D_{k}) +  I(D_{k};D_{\tilde{k}}) \nonumber \\
    & \overset{\mathrm{(a)}}{=} I(X;D_{\tilde{k}}|Z_{k},D_{k}) \nonumber \\
    &=I(X', J_{1}, J_{2}, J_{3};D_{\tilde{k}}|\pi_{k}(\cL_{U_k}), S_{k,1},S_{k,2},\ldots,S_{k,2N-1},P_{k},U_k,D_{k}) \nonumber \\
    &\overset{\mathrm{(b)}}{=}I(J_{1}, J_{2}, J_{3};D_{\tilde{k}}| S_{k,1},S_{k,2},\ldots,S_{k,2N-1},P_{k},U_k,D_{k})\nonumber \\
    & = I(J_{1,0},J_{1,1}, J_{2,0},J_{2,1}, J_{3,0},J_{3,1};D_{\tilde{k}}| S_{k,1},S_{k,2},\ldots,S_{k,2N-1},P_{k},U_k,D_{k})\nonumber \\
    & \overset{\mathrm{(c)}}{=} I(J_{1,k}, J_{2,k}, J_{3,k};D_{\tilde{k}}| S_{k,1},S_{k,2},\ldots,S_{k,2N-1},P_{k},U_k,D_{k})\nonumber \\
    &=I(C^{k}_{D_{k},0}, C^{k}_{D_{k},1},\ldots, C^{k}_{D_{k},N-1}, H_{k,1}, H_{k,2},\ldots,H_{k,N-1}, T_{k};D_{\tilde{k}}| S_{k,1},S_{k,2},\ldots,S_{k,2N-1},P_{k},U_k,D_{k})\nonumber \\
    &\overset{\mathrm{(d)}}{=}I(C^{k}_{D_{k},0}, C^{k}_{D_{k},1},\ldots, C^{k}_{D_{k},N-1}, H_{k,1}, H_{k,2},\ldots,H_{k,N-1}, T_{k};D_{\tilde{k}}|U_k,D_{k}) \label{eq:pause2}
\end{align}
where (a) follows because $Z_k$ is independent of $(D_0,D_1)$ and also $D_0$ and $D_1$ are independent; (b) follows because for any fixed value of \((J_1,J_2,J_3,S_{k,1},S_{k,2},\ldots,S_{k,2N-1},P_{k},U_k,D_{k},D_{\tilde{k}})\), we have
\begin{align*}
 (X',\cL_{U_k}|J_1,J_2,J_3,S_{k,1},S_{k,2},\ldots\,S_{k,2N-1},P_{k},U_k,D_{k},D_{\tilde{k}}) \sim unif\{0,1\}^{F\frac{(N^{2}+N-1)}{(2N-1)}};
\end{align*}
 (c) follows because $(J_{1,\tilde{k}}, J_{2,\tilde{k}}, J_{3,\tilde{k}})$ are encoded using one-time pads which are only available with user $\tilde{k}$; (d) follows because $(S_{k,1},S_{k,2},\ldots,S_{k,2N-1},P_{k})$ are one-time pads which are independent of all other random variables.

Next we show that the RHS of~\eqref{eq:pause2} is 0. To this end, we need the following.
For all \(C^{k}_{D_{k},i}\) and \(H_{k,j}\) distinct, observe that:
\begin{enumerate}[(i)]
	\item  For $d_0 \in [0:N-1], d_1 \in [0:N-1]$, \(d_0 \neq d_1\) and any fixed values of $(U_0,U_1)$,
	\begin{align}
	&\Pr(C^{k}_{D_{k},0}, C^{k}_{D_{k},1},\ldots, C^{k}_{D_{k},N-1}, H_{k,1}, H_{k,2},\ldots,H_{k,N-1}|D_{0}=d_0, D_{1}=d_1, U_0, U_1)\nonumber\\
	&\overset{\mathrm{(e)}}{=}\Pr(C^{k}_{d_k,0}, C^{k}_{d_k,1},\ldots, C^{k}_{d_k,N-1}, C^{k}_{d_{\tilde{k}},1}, C^{k}_{d_{\tilde{k}},2},\ldots,C^{k}_{d_{\tilde{k}},N-1}|D_{0}=d_{0}, D_{1}=d_{1}, U_0, U_1)\nonumber\\
	&\overset{\mathrm{(f)}}{=}\Pr(C^{k}_{d_k,0}, C^{k}_{d_k,1},\ldots, C^{k}_{d_k,N-1}, C^{k}_{d_{\tilde{k}},1}, C^{k}_{d_{\tilde{k}},2},\ldots,C^{k}_{d_{\tilde{k}},N-1})\nonumber\\
	&=\frac{(N^{2}-2N+1)!}{(N^{2})!}. \label{eq:obs3}
	\end{align}
	\item  For $d_0 \in [0:N-1], d_1 \in [0:N-1]$, \(d_0 = d_1\) and any fixed values of $(U_0,U_1)$,
	\begin{align}
	&\Pr(C^{k}_{D_{k},0}, C^{k}_{D_{k},1},\ldots, C^{k}_{D_{k},N-1}, H_{k,1}, H_{k,2},\ldots,H_{k,N-1}|D_{0}=d_0, D_{1}=d_1, U_0, U_1)\nonumber\\
	&\overset{\mathrm{(g)}}{=}\Pr(C^{k}_{d_{k},0}, C^{k}_{d_{k},1},\ldots, C^{k}_{d_{k},N-1}, C^{k}_{m_{1},0}, C^{k}_{m_{2},0},\ldots,C^{k}_{m_{N-1},0}|D_{0}=d_{0}, D_{1}=d_{1}, U_0, U_1)\nonumber\\
	&\overset{\mathrm{(h)}}{=}\Pr(C^{k}_{d_{k},0}, C^{k}_{d_{k},1},\ldots, C^{k}_{d_{k},N-1}, C^{k}_{m_{1},0}, C^{k}_{m_{2},0},\ldots,C^{k}_{m_{N-1},0})\nonumber\\
	&=\frac{(N^{2}-2N+1)!}{(N^{2})!}. \label{eq:obs4}
	\end{align}
\end{enumerate}
Here, (e) and (g) follow from \eqref{eq:defh}; (f) follows because $(C^{k}_{d_k,0}, C^{k}_{d_k,1},\ldots, C^{k}_{d_k,N-1}, C^{k}_{d_{\tilde{k}},1}, C^{k}_{d_{\tilde{k}},2},\ldots,C^{k}_{d_{\tilde{k}},N-1})$ only depends on $\pi_k$ which is independent of $(D_0,D_1,U_0,U_1)$; (h) follows for similar reasons as (f).
Note that by definition \(T_{k}\) is a function of \((D_{0},D_{1},U_0,U_1)\).

Now using~\eqref{eq:obs3} and~\eqref{eq:obs4}, we  conclude that \((C^{k}_{D_{k},0}, C^{k}_{D_{k},1},\ldots, C^{k}_{D_{k},N-1}, H_{k,1}, H_{k,2},\ldots,H_{k,N-1})\) is independent of
 \((D_{0}, D_{1},T_{k}, U_k )\). Thus, it follows from~\eqref{eq:pause2} that
\begin{align}
&I(C^{k}_{D_{k},0}, C^{k}_{D_{k},1},\ldots, C^{k}_{D_{k},N-1}, H_{k,1}, H_{k,2},\ldots,H_{k,N-1}, T_{k};D_{\tilde{k}}|U_k,D_{k}) = I(T_{k};D_{\tilde{k}}|U_k,D_{k}). \label{eq:pause3}
\end{align}
For $d_0 \in [0:N-1], d_1 \in [0:N-1], u_k \in [0:2], t \in [0:2]\setminus\{u_k\} $, 
\begin{align}
\Pr(T_{k}=t|D_{0}=d_0,D_{1} = d_1,U_k=u_k) =
\begin{cases}
\Pr(U_{\tilde{k}}=t|D_{0}=d_0,D_{1} = d_1,U_k=u_k) = \frac{1}{2}, \quad  &\text{if } d_0 \neq d_1 \\\\
 \Pr(V=t|D_{0}=d_0,D_{1} = d_1,U_k=u_k) = \frac{1}{2}, \quad &\text{if } d_0 = d_1.
\end{cases}
\label{eq:obs5}
\end{align}
From~\eqref{eq:obs5}, we obtain that $T_{k}$ is independent of $(U_k,D_{k},D_{\tilde{k}})$.  It thus follows from~\eqref{eq:pause3} and~\eqref{eq:pause2}  that
\begin{align}
&I(D_{\tilde{k}};X,Z_k,D_k) = I(T_{k};D_{\tilde{k}}|U_k,D_{k}) = 0.\nonumber
\end{align}
This completes the proof for privacy.

%% file: append.tex
\section{Proof of Lemma~\ref{Lem_subset_users}}
\label{Sec_proof_subset_lem}

To prove this lemma, we need to show that user $u'_i$ can recover all $Z^{i}_{j,\cS}$ such that $\cR^{-} \subset \cT, |\cR^{-}|=r-1$, and $0 \in \cS$. For $\cA := \cS \setminus \{0\}$, it follows from the definition that
\begin{align}
Z^{j}_{i,\cS} &= \bigoplus_{t \in \cV_{i} \backslash \cV_{i} \cap \cS} W_{j,\cS \cup \{t\}} \nonumber\\
&= \bigoplus_{t \in \cV_{i} \backslash \cV_{i} \cap \cR^{-}} W_{j,\cA \cup \{0, t\} }.\label{eq:pause10} 
\end{align}
For $t \in \cV_{i} \backslash \cV_{i} \cap \cS$, we can write
\begin{align}
Z^{j}_{i,\cA \cup \{t\}} &= \bigoplus_{u \in \cV_{i} \backslash \cV_{i} \cap (\cA \cup \{t\})} W_{j, \cA \cup \{t,u\}}\nonumber\\
&\overset{(a)}{=} W_{j,\cA \cup\{ 0, t\}} \oplus \bigoplus_{u \in \cV_{i} \backslash \cV_{i} \cap (\cS \cup \{t\})} W_{j, \cA \cup \{t,u\}} \nonumber
\end{align}
where, $(a)$ follows because $\cV_{i} \backslash \cV_{i} \cap (\cA \cup \{t\})   = \left(\cV_{i} \backslash \cV_{i} \cap (\cS \cup \{t\})\right) \cup \{0\}$. Thus, we have
\begin{align}
W_{j,\cA \cup \{ 0, t\}} &= Z^{j}_{i,\cA \cup \{t\}} \oplus \bigoplus_{u \in \cV_{i} \backslash \cV_{i} \cap (\cS \cup \{t\})} W_{j,  \cA \cup \{t,u\}}. \nonumber
\end{align}
Substituting the above expression of $W_{j,\cA \cup \{ 0, t\}}$ in~\eqref{eq:pause10}, we get
\begin{align}
Z^{j}_{i,\cS} &= \bigoplus_{t \in \cV_{i} \backslash \cV_{i} \cap \cS} Z^{j}_{i,\cA \cup \{t\}} \oplus \bigoplus_{t \in \cV_{i} \backslash \cV_{i} \cap \cS} \; \bigoplus_{u \in \cV_{i} \backslash \cV_{i} \cap (\cS \cup \{t\})} W_{j, \cA \cup \{t,u\}}. \nonumber
\end{align}
Observe that $Z^{j}_{i,\cA \cup \{t\}} $ is cached at the user while the second term is zero because every term $W_{j, \cA \cup \{t,u\}}$ appears twice in the double summation. This shows that  $Z^{j}_{i,\cS}$ can be recovered using only the cache contents. This completes the proof of Lemma~\ref{Lem_subset_users}.

\section{Proof of Lemma~\ref{Lem_bound_lin}}
\label{Sec_append}
To prove the lemma, we follow the proof of~\cite[Theorem~2]{Ghasemi17}, where a lower bound on the optimal rate that is
tighter than the cut-set bound was obtained. We also use that  $R^{\text{MAN,lin}}_{N,K}(M)$  is monotonically non-increasing for all $M \geq 0$ which can be shown as follows. Let $g_1(M) = N-M$, $g_2(M) = K(1-M/N)\frac{1}{1+\frac{KM}{N}}$ and  $t_0 = \frac{KM}{N} \in\{0,1, \ldots, K\}$. It is easy to see that, for $t'_0 ,t''_0 \in \{0,1, \ldots, K \}$,
\begin{align}
g_1\left(\frac{t'_0N}{K}\right) \leq g_1\left(\frac{t''_0N}{K}\right), \quad \mbox{if } t'_0 >t''_0,\; . \label{Eq_mont_g1}
\end{align}
We also have 
\begin{align}
g_2\left(\frac{t'_0N}{K}\right) \leq g_2\left(\frac{t''_0N}{K}\right) , \quad \mbox{if } t'_0 >t''_0 \label{Eq_mont_g2}
\end{align}
since 
\begin{align*}
g_2\left(\frac{t_0N}{K}\right)  - g_2\left(\frac{(t_0+1)N}{K}\right)  & = \frac{K-t_0 }{1+t_0} - \frac{K-(t_0+1)}{2+t_0}\\
& = \frac{K+1}{(1+t_0)(2+t_0)}\\
& \geq 0.
\end{align*}
From \eqref{Eq_mont_g1} and \eqref{Eq_mont_g2}, it follows that
\begin{align}
\min\left(g_1\left(\frac{t'_0N}{K}\right), g_2\left(\frac{t'_0N}{K}\right)\right) & \leq  \min\left(g_1\left(\frac{t''_0N}{K}\right), g_2\left(\frac{t''_0N}{K}\right)\right), \quad \mbox{for } t'_0 > t''_0. \label{Eq_mont_g}
\end{align}
Since $R^{\text{MAN,lin}}_{N,K}(M)$ is the linear interpolation of  $\min\left(g_1\left(\frac{t_0N}{K}\right), g_2\left(\frac{t_0N}{K}\right)\right)$ for $ t_0 \in\{0,1,\ldots,K\}$, \eqref{Eq_mont_g} implies that $R^{\text{MAN,lin}}_{N,K}(M)$ is monotonically non-increasing in $M$.

Now we consider the two memory regions studied in the proof of~\cite[Theorem~2]{Ghasemi17}. For $N \leq K$, the two regions are as follows: 

\underline{Region I: $0 \leq M \leq 1 $: }
Since $R^{\text{MAN,lin}}_{N,K}(0) = N$, and also that $R^{\text{MAN, lin}}_{N,K}(M) $ is monotonically non-increasing in $M$, we get
\begin{align*}
R^{\text{MAN, lin}}_{N,K}(M) \leq N.
\end{align*}
For this regime, it was shown~\cite[Theorem~2]{Ghasemi17} that 
$\Rm \geq N/4$. Then, we have 
\begin{align*}
\frac{R^{\text{MAN, lin}}_{N,K}(M) }{\Rm } \leq 4, \quad \mbox{ for } 0 \leq M  \leq 1.
\end{align*}

\underline{Region II: $1 \leq M \leq N/2 $: }
Let  us define $f_1(M): = \frac{N}{M} - \frac{1}{2}$. 
For $t_0 \geq 1$ and $\frac{Nt_0}{K} \leq M \leq \frac{N(t_0+1)}{K}$, 
it was shown~\cite[Theorem~2]{Ghasemi17}  that 
\begin{align}
R^{\text{MAN, lin}}_{N,K}\left(\frac{Nt_0}{K}\right) = \frac{K-t_0}{t_0+1} \leq f_1(M) \label{Eq_2regio1}
\end{align}
and also that
\begin{align}
\frac{f_1(M)}{\Rm} \leq 4. \label{Eq_2regio3}
\end{align}
Since
$R^{\text{MAN, lin}}_{N,K}(M)$ is non-increasing, we get
\begin{align}
R^{\text{MAN, lin}}_{N,K}(M) \leq  R^{\text{MAN, lin}}_{N,K}\left(\frac{Nt_0}{K}\right). \label{Eq_2regio2}
\end{align}
It thus follows from~\eqref{Eq_2regio1}, \eqref{Eq_2regio3} and \eqref{Eq_2regio2} that 
\begin{align}
\frac{R^{\text{MAN,lin}}_{N,K}(M)}{\Rm } \leq 4.
\end{align}
This completes the proof of Lemma~\ref{Lem_bound_lin}.

\section{Proof of Lemma~\ref{Lem_eqiuv_distrbn2}}
\label{Sec_lemma_proof}

We prove~\eqref{Eq_priv_cond3} for $k =1$. Other cases follow similarly.
Any $(N,K,M,R)$-private scheme satisfies that $I(D_0;Z_1,D_1,X) =0$. 
Since $H(W_{D_1}|X,Z_1,D_1) =0$,  we have that $I(D_0;Z_1,D_1,X, W_{D_1}) =0$. Then it follows that
\begin{align*}
& \Pr(D_0=i | X=x, Z_1 = z', W_j=w_j,D_1=j) =\Pr(D_0=i'| X=x, Z_1 = z', W_j=w_j,D_1=j).
\end{align*}
Multiplying both sides by $\Pr(X=x, Z_1 = z', W_j=w_j | D_1=j)$ gives
\begin{align*}
&\Pr(D_0=i, X=x, Z_1 = z',W_j=w_j | D_1=j) =\Pr(D_0=i', X=x, Z_1 = z', W_j=w_j| D_1=j).
\end{align*}
Then it follows that 
\begin{align*}
& \Pr(D_0=i| D_1=j) \times \Pr(X=x, Z_1 = z', W_j=w_j| D_0=i,D_1=j) \\
&= \Pr(D_0=i'| D_1=j)\times \Pr(X=x, Z_1 = z',W_j=w_j|  D_0=i', D_1=j).
\end{align*}
Since the demands are equally likely and they are independent of each other, we get
\begin{align}
& \Pr(X=x, Z_1 = z', W_j=w_j| D_0=i,D_1=j) \notag \\
&= \Pr(X=x, Z_1 = z',W_j=w_j,|  D_0=i', D_1=j). \label{eq:lemp1}
\end{align}
Further, we also have
\begin{align}
& \Pr(X=x, Z_1 = z', W_j=w_j| D_1=j) \notag \\ 
&= \sum_{t=0}^{N-1}\Pr(D_0=t) \times \Pr(X=x, Z_1 = z',W_j=w_j,|  D_0=t, D_1=j). \label{eq:lemp2}
\end{align}
Eq.~\eqref{eq:lemp1} and \eqref{eq:lemp2} together prove~\eqref{Eq_priv_cond3} for $k =1$.

%% file: main.bbl
\begin{thebibliography}{10}
\providecommand{\url}[1]{#1}
\csname url@samestyle\endcsname
\providecommand{\newblock}{\relax}
\providecommand{\bibinfo}[2]{#2}
\providecommand{\BIBentrySTDinterwordspacing}{\spaceskip=0pt\relax}
\providecommand{\BIBentryALTinterwordstretchfactor}{4}
\providecommand{\BIBentryALTinterwordspacing}{\spaceskip=\fontdimen2\font plus
\BIBentryALTinterwordstretchfactor\fontdimen3\font minus
  \fontdimen4\font\relax}
\providecommand{\BIBforeignlanguage}[2]{{%
\expandafter\ifx\csname l@#1\endcsname\relax
\typeout{** WARNING: IEEEtran.bst: No hyphenation pattern has been}%
\typeout{** loaded for the language `#1'. Using the pattern for}%
\typeout{** the default language instead.}%
\else
\language=\csname l@#1\endcsname
\fi
#2}}
\providecommand{\BIBdecl}{\relax}
\BIBdecl

\bibitem{Maddah14}
M.~A. Maddah-Ali and U.~Niesen, ``Fundamental limits of caching,'' \emph{IEEE
  Transactions on Information Theory}, vol.~60, no.~5, pp. 2856--2867, May
  2014.

\bibitem{maddah2014decentralized}
------, ``Decentralized coded caching attains order-optimal memory-rate
  tradeoff,'' \emph{IEEE/ACM Transactions On Networking}, vol.~23, no.~4, pp.
  1029--1040, 2014.

\bibitem{Amiri17}
M.~{Mohammadi Amiri} and D.~{Gunduz}, ``Fundamental limits of coded caching:
  Improved delivery rate-cache capacity tradeoff,'' \emph{IEEE Transactions on
  Communications}, vol.~65, no.~2, pp. 806--815, Feb 2017.

\bibitem{Vilardebo18}
J.~{G\'omez-Vilardeb\'o}, ``Fundamental limits of caching: Improved rate-memory
  tradeoff with coded prefetching,'' \emph{IEEE Transactions on
  Communications}, vol.~66, no.~10, pp. 4488--4497, Oct 2018.

\bibitem{Yu18}
Q.~Yu, M.~A. Maddah-Ali, and A.~S. Avestimehr, ``The exact rate-memory tradeoff
  for caching with uncoded prefetching,'' \emph{IEEE Transactions on
  Information Theory}, vol.~64, no.~2, pp. 1281--1296, Feb. 2018.

\bibitem{Ghasemi17}
H.~Ghasemi and A.~Ramamoorthy, ``Improved lower bounds for coded caching,''
  \emph{IEEE Transactions on Information Theory}, vol.~63, no.~7, pp.
  4388--4413, July 2017.

\bibitem{Wang18}
C.~{Wang}, S.~{Saeedi Bidokhti}, and M.~{Wigger}, ``Improved converses and gap
  results for coded caching,'' \emph{IEEE Transactions on Information Theory},
  vol.~64, no.~11, pp. 7051--7062, Nov 2018.

\bibitem{Yan17}
Q.~{Yan}, M.~{Cheng}, X.~{Tang}, and Q.~{Chen}, ``On the placement delivery
  array design for centralized coded caching scheme,'' \emph{IEEE Transactions
  on Information Theory}, vol.~63, no.~9, pp. 5821--5833, 2017.

\bibitem{Tang18}
L.~{Tang} and A.~{Ramamoorthy}, ``Coded caching schemes with reduced
  subpacketization from linear block codes,'' \emph{IEEE Transactions on
  Information Theory}, vol.~64, no.~4, pp. 3099--3120, 2018.

\bibitem{Suthan19}
H.~H. {Suthan Chittoor}, M.~{Bhavana}, and P.~{Krishnan}, ``Coded caching via
  projective geometry: A new low subpacketization scheme,'' in \emph{2019 IEEE
  International Symposium on Information Theory (ISIT)}, 2019, pp. 682--686.

\bibitem{Niesen17}
U.~{Niesen} and M.~A. {Maddah-Ali}, ``Coded caching with nonuniform demands,''
  \emph{IEEE Transactions on Information Theory}, vol.~63, no.~2, pp.
  1146--1158, 2017.

\bibitem{JiTLC17}
M.~{Ji}, A.~M. {Tulino}, J.~{Llorca}, and G.~{Caire}, ``Order-optimal rate of
  caching and coded multicasting with random demands,'' \emph{IEEE Transactions
  on Information Theory}, vol.~63, no.~6, pp. 3923--3949, 2017.

\bibitem{Zhang18}
J.~{Zhang}, X.~{Lin}, and X.~{Wang}, ``Coded caching under arbitrary popularity
  distributions,'' \emph{IEEE Transactions on Information Theory}, vol.~64,
  no.~1, pp. 349--366, 2018.

\bibitem{Ghasemi20}
H.~{Ghasemi} and A.~{Ramamoorthy}, ``Asynchronous coded caching with uncoded
  prefetching,'' \emph{IEEE/ACM Transactions on Networking}, vol.~28, no.~5,
  pp. 2146--2159, 2020.

\bibitem{Yang19}
Q.~{Yang}, M.~{Mohammadi Amiri}, and D.~{Gunduz},
  ``Audience-retention-rate-aware caching and coded video delivery with
  asynchronous demands,'' \emph{IEEE Transactions on Communications}, vol.~67,
  no.~10, pp. 7088--7102, 2019.

\bibitem{Shanmugam13}
K.~{Shanmugam}, N.~{Golrezaei}, A.~G. {Dimakis}, A.~F. {Molisch}, and
  G.~{Caire}, ``Femtocaching: Wireless content delivery through distributed
  caching helpers,'' \emph{IEEE Transactions on Information Theory}, vol.~59,
  no.~12, pp. 8402--8413, 2013.

\bibitem{Karamchandani16}
N.~{Karamchandani}, U.~{Niesen}, M.~A. {Maddah-Ali}, and S.~N. {Diggavi},
  ``Hierarchical coded caching,'' \emph{IEEE Transactions on Information
  Theory}, vol.~62, no.~6, pp. 3212--3229, 2016.

\bibitem{JiCM16}
M.~{Ji}, G.~{Caire}, and A.~F. {Molisch}, ``Fundamental limits of caching in
  wireless d2d networks,'' \emph{IEEE Transactions on Information Theory},
  vol.~62, no.~2, pp. 849--869, 2016.

\bibitem{maddah2016coding}
M.~A. Maddah-Ali and U.~Niesen, ``Coding for caching: Fundamental limits and
  practical challenges,'' \emph{IEEE Communications Magazine}, vol.~54, no.~8,
  pp. 23--29, 2016.

\bibitem{YossefBJK11}
Z.~Bar-Yossef, Y.~Birk, T.~S. Jayram, and T.~Kol, ``Index coding with side
  information,'' \emph{IEEE Transactions on Information Theory}, vol.~57,
  no.~3, pp. 1479--1494, March 2011.

\bibitem{NarayananRMDKP18}
V.~{Narayanan}, J.~{Ravi}, V.~K. {Mishra}, B.~K. {Dey}, N.~{Karamchandani}, and
  V.~M. {Prabhakaran}, ``Private index coding,'' in \emph{2018 IEEE
  International Symposium on Information Theory (ISIT)}, June 2018, pp.
  596--600.

\bibitem{DauSC12}
S.~H. Dau, V.~Skachek, and Y.~M. Chee, ``On the security of index coding with
  side information,'' \emph{IEEE Transactions on Information Theory}, vol.~58,
  no.~6, pp. 3975--3988, June 2012.

\bibitem{Karmoose20}
M.~{Karmoose}, L.~{Song}, M.~{Cardone}, and C.~{Fragouli}, ``Privacy in index
  coding: $k$ -limited-access schemes,'' \emph{IEEE Transactions on Information
  Theory}, vol.~66, no.~5, pp. 2625--2641, 2020.

\bibitem{SunJ17}
H.~{Sun} and S.~A. {Jafar}, ``The capacity of private information retrieval,''
  \emph{IEEE Transactions on Information Theory}, vol.~63, no.~7, pp.
  4075--4088, 2017.

\bibitem{Sengupta15}
A.~Sengupta, R.~Tandon, and T.~C. Clancy, ``Fundamental limits of caching with
  secure delivery,'' \emph{IEEE Transactions on Information Forensics and
  Security}, vol.~10, no.~2, pp. 355--370, Feb. 2015.

\bibitem{Ravindrakumar18}
V.~Ravindrakumar, P.~Panda, N.~Karamchandani, and V.~M. Prabhakaran, ``Private
  coded caching,'' \emph{IEEE Transactions on Information Forensics and
  Security}, vol.~13, no.~3, pp. 685--694, Mar. 2018.

\bibitem{Wan19}
K.~{Wan} and G.~{Caire}, ``On coded caching with private demands,'' \emph{IEEE
  Transactions on Information Theory}, vol.~67, no.~1, pp. 358--372, Jan. 2021.

\bibitem{Kamath19}
S.~Kamath, ``Demand private coded caching,'' {\tt arXiv:1909.03324 [cs.IT]},
  Sep. 2019.

\bibitem{AravindNCC20}
V.~R. {Aravind}, P.~K. {Sarvepalli}, and A.~{Thangaraj}, ``Subpacketization in
  coded caching with demand privacy,'' in \emph{2020 National Conference on
  Communications (NCC)}, Kharagpur, India, Feb. 2020.

\bibitem{Yan20}
Q.~Yan and D.~Tuninetti, ``Fundamental limits of caching for demand privacy
  against colluding users,'' {\tt arXiv:2008.03642 [cs.IT]}, Aug. 2020.

\bibitem{WanD2D19}
K.~Wan, H.~Sun, M.~Ji, D.~Tuninetti, and G.~Caire, ``Fundamental limits of
  device-to-device private caching with trusted server,'' {\tt arXiv:
  1912.09985 [cs:IT]}, Jan. 2020.

\bibitem{Aravind20}
V.~R. Aravind, P.~K. Sarvepalli, and A.~Thangaraj, ``Coded caching with demand
  privacy: Constructions for lower subpacketization and generalizations,'' {\tt
  arXiv: 2007.07475 [cs.IT]}, Jul. 2020.

\bibitem{KamathRD20}
S.~{Kamath}, J.~{Ravi}, and B.~K. {Dey}, ``Demand-private coded caching and the
  exact trade-off for {N=K}=2,'' in \emph{2020 National Conference on
  Communications (NCC)}, Kharagpur, India, Feb. 2020.

\bibitem{Tian2018}
C.~Tian, ``Symmetry, outer bounds, and code constructions: A computer-aided
  investigation on the fundamental limits of caching,'' \emph{Entropy},
  vol.~20, no.~8, pp. 603.1--603.43, Aug. 2018.

\bibitem{ShaoVZT19}
S.~{Shao}, J.~{G\'omez-Vilardeb\'o}, K.~{Zhang}, and C.~{Tian}, ``On the
  fundamental limit of coded caching systems with a single demand type,'' in
  \emph{2019 IEEE Information Theory Workshop (ITW)}, 2019, pp. 1--5.

\end{thebibliography}
